\documentclass[journal,final, twoside, 10pt]{IEEEtran}

\usepackage{amsthm}
\usepackage{amsfonts}
\usepackage{amssymb,mathrsfs}
\usepackage{graphicx}
\usepackage{tikz}
\usetikzlibrary{calc,decorations.pathreplacing}
\usepackage{xstring}
\usepackage{mathtools}
\usepackage[ruled,vlined,linesnumbered]{algorithm2e}
\usepackage{xcolor}
\usepackage{cite}
\usepackage{microtype}
\usepackage{xspace}
\usepackage{balance}
\usepackage{ifthen}

\usepackage{environ}

\makeatletter
\if@twocolumn
\newcommand{\whencolumns}[2]{
    #2
}
\else
\newcommand{\whencolumns}[2]{
    #1
}
\fi
\makeatother

\newcommand{\onlydouble}[1]{\whencolumns{}{#1}}
\newcommand{\onlysingle}[1]{\whencolumns{#1}{}}

\NewEnviron{multlinecc} 
{\whencolumns{\begin{equation}
            \BODY
\end{equation}}{\begin{multline}
        \BODY
\end{multline}}}

\NewEnviron{multlinecc*} 
{\whencolumns{\begin{equation*}
            \BODY
\end{equation*}}{\begin{multline*}
        \BODY
\end{multline*}}}

\NewEnviron{aligncc*} 
{\whencolumns{\begin{equation*}
            \BODY
\end{equation*}}{\begin{align*}
        \BODY
\end{align*}}}

\newcommand{\dbtilde}[1]{\tilde{\raisebox{0pt}[0.85\height]{$\tilde{#1}$}}}

\allowdisplaybreaks

\newif\ifproof

\newcommand{\Arikan}{Ar\i{}kan\xspace}

\newcommand{\calA}{\mathcal{A}}
\newcommand{\calE}{\mathcal{E}}
\newcommand{\calH}{\mathcal{H}}
\newcommand{\calI}{\mathcal{I}}

\newcommand{\calS}{\mathcal{S}}
\newcommand{\calT}{\mathcal{T}}
\newcommand{\calV}{\mathcal{V}}
\newcommand{\calX}{\mathcal{X}}
\newcommand{\calY}{\mathcal{Y}}

\newcommand{\vecb}{\mathbf{b}}

\newcommand{\vecF}{\mathbf{F}}
\newcommand{\vecG}{\mathbf{G}}
\newcommand{\vecu}{\mathbf{u}}
\newcommand{\vecmu}{\boldsymbol{\mu}}
\newcommand{\vecU}{\mathbf{U}}

\newcommand{\vecV}{\mathbf{V}}

\newcommand{\vecx}{\mathbf{x}}
\newcommand{\vecX}{\mathbf{X}}
\newcommand{\vecy}{\mathbf{y}}
\newcommand{\vecY}{\mathbf{Y}}
\newcommand{\vecz}{\mathbf{z}}
\newcommand{\vecxbr}[1]{\vecx^{[#1]}}

\newcommand{\vecxbrbr}[1]{\vecx^{[#1]}}

\newcommand{\vecZ}{\mathbf{Z}}

\newcommand{\bigvecX}{\mathsf{X}}
\newcommand{\bigvecY}{\mathsf{Y}}

\newcommand{\cupdot}{\mathbin{\mathaccent\cdot\cup}}

\newcommand{\Pe}{P_\mathrm{e}}

\newcommand{\obX}{\underline{X}} 
\newcommand{\obY}{\underline{Y}}
\newcommand{\Ybl}{M} 

\newcommand{\GI}{\vecG_{\mathrm{I}}}
\newcommand{\GII}{\vecG_{\mathrm{II}}}
\newcommand{\Gstar}{\vecG_{\triangle}}
\newcommand{\ZI}{\vecZ_{\mathrm{I}}}
\newcommand{\ZII}{\vecZ_{\mathrm{II}}}
\newcommand{\Zstar}{\vecZ_{\triangle}}
\newcommand{\YI}{\vecY_{\mathrm{I}}}
\newcommand{\YII}{\vecY_{\mathrm{II}}}

\newcommand{\Ystar}{\vecY_{\triangle}}
\newcommand{\XI}{\vecX_{\mathrm{I}}}
\newcommand{\XII}{\vecX_{\mathrm{II}}}
\newcommand{\xI}{\vecx_{\mathrm{I}}}
\newcommand{\xII}{\vecx_{\mathrm{II}}}

\newcommand{\blockCount}{\Phi}

\newcommand{\concat}{\odot}

\newcommand{\lbl}{\ell}
\newcommand{\edgeProb}{w}
\newcommand{\initstate}{q}
\newcommand{\finalstate}{r}

\newcommand{\pimin}{\pi_{\mathrm{min}}}

\newcommand{\WDPT}{W^*}
\newcommand{\DPT}{*}

\newcommand{\mysize}[1]{\left|#1\right|}
\newcommand{\myset}[1]{\left\{#1\right\}}

\newtheorem{theo}{Theorem}
\newtheorem{subclaim}{Sub-claim}
\newtheorem{lemm}[theo]{Lemma}

\newtheorem{defi}{Definition}

\SetKwFunction{recursivelyCalcTrellis}{recursivelyCalcTrellis}
\SetKwFunction{makeDecision}{makeDecision}

\newcommand{\xor}{\oplus}
\newcommand{\floor}[1]{\lfloor #1 \rfloor}
\newcommand{\ceiling}[1]{\lceil #1 \rceil}

\newcommand{\layerVar}{\lambda}

\newcommand{\eqann}[2][=]{\overset{\mathclap{(\text{#2})}}{#1}} 
\newcommand{\eqannref}[1]{$(\text{#1})$}

\SetAlFnt{\small}
\SetKwInOut{Require}{Require} 
\SetKwFunction{getArrayPointerP}{getArrayPointer\_P}
\SetKwFunction{getArrayPointerC}{getArrayPointer\_C}
\SetKwFunction{assignInitialPath}{assignInitialPath}
\SetKwFunction{pathIndexInactive}{pathIndexInactive}
\SetKwFunction{continuePathsFrozenBit}{continuePaths\_FrozenBit}
\SetKwFunction{continuePathsUnfrozenBit}{continuePaths\_UnfrozenBit}
\SetKwFunction{recursivelyCalcP}{recursivelyCalcP}
\SetKwFunction{recursivelyUpdateB}{recursivelyUpdateB}
\SetKwFunction{recursivelyUpdateC}{recursivelyUpdateC}
\SetKwFunction{killPath}{killPath}
\SetKwFunction{clonePath}{clonePath}
\SetKwFunction{findMostProbablePath}{findMostProbablePath}
\SetKwFunction{initializeDataStructures}{initializeDataStructures}
\SetKwData{forksArray}{forksArray}
\SetKwData{contForks}{contForks}
\SetKwData{probForks}{probForks}
\SetKwData{inactivePathIndices}{inactivePathIndices}
\SetKwData{inactiveArrayIndices}{inactiveArrayIndices}
\SetKwData{arrayReferenceCount}{arrayReferenceCount}
\SetKwData{pathIndexToArrayIndex}{pathIndexToArrayIndex}
\SetKwData{arrayPointerP}{arrayPointer\_P}
\SetKwData{arrayPointerC}{arrayPointer\_C}
\SetKwData{activePath}{activePath}
\SetKw{Continue}{continue}
\SetKw{New}{new}
\SetKw{Push}{push}
\SetKw{Pop}{pop}
\SetKw{False}{false}
\SetKw{True}{true}
\SetKw{myand}{and}
\SetKw{myor}{or}
\DontPrintSemicolon
\SetAlgoSkip{-5pt}

\newcommand{\xdownarrow}[1]{%
      {\left\downarrow\vbox to #1{}\right.\kern-\nulldelimiterspace}
  }

\newcommand{\xuparrow}[1]{%
      {\left\uparrow\vbox to #1{}\right.\kern-\nulldelimiterspace}
  }

\begin{document}
\title{Polar Codes for the Deletion Channel: \\ Weak and Strong Polarization}
\author{Ido Tal, Henry D. Pfister, Arman Fazeli, and Alexander Vardy
    \thanks{This paper was presented in part in the at the International Symposium on Information Theory (ISIT'2019).}
\thanks{The work of A.~Fazeli and A.~Vardy was supported in part by the National Science Foundation (NSF) under Grants CCF-1405119 and CCF-1719139.}
\thanks{The work of I.~Tal and A.~Vardy was supported in part by the United-States Israel Binational Science Foundation (BSF) under Grant No.~2018218.}
\thanks{The work of H.~D.~Pfister was supported in part by the National Science Foundation (NSF) under Grant No.~1718494.}
\thanks{I.~Tal is with the Department of Electrical Engineering, Technion, Haifa 32000, Israel (email: idotal@ee.technion.ac.il).}
\thanks{H.~Pfister is with Duke University in Durham, NC, USA (email: henry.pfister@duke.edu).}
\thanks{A.~Fazeli is with the University of California in San Diego, CA, USA (email: afazelic@ucsd.edu). }
\thanks{A.~Vardy is with the University of California in San Diego, CA, USA (email: vardy@ece.ucsd.edu).}
}

\maketitle

\begin{abstract}
    This paper presents the first proof of polarization for the deletion channel with a constant deletion rate and a regular hidden-Markov input distribution.
    A key part of this work involves representing the deletion channel using a trellis and describing the plus and minus polar-decoding operations on that trellis. In particular, the plus and minus operations can be seen as combining adjacent trellis stages to yield a new trellis with half as many stages.
    Using this viewpoint, we prove a weak polarization theorem for standard polar codes on the deletion channel.
    To achieve strong polarization, we modify this scheme by adding guard bands of repeated zeros between various parts of the codeword. This gives a scheme whose rate approaches the mutual information and whose probability of error decays exponentially in the cube-root of the block length.
    We conclude by showing that this scheme can achieve capacity on the deletion channel by proving that the capacity of the deletion channel can be achieved by a sequence of regular hidden-Markov input distributions.
\end{abstract}

\section{Introduction}

In many communications systems, symbol-timing errors may result in insertion and deletion errors.
For example, a \emph{deletion channel} with constant deletion rate maps a length-$N$ input string to a substring using an i.i.d.\ process that deletes each input symbol with probability $\delta$.
These types of channels were first studied in the 1960s~\cite{Gallager_1961,Dobrushin_1967} and modern coding techniques were first applied to them in~\cite{Davey_2001}.
Over the past 15 years, numerical bounds on the capacity of the deletion channel have been significantly improved but a closed-form expression for the capacity remains elusive~\cite{Mitzenmacher_2009,Fertonani_2010,Mercier_2012,Iyengar_2011,Iyengar_2015,Rahmati_2015,Castiglione_2015,Cheraghchi_2019}.
Recently, polar codes were applied to the deletion channel in a series of papers, but the question of polarization for non-vanishing deletion rates remained open~\cite{Thomas_2017,Tian_2017,Tian_2018,Tian_it2018}.
In this work, we show that polar codes can be used to efficiently approach the mutual information rate between a regular (i.e., finite-state, irreducible, and aperiodic) hidden-Markov input process and the output of the deletion channel with constant deletion rate.

\looseness=-1
In~\cite{Thomas_2017}, a polar code is designed for the binary erasure channel (BEC) and evaluated on a BEC that also introduces a single deletion.
An inner cyclic-redundancy check (CRC) code is used and decoding is performed by running the successive cancellation list (SCL) decoder \cite{Tal_2015} exhaustively over all compatible erasure locations.
The results show one can recover a single deletion in this setting.
Extensions to a finite number of deletions are also discussed but the decoding complexity grows faster than $N^{d+1}$, where $N$ is the code length and $d$ is the number of deletions.

In \cite{Tian_2017}, a low-complexity decoder is proposed for the same setup.
Its complexity, for a length-$N$ polar code, is roughly $d^3 N \log N$ when $d$ deletions occur\footnote{In~\cite{Tian_2017}, this complexity is misstated as $O(d^2 N \log N)$.}.
The paper also presents simulation results for polar codes with lengths ranging from 256 to 2048 on two deletion channels.
The first channel has a fixed deletion rate of 0.002 and the second introduces exactly $4$ deletions.
Based on their results, the authors of \cite{Tian_2017} conjecture that polarization occurs when $N \to \infty$ while the total number of deletions, $d$, is fixed.

The final papers~\cite{Tian_2018,Tian_it2018} in this series extend the previous results by proving that weak polarization occurs when $N\to \infty$ and $d=o(N)$.
While this result is quite interesting, its proof does not extend to the case of constant deletion rate.
For the case where $N\to \infty$ with $d$ fixed, these papers also show strong polarization for the deletion channel and weak polarization for the cascade of the deletion channel and a discrete memoryless channel (DMC).

In this paper, we combine the well-known trellis representation for channels with synchronization errors~\cite{Davey_2001} with low-complexity successive-cancellation (SC) trellis decoding for channels with memory~\cite{Wang_2014,Wang_2015}. 
In particular, \cite{Davey_2001} describes how the joint input-output probability of the deletion channel (and other synchronization-error channels) can be represented using a trellis.
This is closely related to fast algorithms for the edit distance between strings based on dynamic programming~\cite{Wagner_1974}.
The main advantage of the trellis perspective is that it naturally generalizes to other channels with synchronization errors (e.g., with insertions, deletions, and errors). The papers~\cite{Wang_2014,Wang_2015} describe how the plus and minus polar-decoding operations can be efficiently applied to a channel whose input-output mapping is represented by a trellis.
Putting these ideas together defines a low-complexity SC decoder for polar codes on the deletion channel that is essentially equivalent to the decoder defined in~\cite{Tian_2017}.

Building on previous proofs of polarization for channels with memory \cite{SasogluTal:18a,Shuval_Tal_Memory_2017}, this paper proves weak and strong polarization for the deletion channel.
In order to prove strong polarization, guard bands of `$0$' symbols are embedded in the codewords of \Arikan's standard polar codes. Effectively, these guard bands allow the decoder to work on independent blocks and enable our proof of strong polarization.

The primary results of this research are summarized in Theorem~\ref{theo:main}.
Conceptually, it provides a polynomial-time method to achieve the mutual information rate between a fixed regular hidden-Markov input process and the binary deletion channel.

\begin{theo}\label{theo:main}
    Fix a regular hidden-Markov input process and a parameter $\nu \in (0,1/3]$. The rate of our coding scheme approaches the mutual information rate between the input process and the binary deletion channel output. The encoding and decoding complexities of our scheme are $O(\Lambda \log \Lambda)$ and $O(\Lambda^{1+3\nu})$, respectively, where $\Lambda$ is the blocklength. For any $0 < \nu' < \nu$ and sufficiently large blocklength $\Lambda$, the probability of decoding error is at most $2^{-\Lambda^{\nu'}}$.
\end{theo}

The family of allowed input distributions is defined in Subsection~\ref{subsec:FAIM} and the structure of the codeword is defined in Section~\ref{sec:strong_setup}.
Its proof can be found in Section~\ref{sec:strong}.
While the theorem is stated for a fixed input process, we note that the encoding and decoding complexities scale cubically with the number of states in the input process.

Theorem~\ref{theo:cap} establishes a sequence of regular hidden-Markov input processes whose mutual information rates approach the deletion channel capacity.

\begin{theo}\label{theo:cap}
    Let $C$ be the capacity of the binary deletion channel with deletion probability $\delta$.
    For any $\epsilon>0$, there is a regular hidden-Markov input process whose mutual information rate on the binary deletion channel output is at least $C-\epsilon$.
\end{theo}
\vspace{1mm}

Together, the two theorems imply that the first scheme can be used to achieve capacity on the binary deletion channel.
We should note, however, that we do not provide an efficient method to optimize the input distribution or to bound its complexity in terms of the gap to capacity.
Also, Theorem~\ref{theo:cap} is weaker than a recent result by Li and Tan which proves the capacity can be approached by a sequence of finite-order Markov input distributions that are both irreducible and aperiodic~\cite{Li_2019}.
Both results are both predated by an earlier proof of Dobrushin that shows a sequence of periodic finite-state Markov input distributions can approach capacity on the deletion channel~\cite{Dobrushin_1967}.

Here is an outline of the structure of this paper.
Section~\ref{sec:background} sets up the basic notation and definitions used in this paper.
Section~\ref{sec:TrellisBasics} defines the concept of a trellis and shows how it can be used to compactly represent various deletion patterns and their corresponding probabilities.
In Section~\ref{sec:TrellisPolarization}, we describe how plus and minus polarization operations are applied to trellises to yield new trellises.
This provides a more detailed description of the SC trellis decoding method introduced in~\cite{Wang_2014}.
It is our hope that all sections up to and including  Section~\ref{sec:TrellisPolarization} will be accessible to practitioners who are primarily interested in the implementation details.
Section~\ref{sec:informationRates} discusses information rates and
Section~\ref{sec:weak} proves that, in our setting, weak polarization occurs.
Section~\ref{sec:strong} focuses on strong polarization.
The practitioner is advised to read Section~\ref{sec:strong_setup} which defines the structure and operation of an encoder with guard bands.
The proof of the main theorem is presented in Section~\ref{sec:strong}.

\section{Background}\label{sec:background}
\subsection{Notation}
The natural numbers are denoted by $\mathbb{N} \triangleq \{1,2,\ldots \}$. 
We also define $[m] \triangleq  \{ 1,2,\ldots,m \}$ for $m \in \mathbb{N}$.
Let $\calX$ denote a finite set (e.g., the input alphabet of a channel). 
In this paper, we fix $\calX = \{0,1\}$ as the binary alphabet. Extensions to non-binary alphabets are straightforward, see for example \cite[Chapter 3]{sasoglu_thesis} and \cite[Appendix A]{Shuval_Tal_Memory_2017}.
Let $\vecx = (x_1,\ldots,x_N) \in \calX^N$ be a vector of length $N = 2^n$. 
We use $[statement]$ to denote the Iverson bracket which evaluates to $1$ if $statement$ is true and $0$ otherwise.
The concatenation of vectors $\vecy \in \calX^{N_1}$ and $\vecy' \in \calX^{N_2}$ lives in $\calX^{N_1+N_2}$ and is denoted by $\vecy \concat \vecy'$.
The length of a vector $\vecy$ is denoted by $|\vecy|$. Random variables will typically be denoted by uppercase letters.

In this paper, we use the standard Ar\i{}kan  transform  presented in the seminal paper \cite{Arikan_2009}. The Ar\i{}kan  transform of $\vecx \in \calX^N$, $N = 2^n$, is defined recursively using length-$N/2$ binary vectors, $\vecx^{[0]}$ and $\vecx^{[1]}$:
\begin{IEEEeqnarray}{rCl}
    \vecxbr{0} & \triangleq & (x_1 \xor x_2,x_3 \xor x_4,\ldots,x_{N-1} \xor x_{N}) \; , \label{eq:vecxZero} \\
\vecxbr{1} & \triangleq &  (\phantom{x_1 \xor{}} x_2,\phantom{x_3\xor{}} x_4,\ldots,\phantom{x_{N-1}\xor{}} x_{N}) \label{eq:vecxOne} \; , 
\end{IEEEeqnarray}
where $\xor$ denotes modulo-2 addition.
Then, for any sequence $b_1,b_2,\ldots,b_\layerVar\in \{0,1\}$ with $\layerVar \leq n$, we extend this notation to define the vector $\vecxbr{b_1,b_2,\ldots,b_\layerVar} \in \calX^{2^{n-\lambda}}$ recursively via
\begin{equation}
    \label{eq:recursiveTransformDefinition}
    \vecxbr{b_1,b_2,\ldots,b_{\layerVar}} = \left( \vecxbr{b_1,b_2,\ldots,b_{\layerVar-1}} \right)^{[b_\layerVar]} \; .
\end{equation}
Specifically, if $\layerVar=n$, then the vector $\vecx^{[b_1,b_2,\ldots,b_\layerVar]}$ is a scalar. This scalar is denoted by $u_{i(\vecb)}$, where $\vecb$ defines the index
\begin{equation}
    \label{eq:bitReversedI}
    i(\vecb) \triangleq 1+\sum_{j=1}^{n} b_j 2^{n-j} \; .
\end{equation}
The transformed length-$N$ vector is given by
\begin{equation}
    \label{eq:vecuIsvecxTransformed}
    \vecu= (u_1,\ldots,u_N) = \calA_n(\vecx) \; , 
\end{equation}
where  $\calA_n \colon \calX^{2^n} \to \calX^{2^n}$ is called the Ar\i{}kan transform of order $n$.
Its inverse is denoted $\calA^{-1}_n$ and satisfies $\calA^{-1}_n=\calA_n$.

Let $\vecb = (b_1,b_2,\ldots,b_n)$ and $\vecx \in \calX^N$ be given, where $N = 2^n$ and $i = i(\vecb)$. As before, let $\vecu = \calA(\vecx)$. Since the vector $u^{i-1} = (u_1,u_2,\ldots,u_{i-1})$ will play an important role later on, we introduce additional notation. First, note that the vectors $\vecb' \in \{0,1\}^n$ can be totally ordered according to $i(\vecb')$ which is equivalent to standard lexicographic ordering. Recalling the notation $\vecx^{[\vecb]}$, we now define the related notation $\vecx^{(\vecb)} = \vecx^{(b_1,b_2,\ldots,b_{n})}$. Namely, $\vecx^{(\vecb)}$ is the concatenation of $\vecx^{[\vecb']}$, over all vectors $\vecb'\in \{0,1\}^n$ satisfying $i(\vecb') < i(\vecb)$. For $i(\vecb') = i(\vecb)-1$, this gives
\begin{equation}
    \label{eq:xRoundBracket}
    \vecx^{(\vecb)} \triangleq 
    \vecx^{[0,0,\ldots,0]} \concat \vecx^{[0,0,\ldots,0,1]} \concat \vecx^{[0,0,\ldots,0,1,0]} \concat \cdots \concat \vecx^{[\vecb']} .
\end{equation}
If $\vecb$ is the all-zero vector, then $\vecx^{(\vecb)}$ is the null vector. From these definitions it follows that $\vecx^{(\vecb)} = u^{i-1}$, where $i = i(\vecb)$ and $\vecu = \calA(\vecx)$.

\subsection{Deletion Channel}

Let $W(\vecy | \vecx)$ denote the transition probability of $N$ uses of the deletion channel with constant deletion rate $\delta$.
The input is denoted by $\vecx \in \calX^N$ and the output $\vecy$ has a random length $M = |\vecy|$ supported on $\{0,1,\ldots,N\}$.
This channel is equivalent to a BEC with erasure probability $\delta$ followed by a device that removes all erasures from the output.
Thus, $W(\vecy | \vecx)$ equals the probability that $N - M$ deletions have occurred, which is $(1-\delta)^M \cdot \delta^{N-M}$, multiplied by the number of distinct deletion patterns that produce $\vecy$ from $\vecx$, see \cite[Section 2]{Mitzenmacher_2009}.

We will also consider a trimmed deletion channel whose output is given by removing all leading and trailing zeros from the output of the standard deletion channel.
See Section~\ref{sec:strong} for details.

\subsection{Trellis Definition}

An $N$-segment \emph{trellis} $\calT$ is a labeled weighted directed graph $(\calV,\calE)$. 
We assume that $\calV$ can be partitioned into $\calV_0,\ldots,\calV_N$ so that $\calV$ is the union of $N+1$ disjoint sets:
\[
    \calV = \calV_0 \cupdot \calV_1 \cupdot \cdots \cupdot \calV_{N-1} \cupdot \calV_N \; ,
\]
where $\cupdot$ denotes a disjoint union.
For channels with memory, $\calV_j$ represents the set of possible channel states after $j$ channel inputs.
Similarly, the edge set $\calE$ is arranged into a sequence of $N$ disjoint sets:
\[
    \calE = \calE_1 \cupdot \calE_2 \cupdot \cdots \cupdot \calE_{N-1} \cupdot \calE_N \; .
\]
An edge in $\calE_j$ connects a vertex in $\calV_{j-1}$ to a vertex in $\calV_j$.
We define $\sigma(e)$ and $\tau(e)$ to be the starting and terminating vertices of edge $e$. Thus, for $e = u \to v$, we have $\sigma(e) = u$ and $\tau(e) = v$. Then, 
\[
    \mbox{$e \in \calE_j$ implies $\sigma(e) \in \calV_{j-1}$ and $\tau(e) \in \calV_j$} \; .
\]
A \emph{trellis section} comprises two adjacent sets of vertices along with the edges that connect them. That is, for $1 \leq j \leq N$, section $j$ comprises vertex sets $\calV_{j-1}$ and $\calV_j$, as well as edge set $\calE_j$. See Fig.~\ref{fig:deletion_trellis} for an example of a trellis with $4$ sections.

Each edge $e\in \calE$ has a weight $w(e)\in [0,1]$ and a label $\ell (e)\in\calX$. We also assume that $\calV_0$ and $\calV_N$ have weight functions,
\[
    q:\calV_0 \to [0,1] \quad \mbox{and} \quad r : \calV_N \to [0,1] \; ,
\]
that are associated with the initial and final states. 

A path through a trellis is a sequence of $N$ edges, $e_1,e_2,\ldots,e_N$, which starts at a vertex in $\calV_0$ and ends at a vertex in $\calV_N$. Namely, $\sigma(e_1) \in \calV_0$, $\tau(e_N) \in \calV_N$, and for each $1 \leq j \leq N-1$, we have $\tau(e_j) = \sigma(e_{j+1})$.
The weight of a path through the trellis is defined as the product of the weights on each edge in the path times the weights of the initial and final vertices. Namely, the weight of the above path is
\[
    q(\sigma(e_1)) \cdot r(\tau(e_N)) \times \prod_{j=1}^N w(e_j) \; .
\]
Thus, an $N$-section trellis naturally defines a \emph{path-sum} function $T\colon \calX^N \to \mathbb{R}$, where $T(\vecx)$ equals the sum of the path weights over all paths whose length-$N$ label sequences match $\vecx$. That is,
\begin{multlinecc}
    \label{eq:T}
    T(\vecx) \triangleq \sum_{\substack{e_1 \in \calE_1, \\ \ell(e_1) = x_1}} 
    \; \sum_{\substack{e_2 \in \calE_2, \\ \ell(e_2) = x_2}} \cdots \sum_{\substack{e_N \in \calE_N, \\ \ell(e_N) = x_N }} \initstate(\sigma(e_1)) \; \finalstate(\tau(e_N)) \\
                                                                    \times \prod_{j=1}^N w(e_j) \times \prod_{j=1}^{N-1} [\tau(e_j)=\sigma(e_{j+1})] \; .
\end{multlinecc}

\subsection{FAIM processes}
\label{subsec:FAIM}

In latter parts of this paper, for simplicity, we will often introduce key ideas by first framing them in the context of the uniform input distribution. That is, by first considering the case in which the input distribution is i.i.d.\ Bernoulli $1/2$. However, the uniform input distribution, or indeed any i.i.d.\ input distribution, is known to generally be sub-optimal with respect to the information rate between input and output, when transmitting over a deletion channel~\cite{Mitzenmacher_2009,Rahmati_2015,Castiglione_2015,Cheraghchi_2019}. Thus, we stand to benefit by considering a larger class of input distributions. 

To this end, let $\calS$ be a given finite set. Each element of $\calS$ is a state of an input process. In the following\footnote{The definition of FAIM and FAIM-derived processes here is a specialization of the definition given in \cite{Shuval_Tal_Memory_2017}. Here, we are interested in FAIM-derived (i.e., hidden-Markov) input processes.  However, the input-output process of a deletion channel is neither FAIM nor FAIM-derived.} definition,  we have for all $j \in \mathbb{Z}$ that $S_j \in \calS$ and $X_j \in \calX$.

\begin{defi}[FAIM process] \label{def:FAIM}
    A strictly stationary process $(S_j,X_j)$, $j \in \mathbb{Z}$ is called a  \emph{finite-state, aperiodic, irreducible, Markov} (FAIM) process if, for all $j$,
\begin{equation} \label{eq_markov property of FAIM}
    P_{S_j, X_j| S_{-\infty}^{j-1}, X_{-\infty}^{j-1}} = P_{S_j, X_j | S_{j-1}} \; ,
\end{equation} 
is independent of $j$ and the sequence $(S_j), j \in \mathbb{Z}$ is a finite-state Markov chain that is stationary, irreducible, and aperiodic. 
\end{defi}

For a FAIM process, consider the sequence $X_j$, for $j \in \mathbb{Z}$.
In principle, the distribution of this sequence can be computed by marginalizing the states of the FAIM process $(S_j,X_j)$.
Such a sequence is typically called a \emph{hidden-Markov process}.
In this paper, we sometimes add the term \emph{regular} to emphasize that the hidden state process is a regular Markov chain.

Let us now connect the concept of a FAIM process to that of a trellis. Let a FAIM process $(X_j,S_j)$ be given, and fix $N \geq 1$. We now define the corresponding trellis, having $N$ stages. The vertex set is $\calV = \calV_0 \cupdot \calV_1 \cupdot \cdots \cupdot \calV_N$, where we define
\[
    \calV_j = \{s_j : s \in \calS \}
\]
for $0 \leq j \leq N$ so that each $\calV_j$ contains a distinct copy of $\calS$. 
For each $x \in \calX$, $1 \leq j \leq N$, $\alpha_{j-1} \in \calV_{j-1}$, and $\beta_{j} \in \calV_j$, define an edge $e$ from $\alpha_{j-1}$ to $\beta_{j}$ with label $\ell(e) = x$ and weight $w(e) = P_{S_j,X_j|S_{j-1}}(\beta,x|\alpha)$. Lastly, for all $\alpha_{0} \in \calV_0$ define $q(\alpha_{0}) = \pi(\alpha)$, where $\pi(\alpha)$ is the stationary probability of state $\alpha$ in the Markov process $(S_j)_{j \in \mathbb{R}}$, and define $r(\beta_{N}) = 1$ for all $\beta_{N} \in \calV_N$. It follows that the probability of $(X_1,X_2,\ldots,X_N) = (x_1,x_2,\ldots,x_N) = \vecx$ equals $T(\vecx)$, where $T$ was defined in (\ref{eq:T}).

\section{Trellis representation of joint probability}\label{sec:TrellisBasics}
We have just seen that a trellis is instrumental in compactly representing a hidden-Markov input distribution. In fact, it is much more versatile than this. Namely, we will now show how a trellis can be used to represent the \emph{joint} distribution of a hidden-Markov input process and the channel output.

\subsection{Trellis for uniform input}

This trellis representation for the deletion channel can also be found in~\cite{Davey_2001}.

As previously explained, it is generally beneficial to use an input distribution with memory. However, for the sake of an easy exposition, we will first consider the simplest possible input distribution, a uniform input distribution (i.e., i.i.d.\ and Bernoulli $1/2$).

The trellis representation will be used on the decoder side. Thus, when building the trellis, we will have already received the output vector $\vecy$. Hence, the primary role of the trellis is to evaluate the probabilities associated with possible input vectors $\vecx$, of length $N$. That is, the trellis will be used to calculate the joint probability of $\vecx$ and $\vecy$, denoted $P_\vecX (\vecx) \cdot W(\vecy|\vecx)$, for $\vecy$ fixed. Recall that $W(\vecy|\vecx)$ is the deletion channel law, and in this subsection $P_\vecX$ is the uniform input distribution.

We will shortly define the concept of a valid path in the trellis. Each valid path will correspond to a specific transmitted $\vecx$ and a specific deletion pattern that is compatible with the received $\vecy$ (see Fig.~\ref{fig:deletion_trellis}). We term this trellis the \emph{base trellis}, as we will ultimately construct other trellises derived from it.

\begin{figure}
    \begin{center}
        \begin{tikzpicture}[scale=1.25,dot/.style={draw,circle,minimum size=1mm,inner sep=0pt,outer sep=0pt,fill=black}, >=latex]
        
          \pgfmathsetmacro{\xlim}{5}
          \pgfmathsetmacro{\ylim}{4}
          \pgfmathsetmacro{\xxlim}{\xlim-1}
          \pgfmathsetmacro{\yylim}{\ylim-1}
          \newcommand{\zz}{0}
          \newcommand{\recv}{011}
          
          \node[draw,circle,minimum size=2.5mm,inner sep=0pt,outer sep=0pt] (a) at (1,\ylim) {};
          \node[draw,circle,minimum size=2.5mm,inner sep=0pt,outer sep=0pt] (a) at (\xlim,1) {};
            
          \foreach \i in {0,...,\yylim}
            \node (i\i) at (0.65,\ylim-\i) {\scriptsize$i\!=\!\i$};
          \foreach \i in {1,...,\yylim}
            \node (y\i) at (0.25,\ylim-\i+0.5) {$y_\i \!=\! \StrChar{\recv}{\i}$};
        
          \foreach \j in {0,...,\xxlim}
            \node (j\j) at (\j+1,\ylim+0.30) {\scriptsize$j\!=\!\j$};    
          \foreach \j in {1,...,\xxlim}
            \node (x\j) at (\j+0.5,\ylim+.7) {$x_\j$};
        
          \node (cx) at (0.5,\ylim+0.7) {$x_j$};
          \node (cy) at (0,\ylim+0.25) {$y_i$};
          \draw (0,\ylim+0.7) -- (0.5, \ylim+0.25); 
        
          \foreach \i in {0,...,\yylim}
            \foreach \j in {0,...,\xxlim}
        	  \node [dot] (\j-\i) at (\j+1,\ylim-\i) {};
        
          
          \pgfmathsetmacro{\xxxlim}{\xlim-2}
          \pgfmathsetmacro{\yyylim}{\ylim-2}
          \foreach \i in {0,...,\yylim}
            \foreach \j in {0,...,\xxxlim}
        	{
              \pgfmathtruncatemacro{\ii}{\i+1}
              \pgfmathtruncatemacro{\jj}{\j+1}
        	  \ifnum \i > \j
        	    \draw[->,blue,dashed,thick] (\j-\i) to [in=150,out=30] (\jj-\i);
        	    \draw[->,red,dashed,thick] (\j-\i) to [in=210,out=330] (\jj-\i);        	    
        	  \else
        	    \ifnum \j > \i+\xlim-\ylim
        	      \draw[->,blue,dashed,thick] (\j-\i) to [in=150,out=30] (\jj-\i);
        	      \draw[->,red,dashed,thick] (\j-\i) to [in=210,out=330] (\jj-\i);
        	    \else        	      
        	      \draw[->,blue,thick] (\j-\i) to [in=150,out=30] (\jj-\i);
        	      \draw[->,red,thick] (\j-\i) to [in=210,out=330] (\jj-\i);
        	    \fi
        	  \fi
        	}
        	
          \pgfmathsetmacro{\xxxlim}{\xlim-2}
          \pgfmathsetmacro{\yyylim}{\ylim-2}
          \foreach \i in {0,...,\yyylim}
            \foreach \j in {0,...,\xxxlim}
        	{
              \pgfmathtruncatemacro{\ii}{\i+1}
              \pgfmathtruncatemacro{\jj}{\j+1}
              \pgfmathtruncatemacro{\iit}{\i+\xlim-\ylim}
        	  \ifnum \i > \j
                \StrChar{\recv}{\ii}[\temp]
        		\ifthenelse{\equal{\temp}{\zz}}
        		  {\draw[->,thick,dashed,blue] (\j-\i) -- (\jj-\ii);}
        		  {\draw[->,thick,dashed,red] (\j-\i) -- (\jj-\ii);}
        	  \else
        	    \ifnum \j > \iit
                  \StrChar{\recv}{\ii}[\temp]
        		  \ifthenelse{\equal{\temp}{\zz}}
        		    {\draw[->,thick,dashed,blue] (\j-\i) -- (\jj-\ii);}
        		    {\draw[->,thick,dashed,red] (\j-\i) -- (\jj-\ii);}
        	    \else
                  \StrChar{\recv}{\ii}[\temp]
        		  \ifthenelse{\equal{\temp}{\zz}}
        		    {\draw[->,thick,blue] (\j-\i) -- (\jj-\ii);}
        		    {\draw[->,thick,red] (\j-\i) -- (\jj-\ii);}
        		\fi
        	  \fi
        	}        	
        \end{tikzpicture}
\end{center}
\vspace{-2.5mm}
\caption{A trellis for the binary deletion channel with uniform input, a codeword length of $N=4$, and a received word $\vecy=(011)$ of length $M=3$. Vertices are denoted $v_{i,j}$ with $0 \leq i \leq M$ and $0 \leq j \leq N$. All blue edges have label `$0$' while all red edges have label `$1$'. The horizontal edges are weighted by the probability $\delta/2$. Diagonal edges are weighted by the probability $(1-\delta)/2$. The two circled vertices have $q(v_{0,0}) = r(v_{M,N}) = 1$, while all other vertices in $\calV_0$ and $\calV_N$ have $q$ and $r$ values equal to $0$, respectively.  Edges that can be pruned without changing the function $T$ in (\ref{eq:T}) are dashed. \label{fig:deletion_trellis}}
\end{figure}
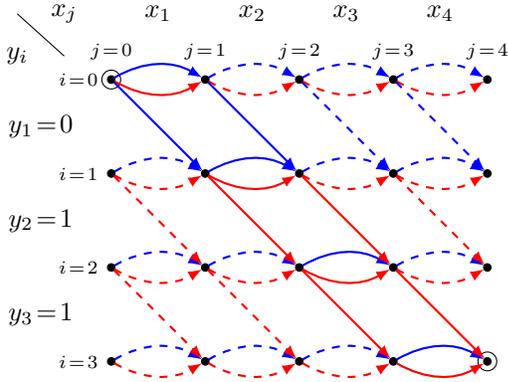

Recalling our notation, we have $\vecx$ as the unknown input vector, of known length $N$. The vector $\vecy$ is the known output, having known length $M = |\vecy|$. The deletion probability is $\delta$. The base trellis is defined as follows.

\begin{defi}[Base Trellis for Uniform Input]\label{defi:symmetricTrellis} For $N$, $\delta$, $M$, and $\vecy \in \calX^M$:
\begin{enumerate}
    \item The vertex set $\calV$ equals the disjoint union
        \[
            \calV = \calV_0 \cupdot \calV_1 \cupdot \cdots \cupdot \calV_N \; ,
        \]
        where, for $0 \leq j \leq N$,
        \begin{equation}
            \label{eq:VjSymmetric}
            \calV_j = \{ v_{i,j} : 0 \leq i \leq M\} \; . 
        \end{equation}

    \item \label{it:vijMeaning_uniform} A path passing through vertex $v_{i,j}$ corresponds to the event where only $i$ of the first $j$ transmitted symbols were received. That is, from $x_1,x_2,\ldots,x_j$, the channel has deleted $j-i$ symbols\footnote{Note that we could have optimized our definition of $\calV_j$. Namely, only $i$ in the range $\max\{0,M-N+j\} \leq i \leq \min\{j,M\}$ are actually consistent with the described event (i.e., only the solid edges in Figure~\ref{fig:deletion_trellis}).  We leave such optimization to the practitioner and settle for the simpler description in (\ref{eq:VjSymmetric}).}.

    \item Vertices $v_{i,j}$ with $0 \leq i \leq M$ and $0 \leq j < N$ each have up to three outgoing edges: two `horizontal' edges, each corresponding to a deletion, and one `diagonal' edge, corresponding to a non-deletion.

    \item For $0 \leq i \leq M$ and $0 \leq j <  N$, there are two edges $e,e'$ from $v_{i,j}$ to $v_{i,j+1}$. From \ref{it:vijMeaning_uniform}) above, we deduce that these two `horizontal' edges are associated with $x_{j+1}$ being deleted by the channel. The first is associated with $x_{j+1} = 0$ and has $\lbl(e) = 0$, while the second is associated with $x_{j+1} = 1$ and has $\lbl(e') = 1$. Since the probability of deletion is $\delta$, and in the uniform distribution $x_{j+1} = 0$ and $x_{j+1} = 1$ each occur with probability $1/2$, we set $\edgeProb(e) = \edgeProb(e') = \delta/2$.

    \item For $0 \leq i < M$ and $0 \leq j < N$, there is a single edge $e$ from $v_{i,j}$ to $v_{i+1,j+1}$.  Recalling \ref{it:vijMeaning_uniform}) above, we deduce that this `diagonal' edge represents $x_{j+1}$ not being deleted, and being observed as $y_{i+1}$.
        Thus, $\lbl(e) = y_{i+1}$. Since the probability of sending $x_{j+1}$ in the uniform case is $1/2$, regardless of its value, and the probability of a non-deletion is $1-\delta$, we set $\edgeProb(e) = (1-\delta)/2$.
    \item We set $q(v_{0,0}) = 1$. All other vertices $v \in \calV_0$ have $q(v) = 0$. Thus, with respect to (\ref{eq:T}), we effectively force all paths to start at $v_{0,0}$. Namely, when starting a path, no symbols have yet been transmitted, and hence no symbols have yet been received.
    \item We set $r(v_{M,N}) = 1$. All other vertices $v \in \calV_N$ have $r(v) = 0$. Thus, with respect to (\ref{eq:T}), we effectively force all paths to end at $v_{M,N}$. That is, at the end of a path, $N$ symbols have been transmitted, and of these, $M$ have been received. 

\end{enumerate}
\end{defi}

In line with the definitions above, let us call a path \emph{valid} if it starts at $v_{0,0}$ and ends at $v_{M,N}$. For example, in Figure~\ref{fig:deletion_trellis}, valid paths are those that start at the circled vertex on the top left, end at the circled vertex on the bottom right, and hence contain only solid edges. Clearly, such a path is comprised of $N$ edges, $e_1,e_2,\ldots,e_N$. Denote by $\vecx = (x_1,x_2,\ldots,x_N)$ the input vector corresponding to the above path, where $x_i = \ell(e_i)$. Each such $\vecx$ is consistent with our received $\vecy$. Indeed, tracing the path, the type of the corresponding edge (horizontal or diagonal) shows exactly which of the $x_i$ to delete and which to keep in order to arrive at $\vecy$.  Also, the probability of the input sequence $\vecx$ being transmitted and experiencing the above chain of deletion/no-deletion events is exactly equal to the product of the $w(e_i)$, times $q(v_{0,0}) \cdot r(v_{M,N}) = 1$.

From the above discussion, one has the following key lemma.

\begin{lemm}
    \label{lemm:baseTrellisProb}
    Let $\calT$ be a trellis as described in Definition~\ref{defi:symmetricTrellis}.  Then, for $\vecx \in \calX^N$ and $T(\vecx)$ as defined in (\ref{eq:T}), we have
    \[
        T(\vecx) = P_{\vecX}(\vecx) \cdot W(\vecy | \vecx) \; ,
    \]
    where $P_\vecX$ is the uniform input distribution and $W$ is the deletion channel law.
\end{lemm}
\begin{IEEEproof}
First, we observe that the weight of a trellis path equals the joint probability of $(\vecx,\vecy)$ and the deletion pattern.
Then, the claim follows from the fact that $T(\vecx)$ sums the path weight over all paths through the trellis (i.e., all deletion patterns) consistent with the given $(\vecx,\vecy)$ pair.
\end{IEEEproof}

\subsection{Trellises for hidden-Markov inputs}
As explained earlier, a trellis is used on the decoding side, in order to capture the joint probability of $\vecx$ and $\vecy$. We now show how such a trellis is built for the more general case in which $\vecx$ is drawn from a regular hidden-Markov input process. Intuitively, this is done by simply ``multiplying'' the trellis corresponding to the input distribution, as described at the end of Section~\ref{sec:background}, with the trellis defined for the uniform case (with the correction that the edge weights $\delta/2$ and $(1-\delta)/2$ are replaced by $\delta$ and $1-\delta$, respectively). A formal definition follows.

\begin{defi}[Base Trellis for Hidden-Markov Input]\label{defi:FAIMTrellis} For $N$, $\delta$, $M$, $\calS$, $P_{S_j,X_j|S_{j-1}}$, $\pi$, and $\vecy \in \calX^M$:
\begin{enumerate}
    \item The vertex set $\calV$ equals the disjoint union
        \[
            \calV = \calV_0 \cupdot \calV_1 \cupdot \cdots \cupdot \calV_N \; ,
        \]
        where, for $0 \leq j \leq N$,
        \begin{equation}
            \label{eq:sijFAIM}
            \calV_j = \{ s_{i,j} : 0 \leq i \leq M \; , \; s \in \calS \} \; . 
        \end{equation}
        Thus, $|\calV_j| = (M+1)\cdot |\calS|$.

    \item \label{it:vijMeaning} A path passes through vertex $s_{i,j}$ if exactly $i$ of the first $j$ transmitted symbols are not deleted and the state of the input process is $s\in \mathcal{S}$ after the $j$-th input (i.e., $S_j = s$).

    \item Vertices $s_{i,j}$ with $0 \leq i \leq M$, $0 \leq j < N$, and $s \in \calS$ each have up to $3 \cdot |\calS|$ outgoing edges.

    \item For $0 \leq i \leq M$, $0 \leq j <  N$, and $\alpha, \beta \in \calS$, there are two edges $e,e'$ from $\alpha_{i,j}$ to $\beta_{i,j+1}$. From item \ref{it:vijMeaning}, we deduce that these two `horizontal' edges are associated with $x_{j+1}$ being deleted by the channel. The first is associated with $x_{j+1} = 0$ and has $\lbl(e) = 0$, while the second is associated with $x_{j+1} = 1$ and has $\lbl(e') = 1$. Recalling that by stationarity $P_{S_{j+1},X_{j+1}|S_j} =  P_{S_j,X_j|S_{j-1}}$, we set
        \begin{equation} \label{eq:horizontalZero}
            \edgeProb(e) = \delta \cdot P_{S_j,X_j|S_{j-1}}(\beta,0|\alpha)
        \end{equation}
        and
        \begin{equation} \label{eq:horizontalOne}
            \edgeProb(e') = \delta \cdot P_{S_j,X_j|S_{j-1}}(\beta,1|\alpha) \; .
        \end{equation}
        That is, the probability of a deletion, times the probability implied by the underlying FAIM distribution.

    \item For $0 \leq i < M$, $0 \leq j < N$, and $\alpha,\beta \in \calS$, there is a single edge $e$ from $\alpha_{i,j}$ to $\alpha_{i+1,j+1}$.  Recalling item \ref{it:vijMeaning} above, we deduce that this `diagonal' edge represents $x_{j+1}$ being observed (i.e., not deleted) as $y_{i+1}$.
        Thus, $\lbl(e) = y_{i+1}$. We set
        \[
            \edgeProb(e) = (1-\delta) \cdot P_{S_j,X_j|S_{j-1}}(\beta,y_{i+1}|\alpha) \; .
        \]
        That is, the probability of a non-deletion, times the probability implied by the underlying FAIM distribution\footnote{As in the uniform case, we have opted for simplicity of exposition over reduced algorithmic complexity. That is, as in the uniform case, we can take the index $i$ in (\ref{eq:sijFAIM})  to have range $\max\{0,M-N+j\} \leq i \leq \min\{j,M\}$. Also, edges $e$ with probability $w(e)=0$ can be removed from the trellis.}.         
    \item For all $s_{0,0} \in \calV_0$, where $s \in \calS$, we set $q(s_{0,0}) = \pi(s)$. All other vertices $v \in \calV_0$ have $q(v)=0$. Thus, with respect to (\ref{eq:T}), we effectively force all paths to start at a vertex $s_{0,0}$, where $s \in \calS$. Namely, when starting a path, no symbols have yet been transmitted, and hence no symbols have yet been received. Moreover, the probability of starting the path at $s_{0,0}$ is $\pi(s)$, the stationary probability of $s$ in the FAIM input process.
    \item For all $s_{M,N} \in \calV_N$, we set $r(s_{M,N}) = 1$. All other vertices $v \in \calV_N$ have $r(v) = 0$. Thus, with respect to (\ref{eq:T}), we effectively force all paths to end at a vertex $s_{M,N}$. That is, at the end of a path, $N$ symbols have been transmitted, and of these, $M$ have been received. 
\end{enumerate}
\end{defi}

As in the uniform case, we have the following lemma, which is easily proved.
\begin{lemm}
    \label{lemm:baseFAIMTrellisProb}
    Let $\calT$ be a trellis as per Definition~\ref{defi:FAIMTrellis}.  Then, for $\vecx \in \calX^N$ and $T(\vecx)$ as defined in (\ref{eq:T}),
    \[
        T(\vecx) = P_{\vecX}(\vecx) \cdot W(\vecy | \vecx) \; ,
    \]
    where $P_\vecX$ is the hidden-Markov input distribution and $W$ is the deletion channel law.
\end{lemm}
\begin{IEEEproof}
First, we observe that the weight of a trellis path equals the joint probability of $(\vecx,\vecy)$ and the deletion pattern.
Then, the claim follows from the fact that $T(\vecx)$ sums the path weight over all paths through the trellis (i.e., all deletion patterns) consistent with the given $(\vecx,\vecy)$ pair.
\end{IEEEproof}

\subsection{Trellis for the trimmed deletion channel}
\label{subsec:TrellisTDC}
For reasons that will shortly become clear, we will now consider a slight variation of the deletion channel. Namely, we now define the trimmed deletion channel (TDC). A TDC
is a deletion channel that, after the deletion process, trims its output of leading and trailing `$0$' symbols. Thus, by definition, the output of a TDC is either an empty string, or a string that starts and ends with a `$1$' symbol.

We now show how to alter Definition~\ref{defi:FAIMTrellis} in order to account for this variation. The change turns out to be minimal.
\begin{defi}[Base Trellis for Hidden-Markov Input and TDC]\label{defi:FAIMTDCTrellis} For $N$, $\delta$, $M$, $\calS$, $P_{S_j,X_j|S_{j-1}}$, $\pi$, and trimmed output $\vecy^* \in \calX^M$, define the trellis $\calT$ as in Definition~\ref{defi:FAIMTrellis}, but with the following changes.
\begin{itemize}
    \item The probability of an edge $e$ from $\alpha_{0,j}$ to $\beta_{0,j+1}$ with $\ell(e) = 0$ must be changed to $\edgeProb(e) = P_{S_j,X_j|S_{j-1}}(\beta,0|\alpha)$. Namely, the $\delta$ factor in (\ref{eq:horizontalZero}) is removed. In short, if the path is currently at vertex $\alpha_{0,j}$, then none of the $j$ symbols $x_1,x_2,\ldots,x_j$ have made it to the output of the channel (they have either been deleted or trimmed). Thus, if $x_{j+1} = 0$, it will surely be either deleted, or else trimmed.
    \item The probability of an edge $e$ from $\alpha_{M,j}$ to $\beta_{M,j+1}$ with $\ell(e) = 0$ must be changed to $\edgeProb(e) = P_{S_j,X_j|S_{j-1}}(\beta,0|\alpha)$. Namely, the $\delta$ factor in (\ref{eq:horizontalZero}) is removed. Note that the exact same reasoning from the previous point applies; the only difference is that now we are correcting for the trimming of the trailing `$0$' symbols.
\end{itemize}
\end{defi}

The result of the above altered trellis definition is the following lemma.
\begin{lemm}
    \label{lemm:baseFAIMTDCTrellisProb}
    Let $\calT$ be a trellis as described in Definition~\ref{defi:FAIMTDCTrellis}.  Then, for $\vecx \in \calX^N$ and $T(\vecx)$ as defined in (\ref{eq:T}),
    \[
        T(\vecx) = P_{\vecX}(\vecx) \cdot \WDPT(\vecy^* | \vecx) \; ,
    \]
    where $P_\vecX$ is the hidden-Markov input distribution and $\WDPT$ is the law of the TDC.
\end{lemm}
\begin{IEEEproof}
First, we observe that the weight of a trellis path equals the joint probability of $(\vecx,\vecy^*)$ and the deletion/trimming event associated with that path.
Then, the claim follows from the fact that $T(\vecx)$ sums the path weight over all paths through the trellis (i.e., all deletion/trimming events) consistent with the given $(\vecx,\vecy^*)$ pair.
\end{IEEEproof}

\section{Polarization operations on a trellis}\label{sec:TrellisPolarization}

Polar plus and minus transforms for channels with memory were first presented in~\cite{Wang_2014,Wang_2015}.
Let an input distribution on $\vecx^N$ be given, for $N$ even. For this input distribution and a vector channel with input $\vecx \in \calX^N$ and output $\vecy$, let $\calT$ be a trellis with $N$ sections whose path-sum function satisfies
\begin{equation}
    \label{eq:generalPathSumFunction}
    T(\vecx) = \Pr(\vecY = \vecy , \vecX=\vecx) \; .
\end{equation}
\subsection{Minus transform}
For a given path-sum function $T(\vecx)$, where $\vecx \in \calX^N$, the polar \emph{minus transform} defines a new path-sum function $T^{[0]} (\vecz)$, $\vecz \in \calX^{N/2}$. Specifically, $T^{[0]}(\vecz)$ is the marginalization of  $T(\vecx)$ over all $\vecx$ vectors satisfying
\[
    \vecz = \vecx^{[0]} = (x_1 \oplus x_2,\ldots,x_{N-1}\oplus x_N) \; .
\]
That is,
\begin{IEEEeqnarray}{rCl}
    T^{[0]} (\vecz) & \triangleq & \sum_{\vecx \in \calX^N: \vecx^{[0]} = \vecz} T(\vecx) \label{eq:minusTransformPathSum} \\
                    &=& \sum_{\vecx \in \calX^{N} } T(\vecx) \prod_{j=1}^{N/2} [x_{2j-1}\oplus x_{2j} = z_j] \IEEEnonumber \\
    & = & \Pr(\vecY = \vecy , \vecX^{[0]} = \vecz) \IEEEnonumber \; ,
\end{IEEEeqnarray}
where the last equality follows under the assumption of (\ref{eq:generalPathSumFunction}).
Due to the local nature of this reparameterization, there is a modified trellis $\calT^{[0]}$ with $N/2$ sections that represents the new path-sum function.

\begin{defi}[Minus Transform] \label{def:minus_transform} Let $\calT = \calT(\calV, \calE, w, \ell, q, r)$ be a length-$N$ trellis, where $N$ is even. The trellis $\tilde{\calT} = \tilde{\calT}(\tilde{\calV}, \tilde{\calE}, \tilde{w}, \tilde{\ell},\tilde{q}, \tilde{r}) = \calT^{[0]}$ is defined as follows.
    \begin{itemize}
        \item The vertex set of $\tilde{\calT}$ is 
        \[
            \tilde{\calV} = \tilde{\calV}_0 \cupdot \tilde{\calV}_1 \cupdot \cdots \cupdot \tilde{\calV}_{N/2} \; ,
        \]
        where
        \[
            \tilde{\calV}_j = \calV_{2j} \; .
        \]

    \item We next define the edge set $\tilde{\calE}$ implicitly. Consider an edge $\tilde{e} = \alpha \to \gamma \in \tilde{\calE}$ in section $j$ of $\tilde{\calT}$ with label $\tilde{\ell}(\tilde{e}) = z$. Then,
        \[
            \alpha \in \tilde{\calV}_{j-1} = \calV_{2j-2} \quad \mbox{and} \quad  \gamma \in \tilde{\calV}_j = \calV_{2j} \; . 
        \]
        The weight $\tilde{w}(\tilde{e})$ of this edge equals the sum of the product of the edge weights along each two-step path
         $\alpha \xrightarrow{e_1}  \beta \xrightarrow{e_2} \gamma$ in $\calT$ with $\lbl(e_1) \xor \lbl(e_2) = z$. That is,
\begin{aligncc*}
    \tilde{w} (\tilde{e}) = \onlydouble{&} \sum_{\substack{e_1 \in \calE_{2j-1} : \\  \sigma(e_1)=\alpha}} \;\;  \sum_{\substack{e_2 \in \calE_{2j} : \\  \tau(e_2)=\gamma} } w(e_1) \, w(e_2) \\
    \onlydouble{& \quad \quad} \times [\tau(e_1)=\sigma(e_2)] \cdot [\ell(e_1)\oplus \ell(e_2) = z].
\end{aligncc*}
Edges with weight $0$ may be removed from $\tilde{\calT}$.
\item The minus operation does not affect initial and final vertices and this implies that $\tilde{\initstate}(s) = \initstate(s)$ and $\tilde{\finalstate}(s) = \finalstate(s)$.
    \end{itemize}
\end{defi}

The following lemma states that applying a minus transform to a trellis indeed results in a trellis whose corresponding path-sum function is the minus transform of the path-sum function of the initial trellis.
\begin{lemm}
    \label{lemm:minusTransform}
    Let $\calT$ be a trellis with $N$ sections, where $N$ is even. Denote the minus transform of $\calT$ by $\calT' = \calT^{[0]}$ per Definition~\ref{def:minus_transform}. Let $T$ and $T'$ be the path-sum functions corresponding to $\calT$ and $\calT'$, respectively, as defined in (\ref{eq:T}) . Then, $T'$ equals $T^{[0]}$ as defined in (\ref{eq:minusTransformPathSum}).
\end{lemm}
\begin{IEEEproof}
This follows from the fact that the minus trellis is constructed by merging adjacent trellis stages and then combining paths according to their $\vecx^{[0]}$ values.
Finally, the new paths are relabeled by their $\vecx^{[0]}$ values. 
\end{IEEEproof}

\subsection{Plus transform}
For a given path-sum function $T(\vecx)$, where $\vecx \in \calX^N$, the polar \emph{plus transform} defines a new path-sum function $T^{[1]} (\vecz')$, $\vecz' \in \calX^{N/2}$. This definition is always with respect to a vector $\vecz \in \calX^{N/2}$, which is assumed to be fixed. Specifically, $T^{[1]}(\vecz')$ equals $T(\vecx)$, where $\vecx$ is the unique vector satisfying
\begin{IEEEeqnarray*}{rCl}
    \vecz &=& \vecx^{[0]} = (x_1 \oplus x_2,\ldots,x_{N-1}\oplus x_N) \quad \mbox{and} \\
    \vecz' &=&  \vecx^{[1]}=(x_2,x_4,\ldots,x_N) \; .
\end{IEEEeqnarray*}
That is,
\begin{IEEEeqnarray}{rCl}
    T^{[1]} (\vecz') & \triangleq & T(\vecx) \big|_{\vecx:\vecx^{[0]} = \vecz, \vecx^{[1]} = \vecz'} \label{eq:plusTransformPathSum} \\ 
& = & \sum_{\vecx \in \calX^{N} } T(\vecx) \prod_{j=1}^{N/2} [x_{2j-1} \xor x_{2j} = z_j] \cdot [x_{2j} = z'_j] \IEEEnonumber \\
& = &\Pr(\vecY = \vecy , \vecX^{[0]} = \vecz, \vecX^{[1]} = \vecz') \; , \IEEEnonumber
\end{IEEEeqnarray}
where the last equality follows under the assumption of (\ref{eq:generalPathSumFunction}).

As with the minus transform, there is a corresponding operation one can apply to the underlying trellis, which we now detail. Note that the plus-transform of a trellis is defined with respect to a fixed vector $\vecz$, which may not be specified explicitly when it is clear from the context.

\begin{defi}[Plus Transform] \label{def:plus_transform} Let $\calT = \calT(\calV, \calE, w, \ell, q, r)$ be a length-$N$ trellis, where $N$ is even and let $\vecz \in \calX^{N/2}$ be given. The trellis $\tilde{\calT} = \tilde{\calT}(\tilde{\calV}, \tilde{\calE}, \tilde{w}, \tilde{\ell},\tilde{q}, \tilde{r}) =  \calT^{[1]}$ is defined as follows.

    \begin{itemize}
        \item The vertex set of $\tilde{\calT}$ is the same as the minus trellis $\calT^{[0]}$. This is also the case for the functions $\tilde{q}$ and $\tilde{r}$. 
    
    \item We next define the edge set $\tilde{\calE}$ implicitly. Consider an edge $\tilde{e} = \alpha \to \gamma \in \tilde{\calE}$ in section $j$ of $\tilde{\calT}$ with label $\tilde{\ell}(\tilde{e}) = z'$. Then,
        \[
            \alpha \in \tilde{\calV}_{j-1} = \calV_{2j-2} \quad \mbox{and} \quad  \gamma \in \tilde{\calV}_j = \calV_{2j} \; . 
        \]
        The weight $\tilde{w}(\tilde{e})$ of this edge equals the sum of the product of the edge weights along each two-step path
        $\alpha \xrightarrow{e_1}  \beta \xrightarrow{e_2} \gamma$ in $\calT$ with $\lbl(e_1) \xor \lbl(e_2) = z_j$ and $\lbl(e_2)=z'$. That is,
\begin{aligncc*}
    \tilde{w} (\tilde{e}) = \onlydouble{&} \sum_{\substack{e_1 \in \calE_{2j-1} : \\  \sigma(e_1)=\alpha}} \;\;  \sum_{\substack{e_2 \in \calE_{2j} : \\  \tau(e_2)=\gamma} } w(e_1) \, w(e_2) \\
    \onlydouble{& \!\!\!} \times [\tau(e_1)=\sigma(e_2)] \cdot [\ell(e_1)\oplus z' = z_j] \cdot [\ell(e_2) = z'] \; .
\end{aligncc*}
Edges with weight $0$ may be removed from $\tilde{\calT}$.
    \end{itemize}
\end{defi}

This lemma states the key property of plus transform.
\begin{lemm}
    \label{lemm:plusTransform}
    Let $\calT$ be a trellis with $N$ sections where $N$ is even, and let $\vecz \in \calX^{N/2}$ be given. Denote the plus transform of $\calT$ by $\calT' = \calT^{[1]}$ per Definition~\ref{def:plus_transform}. Let $T$ and $T'$ be the path-sum functions corresponding to $\calT$ and $\calT'$, respectively, as defined in (\ref{eq:T}) . Then, $T'$ equals $T^{[1]}$ as defined in (\ref{eq:plusTransformPathSum}).
\end{lemm}
\begin{IEEEproof}
This follows from the fact that the plus trellis is constructed by merging adjacent trellis stages and then pruning paths that do not satisfy $\vecx^{[0]}=\vecz$.
Finally, the remaining paths are relabeled with $\vecx^{[1]}$ values. 
\end{IEEEproof}

\subsection{Successive cancellation decoding}
As in Ar\i{}kan's seminal paper~\cite{Arikan_2009}, the transform defined above leads to a SC decoding algorithm. In brief, given $\vecy$ we first construct a base trellis $\calT$. Then, there is a recursive decoder that, given $\calT^{[b_1, b_2, \ldots, b_\lambda]}$, constructs $\calT^{[b_1, b_2, \ldots, b_\lambda, 0]}$ and calls itself with that argument.  When this returns the decoded $\vecxbrbr{b_1, b_2, \ldots, b_\lambda, 0}$, it then builds $\calT^{[b_1, b_2, \ldots, b_\lambda, 1]}$ with respect to those hard decisions and calls itself to decode $\vecxbrbr{b_1, b_2, \ldots, b_\lambda, 1}$. Then, the two decoded vectors are combined to form $\vecxbrbr{b_1, b_2, \ldots, b_\lambda}$ and the function returns. The following lemma makes this precise.
\begin{lemm}
    \label{lemm:recursiveT}
    Let $\calT$ be a base trellis with $N= 2^n$ sections corresponding to a received word $\vecy$ such that (\ref{eq:generalPathSumFunction}) holds for the corresponding path-sum function. For each $i \in [N]$ in order, let $\hat{u}_1^{i-1}$ be a vector of past decisions and $b_1,b_2,\ldots,b_n \in \{0,1\}$ satisfy $i(\vecb) = i$.
     Construct $\calT^{[b_1,b_2,\ldots,b_{n}]}$ iteratively as follows. For $\lambda = 1,2, \ldots, n$, let us define
     \vspace{0.25cm}
     \begin{equation*}
        \calT^{[b_1,b_2,\ldots,b_{\lambda}]} \triangleq
        \begin{cases} 
            (\calT^{[b_1,b_2,\ldots,b_{\lambda-1}]})^{[b_\lambda]} & \mbox{if $\lambda \geq 2$} \; , \\
            \calT^{[b_1]} & \mbox{if } \lambda = 1.
        \end{cases} \vspace{0.25cm}
     \end{equation*}
    If $b_\lambda = 1$, then we apply the plus transform with respect to the fixed vector 
    \begin{equation}
        \label{eq:veczRecursiveT}
        \vecz = \calA_{n-\lambda}^{-1}\left(\hat{u}_{\tau}^{\theta}\right) \; ,
    \end{equation}
    where $\hat{u}_{\tau}^{\theta} \triangleq \left(\hat{u}_{\tau},\hat{u}_{\tau+1},\ldots,\hat{u}_{\theta}\right)$ and 
    \begin{equation}
        \label{eq:thetaTau}
        \theta = \sum_{j=1}^\lambda b_j 2^{n-j} \; , \quad \tau = \theta - 2^{n-\lambda} + 1 \; .
    \end{equation}
    Then, for $\vecU = \calA_n(\vecX)\in \calX^N$, we have
    \[
        T^{[b_1,b_2,\ldots,b_{n}]}(u) = \Pr(U_i = u, U_1^{i-1} = \hat{u}_1^{i-1}, \vecY = \vecy) \; .
    \]
\end{lemm}

\begin{IEEEproof}
    To facilitate a proof by induction, we actually prove a stronger claim. Namely, let  $0 \leq \lambda \leq n$ be given. Define $\vecb_\lambda$ as the vector in $\{0,1\}^n$ whose first $\lambda$ entries equal those of $\vecb$, while the remaining entries are all-zero. That is,
    \begin{equation}
        \label{eq:blambda}
        \vecb_\lambda = (b_1,b_2,\ldots,b_\lambda,0,0,\ldots,0) \; .
    \end{equation}
    Recalling the notation in (\ref{eq:vecxZero})--(\ref{eq:bitReversedI}) and (\ref{eq:xRoundBracket}), we will prove that for all $\vecmu \in \calX^{2^{n-\lambda}}$,
\begin{multlinecc}
    \label{eq:TPlusMinusHalfway}
    T^{[b_1,b_2,\ldots,b_\lambda]}(\vecmu) \\
    =  P(\vecX^{[b_1,b_2,\ldots,b_\lambda]} = \vecmu, \vecX^{(\vecb_\lambda)} = \hat{u}_1^{i(\vecb_\lambda) - 1}, \vecY = \vecy ) \; .
\end{multlinecc}
Clearly, for $\lambda=n$, the reduces to the claimed lemma.

The proof of (\ref{eq:TPlusMinusHalfway}) proceeds by induction on $\lambda$. For the base case, take $\lambda = 0$, and note that (\ref{eq:TPlusMinusHalfway}) holds by assumption: the LHS is by definition $T(\vecmu)$  while the RHS is simply $P(\vecX = \vecmu, \vecY = \vecy)$, and the two are equal by (\ref{eq:generalPathSumFunction}).

For the induction step, we assume that (\ref{eq:TPlusMinusHalfway}) is true for $\lambda$, and prove it to be true for $\lambda + 1$. Assume first that $b_{\lambda+1} = 0$. In this case, $\vecb_\lambda = \vecb_{\lambda+1}$.  Recall that since $b_\lambda = 0$, we get the trellis $\calT^{[b_1,b_2,\ldots,b_\lambda,b_{\lambda+1}]}$ by applying a minus transform (Definition~\ref{def:minus_transform}) on $\calT^{[b_1,b_2,\ldots,b_\lambda]}$. We must prove that (\ref{eq:TPlusMinusHalfway}) holds with $\lambda+1$ in place of $\lambda$, and this is indeed the case by Lemma~\ref{lemm:minusTransform}. Indeed, recall that by our recursive definition, $\vecX^{[b_1,b_2,\ldots,b_\lambda,b_{\lambda+1}]} = \left(\vecX^{[b_1,b_2,\ldots,b_\lambda]}\right)^{[0]}$, and apply Lemma~\ref{lemm:minusTransform}, where in (\ref{eq:generalPathSumFunction}) and (\ref{eq:minusTransformPathSum}) we replace $\vecX$, $\vecY$, and $\vecy$ with $\vecX^{[b_1,b_2,\ldots,b_\lambda]}$, $(\vecY,\vecX^{(\vecb_\lambda)})$, and $(y,\hat{u}_1^{i(\vecb_\lambda) - 1})$, respectively.

Now, let us assume that $b_{\lambda+1} = 1$. Because of this, note that $\vecb_\lambda \neq \vecb_{\lambda+1}$. As before, we assume that (\ref{eq:TPlusMinusHalfway}) is true for $\lambda$, and prove it to be true for $\lambda + 1$. By definition, we get the trellis $\calT^{[b_1,b_2,\ldots,b_\lambda,b_{\lambda+1}]}$ by applying a plus transform (Definition~\ref{def:plus_transform}) on $\calT^{[b_1,b_2,\ldots,b_\lambda]}$, with respect to the vector $\vecz$ defined in (\ref{eq:veczRecursiveT}) and (\ref{eq:thetaTau}), with $\lambda$ replaced by $\lambda+1$. Thus, if we denote by $T$ the probability function associated with $\calT^{[b_1,b_2,\ldots,b_\lambda]}$, we get by Lemma~\ref{lemm:plusTransform} that the probability function associated with $\calT^{[b_1,b_2,\ldots,b_\lambda,b_{\lambda+1}]}$, which we denote by $T'$, satisfies
\begin{IEEEeqnarray*}{rCl}
    T'(\vecz') & = &  T(\vecmu) \\
               & = &  P(\vecX^{[b_1,b_2,\ldots,b_\lambda]} = \vecmu, \vecX^{(\vecb_\lambda)} = \hat{u}_1^{i(\vecb_\lambda) - 1}, \vecY = \vecy ) \; ,
\end{IEEEeqnarray*}
where $\vecmu$ is the unique vector for which $\vecmu^{[0]} = \vecz$ and $\vecmu^{[1]} = \vecz'$. The condition $\vecX^{[b_1,b_2,\ldots,b_\lambda]} = \vecmu$ is equivalent to the pair of conditions
\[
    \vecX^{[b_1,b_2,\ldots,b_\lambda,0]} = \vecmu^{[0]} \quad \mbox{and} \quad \vecX^{[b_1,b_2,\ldots,b_\lambda,1]} = \vecmu^{[1]} \; .
\]
That is, to the pair of conditions
\[
    \vecX^{[b_1,b_2,\ldots,b_\lambda,0]} = \vecz \quad \mbox{and} \quad \vecX^{[b_1,b_2,\ldots,b_\lambda,b_{\lambda+1}]} = \vecz' \; .
\]
We will shortly prove that the pair of conditions 
\begin{equation}
\label{eq:lastStepOfRecursiveT_A}
    \vecX^{(\vecb_\lambda)} = \hat{u}_1^{i(\vecb_\lambda) - 1} \quad \mbox{and} \quad \vecX^{[b_1,b_2,\ldots,b_\lambda,0]} = \vecz
\end{equation}
can be simplified to
\begin{equation}
    \label{eq:lastStepOfRecursiveT_B}
    \vecX^{(\vecb_{\lambda+1})} = \hat{u}_1^{i(\vecb_{\lambda+1}) - 1} \; .
\end{equation}
Once this is proved, the lemma follows, since the above implies that
\begin{aligncc*}
    T'\onlydouble{&}(\vecz') = \\
    \onlydouble{&} P(\vecX^{[b_1,b_2,\ldots,b_\lambda,b_{\lambda+1}]} = \vecz', \vecX^{(\vecb_{\lambda+1})} = \hat{u}_1^{i(\vecb_{\lambda+1}) - 1}, \vecY = \vecy ).
\end{aligncc*}

Let us now show that (\ref{eq:lastStepOfRecursiveT_A}) is equivalent to (\ref{eq:lastStepOfRecursiveT_B}). Since $b_{\lambda+1}=1$, the set of 
transforms we need to add to $\vecX^{(\vecb_\lambda)}$ in order to get $\vecX^{(\vecb_{\lambda+1})}$ are those with prefix $(b_1,b_2,\ldots,b_\lambda,0)$. That is, we are missing the $\calA_{n-(\lambda+1)}$ transform of $\vecX^{[b_1,b_2,\ldots,b_\lambda,0]}$, and this transform must equal $\hat{u}_{i(\vecb_{\lambda})}^{i(\vecb_{\lambda+1}) - 1}$. To see that this indeed is the case, we observe that $\vecz$ is defined by (\ref{eq:veczRecursiveT}) and (\ref{eq:thetaTau}) with $\lambda+1$ in place of $\lambda$.
Recalling (\ref{eq:bitReversedI}) and (\ref{eq:blambda}), and keeping in mind that in  (\ref{eq:thetaTau}) we replace $\lambda$ by $\lambda+1$, we see that $\theta = i(\vecb_{\lambda+1})-1$ while $\tau = i(\vecb_{\lambda})$.
\end{IEEEproof}
Actually, the above lemma is not unique to the deletion channel and it applies to any base trellis for which (\ref{eq:generalPathSumFunction}) holds.  The above lemma also gives an efficient method for deciding the value of $\hat{u}_i$ at stage $i$, since
\begin{multlinecc}
    \label{eq:trellisConditionalProbCalc}
    \Pr(U_i = u | U_1^{i-1} = \hat{u}_1^{i-1}, \vecY = \vecy)
    \\
    = \frac{ T^{[b_1,b_2,\ldots,b_{n}]}(u)}{\displaystyle \sum_{u' \in \cal X} T^{[b_1,b_2,\ldots,b_{n}]}(u') } 
\end{multlinecc}
when $\Pr(U_1^{i-1} = \hat{u}_1^{i-1}, \vecY = \vecy) > 0$.

\subsection{Complexity}
\label{subsec:complexity}

In~\cite{Wang_2014}, SC trellis decoding is generalized to finite-state channels with memory.
For a finite-state channel with $A$ states, the decoding complexity of a length-$N$ code is shown to be $O(A^3 N \log N)$.
While there are some connections between finite-state channels and deletion channels~\cite{Castiglione_2015}, it is not clear if this complexity result can be applied directly to the deletion channel.
Using a different formulation, a SC decoder for polar codes on the deletion channel is defined in~\cite{Tian_2017}. Its complexity is $O(N^4 \log N)$ for a constant deletion rate and a uniform input distribution\footnote{As noted earlier, the complexity of the decoding algorithm in~\cite{Tian_2017} is misstated as $O(d^2 N \log N)$ for $d$ deletions but it is actually $O(d^3 N \log N)$.}.

In this section, we bound the complexity of computing the plus and minus transformations of a trellis. 
For a trellis $\mathcal{T}$ with $N$ sections, let $P_2 (j)$ be the number of distinct 2-step paths from states in $\mathcal{V}_{2j}$ to states in $\mathcal{V}_{2j+2}$ and define
\[ C(\mathcal{T}) \triangleq \sum_{j=0}^{N/2 - 1} P_2 (j). \]
From Definition~\ref{def:minus_transform}, one can verify that the minus transform requires $C(\mathcal{T})$ multiplies and adds to compute $\mathcal{T}^{[0]}$.
Similarly, from Definition~\ref{def:plus_transform}, it follows that the plus transform requires at most $C(\mathcal{T})$ multiplies and adds to compute $\mathcal{T}^{[1]}$.

Consider a trellis $\mathcal{T}_\lambda$ at depth-$\lambda$ in the decoding process.
Such a trellis will have $2^{n-\lambda}$ sections each corresponding to $2^\lambda$ channel uses.
For the deletion channel, we observe that each state in $\mathcal{V}_{2j}$ has at most $2(2^\lambda+1)|\mathcal{S}|$ outgoing edges.
This is because each edge can be labeled by 0 or 1, the number of deletions (between $0$ and $2^{\lambda}$) determines the change in the channel state, and the input state can change to any of $|\mathcal{S}|$ possibilities.
Combining these observations, and noting that the number of vertices in each segment is at most $2^n|\calS|$, we see that
\begin{align*}
C(\mathcal{T}_{\lambda}) &\leq  2^n |\mathcal{S}| \cdot \left( 2(2^{\lambda}+1) |\mathcal{S}| \right)^2 2^{n-\lambda} \leq 2^{2n+2} (2^\lambda+3)  |\mathcal{S}|^3.
\end{align*}
Since the full decoder uses $2^\lambda$ plus and minus operations at depth $\lambda$, the overall decoding complexity is
\[ \sum_{\lambda=0}^{n-1} 2^\lambda 2^{2n+2} (2^\lambda+3)  |\mathcal{S}|^3 = O(|\mathcal{S}|^3 N^4), \]
which is lower than previous methods by a $\log N$ factor.
This occurs because the $\lambda=n-1$ decoding step dominates the calculation and has $O(|\mathcal{S}|^3 N^4)$ complexity by itself.

The reader should happily note that the above quartic growth in $N$ is \emph{not} present in Theorem~\ref{theo:main}. The overall complexity of our scheme is much smaller because the guard bands allow the codeword to be separated into many smaller blocks whose trellises can be processed separately.

\section{Information rates}
\label{sec:informationRates}
In this section, we will introduce and analyze various information rates related to polar codes on the deletion channel.
For a given regular hidden-Markov input distribution, let $\vecX$ be an input vector of length $N$ and let $\vecY$ be the corresponding output vector (i.e., the observation of $\vecX$ through the deletion channel).
The main goal of this paper is to show that our polar coding scheme achieves the information rate
\begin{equation}
    \label{eq:informationRate}
    \calI = \lim_{N \to \infty} \frac{I(\vecX;\vecY)}{N} \; ,
\end{equation}
where $\vecX$ and $\vecY$ depend implicitly on $N$.
This existence of this limit is well-known~\cite{Dobrushin_1967} but we revisit it here because the same argument will be used later with slight variations.

\begin{lemm}
    \label{lemm:twoLimitsExist}
    Fix a hidden-Markov input distribution.  For a given $N$, let $\vecX = (X_1,X_2,\ldots,X_N)$ be a random vector with the above distribution. Let $\vecY$ be the result of passing $\vecX$ through a deletion channel with deletion probability $\delta$. Then, the following two limits exist,
\begin{equation}
    \label{eq:twoLimits}
    \lim_{N \to \infty} \frac{H(\vecX)}{N} \quad  \mbox{and} \quad \lim_{N \to \infty} \frac{H(\vecX|\vecY)}{N} \; .
\end{equation}
\end{lemm}
\begin{IEEEproof}
The proof of this lemma is detailed below for uniform inputs in Section~\ref{sec:info_rate_lemma_uniform} and hidden-Markov inputs in Section~\ref{sec:info_rate_lemma_hm}.
\end{IEEEproof}

Once the limits in~\eqref{eq:twoLimits} are established, the limit in~(\ref{eq:informationRate}) follows because
\[
    \frac{I(\vecX;\vecY)}{N} = \frac{H(\vecX)}{N} - \frac{H(\vecX|\vecY)}{N} \; .
\]

\subsection{Uniform input} \label{sec:info_rate_lemma_uniform}
In this subsection, we prove Lemma~\ref{lemm:twoLimitsExist}, for the restricted case in which the input distribution is i.i.d.\ and uniform. 
\begin{IEEEproof}[Proof of Lemma~\ref{lemm:twoLimitsExist} for Uniform Inputs]
In such a setting, the first limit in (\ref{eq:twoLimits}) clearly exists and equals $1$. To prove the second limit in (\ref{eq:twoLimits}), let us first define
\begin{equation}
    \label{eq:calH}
    \calH_N = H(\vecX|\vecY) \; , \quad |\vecX| = N \; .
\end{equation}
Our plan is to show that the sequence $\calH_N$ is superadditive, implying \cite[Lemma 1.2.1, page 3]{Steele:97b} the existence of the second limit in (\ref{eq:twoLimits}). Indeed, let $N_1$ and $N_2$ be given, and let $\vecX$ and $\vecX'$ be distributed according to the input distribution, and having lengths $N_1$ and $N_2$, respectively. Denote the outputs corresponding to to $\vecX$ and $\vecX'$ by $\vecY$ and $\vecY'$, respectively. We have
\begin{IEEEeqnarray*}{rCl}
    \calH_{N_1 + N_2} &=& H(\vecX \concat \vecX'| \vecY \concat \vecY') \\
                      &\eqann[=]{a}& H(\vecX, \vecX'| \vecY \concat \vecY') \\
                      &\geq& H( \vecX, \vecX'| \vecY \concat \vecY', \vecY, \vecY') \\
                      & \eqann{b} & H(\vecX,\vecX'| \vecY, \vecY') \\
                      & \eqann{c} & H(\vecX| \vecY, \vecY') + H(\vecX'|\vecX, \vecY, \vecY') \\
                      & \eqann{d} & H(\vecX| \vecY) + H(\vecX'|\vecY') \\
                      & = & \calH_{N_1} + \calH_{N_2} \; ,
\end{IEEEeqnarray*}
where \eqannref{a} holds because $N_1$ and $N_2$, the lengths of $\vecX$ and $\vecX'$, respectively, are constant parameters; \eqannref{b} holds because $ \vecY \concat \vecY'$ is a function of $\vecY$ and $\vecY'$; \eqannref{c} follows by the chain rule; \eqannref{d} holds because, for the i.i.d.\ uniform input distribution, the pair $(\vecX,\vecY)$ is independent of the pair $(\vecX',\vecY')$.
Hence, the sequence $\calH_N$ is indeed superadditive.
\end{IEEEproof}

\subsection{Hidden-Markov input} \label{sec:info_rate_lemma_hm}

We now prove Lemma~\ref{lemm:twoLimitsExist} for the case where the input distribution is a regular hidden-Markov process. Since now $\calH_N$ is not generally superadditive, we will take an indirect route to prove Lemma~\ref{lemm:twoLimitsExist}. Indeed, the following lemma is proved by defining a related quantity, $\hat{\calH}_N$, which is superadditive.
\begin{lemm}
    \label{lemm:HHatLimitExists}
    Fix a regular hidden-Markov input distribution.  For a given $N$, let $\vecX = (X_1,X_2,\ldots,X_N)$ be a random vector with the above distribution. Let $\vecY$ be the result of passing $\vecX$ through a deletion channel with deletion probability $\delta$. Then, the following limit exists:
\begin{equation}
    \label{eq:HHatLimit}
    \lim_{N \to \infty} \frac{H(\vecX|\vecY,S_0,S_N)}{N} \; .
\end{equation}
\end{lemm}

\begin{IEEEproof}
Define
\begin{equation}
    \hat{\calH}_N = H(\vecX|\vecY,S_0,S_N) \; , \quad |\vecX| = N \; . \label{eq:Hhat}
\end{equation}
To borrow the terminology of \cite{Shuval_Tal_Memory_2017}, the above defines the \emph{boundary-state-aware entropy}. Note that $S_0$ and $S_N$ are the states just before transmission has started, and just after transmission has ended, respectively.

We now show that $\hat{\calH}_N$ is superadditive. Indeed, let $\vecX$ and $\vecX'$ be consecutive input vectors of length $N_1$ and $N_2$, respectively. That is, $\vecX \concat \vecX'$ is a vector of length $N_1+N_2$ drawn from the input distribution. Denote by $\vecY$ and $\vecY'$ the output vectors corresponding to $\vecX$ and $\vecX'$, respectively. Then,
\begin{IEEEeqnarray*}{rCl}
    \hat{\calH}_{N_1 + N_2} &=& H(\vecX \concat \vecX'| \vecY \concat \vecY', S_0, S_{N_1+N_2}) \\
    &\eqann[=]{a}& H(\vecX,\vecX'| \vecY \concat \vecY', S_0, S_{N_1+N_2}) \\
                      & \eqann[\geq]{b} & H(\vecX,\vecX'| \vecY, \vecY', S_0, S_{N_1+N_2}) \\
                      & \geq & H(\vecX,\vecX'| \vecY, \vecY', S_0, S_{N_1}, S_{N_1+N_2}) \\
                      & \eqann{c} & H(\vecX| \vecY, \vecY', S_0, S_{N_1}, S_{N_1+N_2}) \\
                      & & + H(\vecX'|\vecX, \vecY, \vecY', S_0, S_{N_1}, S_{N_1+N_2}) \\
                      & \eqann{d} & H(\vecX| \vecY,S_0,S_{N_1}) + H(\vecX'|\vecY',S_{N_1},S_{N_1+N_2}) \\
                      & = & \hat{\calH}_{N_1} + \hat{\calH}_{N_2} \; ,
\end{IEEEeqnarray*}
where \eqannref{a} holds because $N_1$ and $N_2$, the lengths of $\vecX$ and $\vecX'$, respectively, are constant parameters; \eqannref{b} holds because $ \vecY \concat \vecY'$ is a function of $\vecY$ and $\vecY'$; \eqannref{c} follows by the chain rule; \eqannref{d} holds because of conditional independence: given $S_{N_1}$, $(\vecX, \vecY,S_0)$ is independent of $(\vecX',\vecY',S_{N_1+N_2})$.
Hence, the sequence $\hat{\calH}_N$ is indeed superadditive, and the following limit exists by \cite[Lemma 1.2.1, page 3]{Steele:97b},
\[
    \lim_{N \to \infty} \frac{\hat{\calH}_N}{N} \; .
\]
\end{IEEEproof}

All that remains now is to account for the difference in the entropies of $\calH_N$ and $\hat{\calH}_N$, incurred by conditioning on $S_0$ and $S_N$. As will be made clear in the following proof, this difference can be bounded by a constant, and hence vanishes when we divide by $N$.

\begin{IEEEproof}[Proof of Lemma~\ref{lemm:twoLimitsExist} for hidden-Markov inputs]
    We first note that the existence of the second limit in (\ref{eq:twoLimits}) implies the existence of the first limit. Indeed, taking the deletion probability $\delta$ equal to $1$ makes the second limit equal the first. Hence, all that remains is to prove the existence of the second limit.
    
    To show that the second limit in (\ref{eq:twoLimits}) exists, note that, for $|\vecX| = N$, we have on the one hand that
\begin{IEEEeqnarray*}{rCl}
    H(\vecX,S_0,S_N|\vecY) & = & H(\vecX|\vecY) + H(S_0,S_N|\vecX,\vecY) \\
                           & \geq & H(\vecX|\vecY) \\
                           & = & \calH_N \; ,
\end{IEEEeqnarray*}
and on the other hand that
\begin{IEEEeqnarray*}{rCl}
    H(\vecX,S_0,S_N|\vecY) & = & H(S_0,S_N|\vecY) + H(\vecX|\vecY,S_0,S_N) \\
                           & \leq &  2 \log_2 |\calS| +  H(\vecX|\vecY,S_0,S_N) \\
                           & = &  2 \log_2 |\calS| + \hat{\calH}_N  \; .
\end{IEEEeqnarray*}
Thus,
\[
 \calH_N \leq \hat{\calH}_N + 2 \log_2 |\calS|   \; .
\]
Since it is easily seen that $\hat{\calH}_N \leq \calH_N$, we have that
\begin{equation}
    \label{eq:preLimit}
    \frac{\hat{\calH}_N}{N} \leq \frac{\calH_N}{N} \leq \frac{\hat{\calH}_N}{N} + \frac{2 \log_2 |\calS|}{N} \; .
\end{equation}
We have already proved that the limit of the LHS of (\ref{eq:preLimit}) exists, in Lemma~\ref{lemm:HHatLimitExists}. Since the limit of $(2 \log_2 |\calS|)/N$ is $0$, the limit of the RHS of (\ref{eq:preLimit}) exists and equals that of the LHS. By the sandwich property, the limit of the middle term exists as well, which is the desired result. 
\end{IEEEproof}

We finish by restating the last part of the proof as a lemma.
\begin{lemm}
    \label{lemm:HatNoHatLimitsEqual}
    Fix a hidden-Markov input distribution.  For a given $N$, let $\vecX = (X_1,X_2,\ldots,X_N)$ be a random vector with the above distribution. Let $\vecY$ be the result of passing $\vecX$ through a deletion channel with deletion probability $\delta$. Then,
\begin{equation}
    \label{eq:HatNoHatLimitsEqual}
    \lim_{N \to \infty} \frac{H(\vecX|\vecY,S_0,S_N)}{N} = \lim_{N \to \infty} \frac{H(\vecX|\vecY)}{N} \; .
\end{equation}
\end{lemm}

\section{Weak polarization}
\label{sec:weak}

In this section, we prove weak polarization for both the deletion channel and the trimmed deletion channel, as defined in Subsection~\ref{subsec:TrellisTDC}. As in \cite{Arikan_2009}, we will first prove that a certain process is submartingale, and then prove that it either converges to $0$ or to $1$.

As a first step, we will shortly define three entropies. These are defined with respect to an input $\vecX$ of length $N=2^n$, which has a regular hidden-Markov input distribution, and $\vecU= \calA_n (\vecX)$. The corresponding output is denoted $\vecY$. Recall that $S_0$ and $S_N$ are the (hidden) states of the input process, just before $\vecX$ is transmitted and right after $\vecX$ is transmitted, respectively. Lastly, denote by $\vecY^\DPT$ the result of trimming all leading and trailing `$0$' symbols from $\vecY$. Then, for a given $n$ and $1 \leq i \leq N = 2^n$, define
the following (deterministic) entropies:
\begin{IEEEeqnarray}{rCl}
    h_i &=& H(U_i|U_1^{i-1},\vecY) \; , \label{eq:hi} \\
    \hat{h}_i &=& H(U_i|U_1^{i-1},S_0,S_N,\vecY) \; , \label{eq:hathi}\\
    h_i^\DPT &=& H(U_i|U_1^{i-1},\vecY^\DPT) \label{eq:hiTDC} \; .
\end{IEEEeqnarray}
Clearly,
\[
    h_i^\DPT \geq h_i \geq \hat{h}_i \; .
\]
Note that in the case of a uniform input distribution, there is only one state, and hence $h_i$ and $\hat{h}_i$ are equal.

Following \cite{Arikan_2009}, we show weak polarization by considering a sequence $B_1,B_2,\ldots$ of i.i.d.\ $\mathrm{Ber}(1/2)$ random variables. For any $n\in \mathbb{N}$, let $J_n = i(B_1,B_2,\ldots,B_{n})$ be the random index defined by (\ref{eq:bitReversedI}), with $B_t$ in place of $b_t$.
We will study the three related random processes defined for $n\in \mathbb{N}$ by
\begin{IEEEeqnarray}{rCl}
    H_n  &=& h_{J_n} \; , \label{eq:Hn} \\
    \hat{H}_n     &=& \hat{h}_{J_n} \; , \label{eq:hatHn} \\
    H^\DPT_{n} &=& h_{J_n}^\DPT \; . \label{eq:HnTDC}
\end{IEEEeqnarray}
The arguments below will show that $\hat{H}_n$ is a submartingale, converging to either $0$ or $1$. From this we will infer that $H_n$ and $H^\DPT_{n}$ must converge to either $0$ or $1$ as well. Though neither $H_n$ nor $H^\DPT_{n}$ are necessarily submartingales.

\begin{theo} \label{thm:hatH_polarizes}
The sequence $\hat{H}_n$ converges (almost surely and in $L^1$) to a well-defined random variable $\hat{H}_\infty\ \in \{0,1\}$ and, for any $\epsilon > 0$, it follows that 
\begin{align}
\frac{1}{N} \left| \left\{ i \in [N] \, | \, H(U_i|U_1^{i-1},S_0,S_N,\vecY) \in [\epsilon, 1-\epsilon] \right\} \right| & \to  0. 
\end{align}
\end{theo}

\begin{IEEEproof}
Lemma~\ref{lemm:subMartingale} below shows that $\hat{H}_1,\hat{H}_2,\hat{H}_3,\ldots \in [0,1]$ is a bounded submartingale with respect to $J_n$.
This implies that the sequence $\hat{H}_n$ converges (almost surely and in $L^1$) to a limit that is denoted by $\hat{H}_\infty$~\cite[p.~236]{Durrett-2019}.
Lemma~\ref{lemm:weakPolarizationFAIM} below shows that, for any $\epsilon >0$, there is a $\Delta > 0$ such that $\hat{H}_n \in [\epsilon,1-\epsilon]$ implies $\hat{H}_{n+1} > \hat{H}_n + \Delta$ with probability $\frac{1}{2}$.
Thus, the sequence $\hat{H}_n$ cannot converge to the set $(0,1)$ and hence $\hat{H}_\infty \in \{0,1\}$.

From~\eqref{eq:hathi} and~\eqref{eq:hatHn}, we see that $\Pr \left( \hat{H}_n \in [\epsilon,1-\epsilon] \right)$ equals
\[ \frac{1}{N} \left| \left\{ i \in [N] \, | \, H(U_i|U_1^{i-1},S_0,S_N,\vecY) \in [\epsilon, 1-\epsilon] \right\} \right|. \]
Since $\hat{H}_n$ converges almost surely to $\hat{H}_\infty$ and $\epsilon,1-\epsilon$ are continuity points of $\Pr(\hat{H}_\infty \leq x)$~\cite[Ch.~4]{Durrett-2019}, it follows that
\[ \lim_{n\to \infty} \Pr \left(\hat{H}_n \in [\epsilon,1-\epsilon] \right) =  \Pr \left( \hat{H}_\infty \in [\epsilon,1-\epsilon] \right) = 0. \]
This completes the proof.
\end{IEEEproof}

\begin{lemm}
    \label{lemm:subMartingale}
    For a hidden-Markov input distribution and a deletion channel with deletion probability $\delta$, let $\hat{H}_n$ and $J_n$ be as defined above. Then, the sequence $\hat{H}_1,\hat{H}_2,\hat{H}_3,\ldots$ is a bounded submartingale with respect to the $J_1,J_2,J_3,\ldots$ sequence.
\end{lemm}

\begin{IEEEproof}
    Since $\hat{H}_n$ is clearly bounded between $0$ and $1$, it remains to show that $E(\hat{H}_{n+1}|J_1,J_2,\ldots,J_n) \geq  \hat{H}_n$. Let $\vecX \concat \vecX'$ be a length-$2N$ input to the channel. Denote by $\vecY \concat \vecY'$ the corresponding output, where $\vecY$ only contains inputs from $\vecX$ and $\vecY'$ only contains inputs from $\vecX'$. 
    Recall that $\vecU = \calA_{n}(\vecX)$ and define $\vecV = \calA_n(\vecX')$ and
    \[ \vecF = (U_1 \xor V_1, V_1,U_2 \xor  V_2, V_2,\ldots,U_N \xor V_N,V_N). \]
    By (\ref{eq:bitReversedI}), we have that $J_{n+1} = 2J_n -1$ with probability $1/2$ and $J_{n+1} = 2J_n$ with probability $1/2$. Thus,
    \begin{IEEEeqnarray*}{rCl}
        \IEEEeqnarraymulticol{3}{l}{E(\hat{H}_{n+1} |J_1^n)} \\
        \quad  & =& E \big( H(F_{J_{n+1}} | F_1^{J_{n+1}-1},\vecY \concat \vecY',S_0,S_{2N}) | J_1^n \big) \\
        & = &  \frac{1}{2} H(F_{2 J_{n}-1} | F_1^{2 J_{n}-2},\vecY \concat \vecY',S_0,S_{2N}) \\
        & & \quad + \frac{1}{2} H(F_{2 J_{n}} | F_1^{2 J_{n}-1},\vecY \concat \vecY',S_0,S_{2N}) \\    
        & = &\frac{1}{2} H(F_{2 J_{n}-1},F_{2 J_{n}} | F_1^{2 J_{n}-2},\vecY \concat \vecY',S_0,S_{2N})  \\
        & =& \frac{1}{2} H(U_{J_n} \xor V_{J_n},V_{J_n} | F_1^{2 J_{n}-2},\vecY \concat \vecY',S_0,S_{2N})  \\
        & = & \frac{1}{2} H(U_{J_n},V_{J_n} | U_1^{J_{n}-1},V_1^{J_{n}-1},\vecY \concat \vecY',S_0,S_{2N}) \\
        & \eqann[\geq]{a} & \frac{1}{2} H(U_{J_n},V_{J_n} | U_1^{J_{n}-1},V_1^{J_{n}-1},\vecY,\vecY',S_0,S_{2N}) \\
        & \eqann[\geq]{b} & \frac{1}{2} H(U_{J_n},V_{J_n} | U_1^{J_{n}-1},V_1^{J_{n}-1},\vecY,\vecY',S_0,S_N,S_{2N}) \\
        & \eqann{c} & \frac{1}{2} H(U_{J_n} | U_1^{J_{n}-1},\vecY,S_0,S_N)  \\
        && \quad + \frac{1}{2} H(V_{J_n} | V_1^{J_{n}-1},\vecY',S_N,S_{2N}) \\
        & \eqann{d} & \hat{H}_{n}.
    \end{IEEEeqnarray*}

    The  inequality \eqannref{a} follows from the fact that $\vecY \concat \vecY'$ is a deterministic function of $\vecY,\vecY'$. Inequality \eqannref{b} follows since conditioning reduces entropy. Step \eqannref{c} holds by the Markov property.
    Finally, \eqannref{d} is  due to stationarity: $\hat{H}_n = H(U_{J_n} | U_1^{J_{n}-1},\vecY,S_0,S_N) = H(V_{J_n} | V_1^{J_{n}-1},\vecY',S_N,S_{2N})$.
\end{IEEEproof}

\vspace{1mm}
Since the sequence $\hat{H}_n$ is a bounded submartingale, it converges almost surely and in $L_1$ to a random variable $\hat{H}_\infty \in [0,1]$. 
To show that $\hat{H}_\infty \in \{0,1\}$ with probability 1, one can show that, if $\epsilon \leq \hat{H}_n \leq 1-\epsilon$, then there is a $\Delta = \Delta(\epsilon) > 0$ such that $\hat{H}_n^- - \hat{H}_n > \Delta(\epsilon)$, where
\begin{equation}
    \label{hat_n_plus}
    \hat{H}_n^- \triangleq  H(U_{J_n} \xor V_{J_n} |U_1^{J_{n}-1},V_1^{J_{n}-1},\vecY \concat \vecY',S_0,S_{2N}) \; .
\end{equation}
That is, a `minus' operation applied to non-polarized entropy changes the entropy by at least $\Delta$. Such a result indeed establishes the above, since it dictates that $\hat{H}_n$ cannot converge to anything other than either $0$ or $1$. As before, we first prove the above for the simple case of i.i.d.\ uniform input, and then generalize to a hidden-Markov input.
\subsection{Uniform input}

\begin{lemm}
    \label{lemm:weakPolarizationSymmetric}
    Let $\vecX$ and $\vecX'$ be independent vectors of length $N = 2^n$, both drawn from an i.i.d.\ uniform distribution. Let $\hat{H}_n$ and $\hat{H}_n^-$ be as defined in (\ref{eq:hathi}), (\ref{eq:hatHn}), and (\ref{hat_n_plus}), with $S_0$, $S_N$ and $S_{2N}$ being degenerate random variables always taking the value $1$. Then, for every $\epsilon > 0$ there exists $\Delta(\epsilon) > 0$ such that if $\epsilon \leq \hat{H}_n \leq 1- \epsilon$, then $\hat{H}_n^- - \hat{H}_n > \Delta(\epsilon)$.
\end{lemm}

\begin{IEEEproof}
    Denote $i=J_n$, and assume a fixed $\epsilon$ for which $\epsilon \leq \hat{H}_n \leq 1- \epsilon$.  Then, since $S_0$, $S_N$, and $S_{2N}$ are degenerate, we observe that $(U_i,U_1^{i-1},\vecY)$ is independent of $(V_i,V_1^{i-1},\vecY')$.
    It follows that
    \[
            H(U_i \xor V_i|U_1^{i-1},V_1^{i-1},\vecY, \vecY')
    \]
is the entropy of the modulo-2 sum of the independent binary random variables $U_i$ and $V_i$ .
Thus, Mrs.\ Gerber's Lemma~\cite[Lemma 2.2]{sasoglu:12b} implies that, for every $\epsilon >0$, there is $\Delta >0$ such that
    \[
        H(U_i \xor V_i|U_1^{i-1},V_1^{i-1},\vecY, \vecY') - H(U_i|U_1^{i-1},\vecY) \geq \Delta \; .
    \]
	Since
    \begin{align*}
        \hat{H}_{n+1}^- &= H(U_i \xor V_i|U_1^{i-1},V_1^{i-1},\vecY \concat \vecY') \\
        &\geq H(U_i \xor V_i|U_1^{i-1},V_1^{i-1},\vecY, \vecY'),
    \end{align*}
    the result follows.
\end{IEEEproof}

\subsection{Hidden-Markov input}
The proof of Lemma~\ref{lemm:weakPolarizationSymmetric} above relied on the mutual independence of $(U_i,U_1^{i-1},\vecY)$ and $(V_i,V_1^{i-1},\vecY')$. To emulate\footnote{For independence, it is sufficient to condition on the event $S_N = s_N$. Conditioning on the more specific event $S_0 = s_0$, $S_N = s_N$, $S_{2N} = s_{2N}$ is needed for latter parts.} this property in a FAIM setting, we note that for $s_0$, $s_N$, and $s_{2N}$ fixed, we indeed have that $(U_i,U_1^{i-1},\vecY)$  and $(V_i,V_1^{i-1},\vecY')$ are independent, when conditioning on the event $S_0 = s_0$, $S_N = s_N$, $S_{2N} = s_{2N}$. 
Towards this end, for $s_0,s_N,s_{2N} \in \calS$, we denote the probability of these three states occurring as
\begin{equation}
    \label{eq:ps0sNs2N}
    p(s_0,s_N,s_{2N}) =  \Pr(S_0 = s_0, S_N = s_N, S_{2N} = s_{2N} ) \; .
\end{equation}

In the reminder of this subsection, we will assume that $N$ is large enough such that the above probability is always positive. This is indeed possible, by the following lemma.
\begin{lemm}
    \label{lemm:allTripletsProbable}
    For $s \in \calS$, denote by $\pi(s)$ the stationary probability of $s$. That is, the probability that $S_0=s$. Let
    \[
        \pimin = \min_{s \in \calS} \pi(s) \; , 
    \]
    Then, $\pimin > 0$, and there exists a $\nu$ such that for all $N \geq 2^{\nu}$ and all $s_0,s_N, s_{2N} \in \calS$ we have
    \begin{equation}
        \label{eq:allStatesProbable}
        \Pr(S_0 = s_0, S_N = s_N, S_{2N} = s_{2N} ) > \frac{(\pimin)^3}{2} \; .
    \end{equation}                                                                                                                                                                 
\end{lemm}
\begin{IEEEproof}
    Since the underlying Markov chain is regular (i.e., finite-state, irreducible, and aperiodic), some power of the transition matrix must be strictly positive and this implies that $\pimin>0$. Regularity further implies that $S_0,S_N,S_{2N}$ become asymptotically independent as $N$ increases.
    Thus, there must be an $N_0 =2^{n_0}$ such that (\ref{eq:allStatesProbable}) holds for all $N\geq N_0$.
\end{IEEEproof}

For $(s_0,s_N,s_{2N})$, we define the quantities $\alpha(s_0,s_N,s_{2N})$ and $\beta(s_0,s_N,s_{2N})$ as follows.
\begin{IEEEeqnarray}{l}
 \alpha(s_0,s_N,s_{2N}) \triangleq \label{eq:alpha} \\ 
 H(U_i \xor V_i|U_1^{i-1},V_1^{i-1},\vecY, \vecY',S_0=s_0,S_N=s_N,S_{2N}=s_{2N}) \nonumber
 \end{IEEEeqnarray}
 and
 \begin{equation}
     \label{eq:beta}
     \beta(s_0,s_N,s_{2N}) \triangleq  \frac{\gamma(s_0,s_N) + \gamma(s_N,s_{2N})}{2} \; ,
\end{equation}
where
 \begin{equation}
     \label{eq:gamma}
    \gamma(s_0,s_N) \triangleq  H(U_i|U_1^{i-1},\vecY,S_0=s_0,S_N=s_N) \; .
\end{equation}
Note that by stationarity,
\[
    \gamma(s_N,s_{2N}) =  H(V_i|V_1^{i-1},\vecY',S_N=s_N,S_{2N}=s_{2N}) \; .
\]

The following lemma states how $\alpha$ and $\beta$ are related to our quantities of interest, $\hat{H}_n$ and  $\hat{H}_n^-$. 
\begin{lemm}
    \label{lemm:alphaBetaConnection}
    Let $N = 2^n > 2^{\nu}$, where $\nu$ was promised in Lemma~\ref{lemm:allTripletsProbable}. Then, for $\alpha$ and $\beta$ as defined above, we have that
    \begin{equation}
        \label{eq:hatH_alpha}
        \hat{H}_n^- \geq \sum_{s_0,s_N,s_{2N} \in \calS}  p(s_0,s_N,s_{2N}) \cdot \alpha(s_0,s_N,s_{2N}) \; ,
    \end{equation}
    and
    \begin{equation}
        \label{eq:hatH_beta}
     \hat{H}_n =  \sum_{s_0,s_N,s_{2N} \in \calS}  p(s_0,s_N,s_{2N}) \cdot \beta(s_0,s_N,s_{2N}) \; .  
    \end{equation}
    Furthermore, for all $s_0,s_N,s_{2N} \in \calS$,
    \begin{equation}
    \label{eq:alphaGeqBeta}
    \alpha(s_0,s_N,s_{2N}) \geq \beta(s_0,s_N,s_{2N}) \; .
    \end{equation}
\end{lemm}
\begin{IEEEproof}
    Define $i = J_n$. To prove (\ref{eq:hatH_alpha}), we proceed similarly to the proof in Lemma~\ref{lemm:subMartingale} and deduce that
    \begin{IEEEeqnarray*}{rCl}
        \hat{H}_n^- & = & H(U_i \xor V_i|U_1^{i-1},V_1^{i-1},\vecY \concat \vecY',S_0,S_{2N}) \\
        & \geq & H(U_i \xor V_i|U_1^{i-1},V_1^{i-1},\vecY, \vecY',S_0,S_{2N}) \\
        & \geq & H(U_i \xor V_i|U_1^{i-1},V_1^{i-1},\vecY, \vecY',S_0,S_N,S_{2N}) \\
        & = & \sum_{s_0,s_N,s_{2N} \in \calS}  p(s_0,s_N,s_{2N}) \cdot \alpha(s_0,s_N,s_{2N}) \; ,  
    \end{IEEEeqnarray*}

    The proof of (\ref{eq:hatH_beta}) follows by stationarity. That is,
\begin{IEEEeqnarray*}{rCl}
    \hat{H}_n & = & H(U_i |U_1^{i-1},\vecY, S_0,S_N) \\
              & = & \frac{H(U_i |U_1^{i-1},\vecY, S_0,S_N) + H(V_i |V_1^{i-1},\vecY', S_N,S_{2N})}{2} \\
              & = & \sum_{s_0,s_N,s_{2N} \in \calS}  p(s_0,s_N,s_{2N}) \cdot \frac{\gamma(s_0,s_N) + \gamma(s_N,s_{2N})}{2} \\ 
              & = & \sum_{s_0,s_N,s_{2N} \in \calS}  p(s_0,s_N,s_{2N}) \cdot \beta(s_0,s_N,s_{2N}) \; .  
\end{IEEEeqnarray*}

By (\ref{eq:beta}), we deduce that (\ref{eq:alphaGeqBeta}) will follow from proving that
\begin{equation}
    \label{eq:alphaGeqGammaOne}
    \alpha(s_0,s_N,s_{2N}) \geq \gamma(s_0,s_N) 
\end{equation}
and
\begin{equation}
    \label{eq:alphaGeqGammaTwo}
 \alpha(s_0,s_N,s_{2N}) \geq \gamma(s_N,s_{2N})
\end{equation}
W.l.o.g, we prove (\ref{eq:alphaGeqGammaOne}). Indeed, given that $S_N = s_N$, we have by the Markov property that  $(S_0, U_1^{i-1}, U_i, \vecY)$ and $(V_1^{i-1}, V_i,\vecY',S_{2N})$ are independent. Hence, for any $s_{2N}$ we may also write $\gamma$, defined in (\ref{eq:gamma}), as
\begin{aligncc*}
    \gamma(s_0,s_N) =  
    H( U_i| \onlydouble{&} U_1^{i-1},V_1^{i-1},V_i,\vecY,\vecY', \\ \onlydouble{&} S_0=s_0,S_N=s_N,S_{2N} = s_{2N}) \; .
\end{aligncc*}
Lastly, note that in the above expression for $\gamma$, since we condition on $V_i$, we could have written $U_i \xor V_i$ in place of $U_i$. This would give us the expression for $\alpha$ in (\ref{eq:alpha}), up to a further conditioning on $V_i$. Since conditioning reduces entropy, (\ref{eq:alphaGeqGammaOne}) follows. As noted, the proof of (\ref{eq:alphaGeqGammaTwo}) is similar. Hence, we deduce (\ref{eq:alphaGeqBeta}). 
\end{IEEEproof}

In light of Lemma~\ref{lemm:alphaBetaConnection}, our plan is to show the existence of a triplet $(s_0,s_N,s_{2N})$ for which $\alpha(s_0,s_N,s_{2N})$ is substantially greater than $\beta(s_0,s_N,s_{2N})$. The next lemma assures us such a triplet indeed exists.

\begin{lemm}
    \label{lemm:alphaSubstantiallyLargerThanBeta}
    For every $\epsilon > 0$ there exists a $\Delta' = \Delta'(\epsilon)$ for which the following holds. Let $N = 2^n > 2^{\nu}$, where $\nu$ was promised in Lemma~\ref{lemm:allTripletsProbable}. Then, if $\epsilon \leq \hat{H}_n \leq 1- \epsilon$, then there exists a triplet $s_0,s_N,s_{2N}$ such that
    \begin{equation}
        \label{eq:alphaSubstantiallyLargerThanBeta}
        \alpha(s_0,s_N,s_{2N}) > \beta(s_0,s_N,s_{2N}) + \Delta' \; .
    \end{equation}
\end{lemm}
\begin{IEEEproof}
    By definition of $\gamma$ in (\ref{eq:gamma}), we have that 
    \begin{IEEEeqnarray}{rCl}
        \hat{H}_n &=& \sum_{s_0,s_N \in \calS} \Pr(S_0 = s_0, S_N = s_N) \cdot  \gamma(s_0,s_N) \label{eq:HnGamma} \\
                  &=& \sum_{s_N,s_{2N} \in \calS} \Pr(S_N = s_N, S_{2N} = s_{2N}) \cdot  \gamma(s_N,s_{2N}) \; , \IEEEnonumber
    \end{IEEEeqnarray}
    where the second equality follows by stationarity.
    A crucial point will be to show the existence of a triplet $(s_0,s_N,s_{2N})$ for which $(\hat{H}_n - \gamma(s_0,s_N))\cdot(\hat{H}_n - \gamma(s_N,s_{2N})) \leq 0$.
    In other words, either
    \begin{equation}
        \label{eq:gammaMixing}
        \gamma(s_0,s_N) \leq \hat{H}_n \quad \mbox{and} \quad  \gamma(s_N,s_{2N}) \geq \hat{H}_n \; ,
    \end{equation}
    or
    \begin{equation}
  \label{eq:gammaMixing2}
        \gamma(s_0,s_N) \geq \hat{H}_n \quad \mbox{and} \quad  \gamma(s_N,s_{2N}) \leq \hat{H}_n \; .
    \end{equation}
    
    To show this by contradiction, we start by supposing that this is not the case. Then, for all $s_0,s_N,s_{2N} \in \calS$, it must be that 
    \begin{equation}
        \label{eq:bridge}
        (\hat{H}_n - \gamma(s_0,s_N))\cdot(\hat{H}_n - \gamma(s_N,s_{2N})) > 0 \; .
    \end{equation}

    Fix some arbitrary $a,b \in \calS$. By specializing $s_0$ to $a$ and $s_N$ to $b$ in (\ref{eq:bridge}), we deduce that $\hat{H}_n \neq \gamma(a,b)$. Assume w.l.o.g.\ that $\gamma(a,b) < \hat{H}_n$. We now claim that for all $c,d \in \calS$,
    \begin{equation}
        \label{eq:gammacd}
        \gamma(c,d) < \hat{H}_n \; .
    \end{equation}
    Indeed, let $c,d \in \calS$ be given. By setting $s_0 = a$, $s_N = b$, $s_{2N} = c$, we deduce from (\ref{eq:bridge}) that $\gamma(b,c) < \hat{H}_n$. Hence, if we set $s_0 = b$, $s_N = c$, $s_{2N} = d$ in (\ref{eq:bridge}), we deduce (\ref{eq:gammacd}).

    From the above paragraph, we conclude that for all $s_0, s_N \in \calS$, we must have that $\gamma(s_0,s_N) < \hat{H}_n$. However, recalling from (\ref{eq:HnGamma}) that $\hat{H}_n$ is a weighted average of such $\gamma$ terms, we arrive at a contradiction. Hence, there exists a triplet $(s_0,s_N,s_{2N})$ for which either (\ref{eq:gammaMixing}) or (\ref{eq:gammaMixing2}) holds.  This is the triplet we are searching for. Indeed, since we have assumed that $\epsilon \leq \hat{H}_n \leq 1- \epsilon$, the above triplet satisfies
    \[
        \min \{ \gamma(s_0,s_N), \gamma(s_N,s_{2N}) \} \leq 1-\epsilon
    \]
and
    \[
        \max \{ \gamma(s_0,s_N), \gamma(s_N,s_{2N}) \} \geq \epsilon \; .
    \]
    Our result now follows by combining part (i) of \cite[Lemma 2.2]{sasoglu:12b} with\footnote{The first two strict inequalities in the statement of \cite[Lemma 11]{SasogluTal:18a} are essentially typos: they should both be replaced by weak inequalities, as is evident from reading the beginning of the proof.} \cite[Lemma 11]{SasogluTal:18a}.
\end{IEEEproof}

Combining Lemmas~\ref{lemm:alphaBetaConnection} and~\ref{lemm:alphaSubstantiallyLargerThanBeta} gives the following key result.
\begin{lemm}
    \label{lemm:weakPolarizationFAIM}
    For every $\epsilon > 0$ there exists $\Delta = \Delta(\epsilon)$ for which the following holds. Let $N = 2^n > 2^{\nu}$, where $\nu$ was promised in Lemma~\ref{lemm:allTripletsProbable}. Then, if $\epsilon < \hat{H}_n \leq 1- \epsilon$, then 
    \[
\hat{H}_n^- - \hat{H}_n > \Delta(\epsilon)
    \] 
\end{lemm}
\begin{IEEEproof}
    Take
\[
    \Delta =   \frac{\Delta' \cdot (\pimin)^3}{2} \; ,
\]
where $\Delta'$ is as defined in Lemma~\ref{lemm:alphaSubstantiallyLargerThanBeta}. Now, simply combine (\ref{eq:allStatesProbable}), (\ref{eq:hatH_alpha}), (\ref{eq:hatH_beta}), (\ref{eq:alphaGeqBeta}) and the existence of triplet $s_0,s_N,s_{2N}$ for which (\ref{eq:alphaSubstantiallyLargerThanBeta}) holds, to yield the claim.
\end{IEEEproof}

The following lemma will be useful.
\begin{lemm} \label{lem:order_rv_convergence}
For $n\in \mathbb{N}$, let $A_n$ and $B_n$ be real random variables defined on a common probability space.
Suppose $B_n$ converges in $L^1$ to $B_\infty$ and $E(A_n)$ converges to $E(B_\infty)$.
If $A_n \geq B_n$ for all $n\in \mathbb{N}$, then $A_n$ converges in $L^1$ to $B_\infty$.
\end{lemm}
\begin{IEEEproof}
By definition, $B_n$ converges to $B_\infty$ in $L^1$ if and only if $E(|B_n-B_\infty|) \to 0$.
Thus, by the triangle inequality,
\begin{align*}
E(|A_n - B_\infty|) &\leq E(|A_n - B_n|) + E(|B_n - B_\infty|) \\
& = E(A_n - B_n) + E(|B_n - B_\infty |) \\
& = E(A_n) - E(B_n) + E(|B_n - B_\infty|).
\end{align*}
In the limit, the first two terms converge to $E(B_\infty)$ and the last term converges to $0$.
Thus, $E(|A_n - B_\infty|) \to 0$.
\end{IEEEproof}

The following theorem claims weak polarization for the three cases discussed earlier.
\begin{theo}
    \label{theo:slowPolarization}
    Fix $\epsilon \in (0,1)$ and let $N = 2^n$. For a given hidden-Markov input distribution, let $\vecX = (X_1,X_2,\ldots,X_N)$ be a random vector with the above distribution. Let $\vecY$ be the result of passing $\vecX$ through a deletion channel with deletion probability $\delta$. Denote $\vecU = \calA(\vecX)$. Let $S_0$ and $S_N$ be as in Definition~\ref{def:FAIM}. Then, 
\begin{IEEEeqnarray}{rCl}
    \IEEEyesnumber\label{eq:SlowLowEntropy}\IEEEyessubnumber*
    && \lim_{n \to \infty} \frac{\mysize{\myset{i : H(U_i | U_1^{i-1},\vecY,S_0,S_N) < \epsilon}}}{N} \label{eq:slowLowEntropy_stateInformed} \\
    & = &  \lim_{n \to \infty} \frac{\mysize{\myset{i : H(U_i | U_1^{i-1},\vecY) < \epsilon}}}{N} \label{eq:slowLowEntropy_regular} \\
    & = & \lim_{n \to \infty} \frac{\mysize{\myset{i : H(U_i | U_1^{i-1},\vecY^\DPT) < \epsilon}}}{N} \label{eq:slowLowEntropy_TDC}\\
    & = & 1 - \lim_{n \to \infty} \frac{H(\vecX|\vecY)}{N} \label{eq:slowLowEntropy_target}
\end{IEEEeqnarray}
and
\begin{IEEEeqnarray}{rCl}
    \IEEEyesnumber\label{eq:SlowHighEntropy}\IEEEyessubnumber*
    && \lim_{n \to \infty} \frac{\mysize{\myset{i : H(U_i | U_1^{i-1},\vecY,S_0,S_N) > 1- \epsilon}}}{N} \label{eq:slowHighEntropy_stateInformed}\\
    & = &  \lim_{n \to \infty} \frac{\mysize{\myset{i : H(U_i | U_1^{i-1},\vecY) > 1- \epsilon}}}{N} \label{eq:slowHighEntropy_regular} \\
    & = & \lim_{n \to \infty} \frac{\mysize{\myset{i : H(U_i | U_1^{i-1},\vecY^\DPT) > 1- \epsilon}}}{N} \label{eq:slowHighEntropy_TDC} \\
    & = & \lim_{n \to \infty} \frac{H(\vecX|\vecY)}{N} \label{eq:slowHighEntropy_target}
\end{IEEEeqnarray}
\end{theo}

\begin{IEEEproof}
    For simplicity, the proof is split into 4 parts.

    \paragraph*{Part I: (\ref{eq:slowLowEntropy_target}) and (\ref{eq:slowHighEntropy_target}) are well defined}
    Recall from Lemma~\ref{lemm:twoLimitsExist} that $\lim_{n \to \infty} H(\vecX|\vecY)/N$ exists. Thus, the right hand sides of both (\ref{eq:slowLowEntropy_target}) and (\ref{eq:slowHighEntropy_target}) are well defined.
    
    \paragraph*{Part II: (\ref{eq:slowLowEntropy_stateInformed})$=$(\ref{eq:slowLowEntropy_target}) and (\ref{eq:slowHighEntropy_stateInformed})$=$(\ref{eq:slowHighEntropy_target})}

    Since the \Arikan\ transform is invertible, it follows that $\hat{\calH}_N = H(\vecX|\vecY,S_0,S_N) = H(\vecU|\vecY,S_0,S_N)$, where $\hat{\calH}_N$ is defined in (\ref{eq:Hhat}). Thus, from the chain rule for entropy, we observe that
    \begin{align*} E(\hat{H}_n) &= \frac{1}{N} \sum_{i=1}^N H(U_i | U_1^{i-1},\vecY,S_0,S_N) \\ &= \frac{1}{N} H(\vecU|\vecY,S_0,S_N) \\ &= \frac{1}{N} \hat{\calH}_N .
    \end{align*}    
    From Theorem~\ref{thm:hatH_polarizes}, we see that $\hat{H}_n$ converges in $L^1$ to $\hat{H}_\infty \in \{0,1\}$.
    This implies that $E(\hat{H}_\infty) = \lim_{n\to\infty} E(\hat{H}_n)$ which exists and equals $\lim_{N \to \infty} \hat{\calH}_N/N$ by Lemma~\ref{lemm:HHatLimitExists}.
    Since $\hat{H}_\infty \in \{0,1\}$, observing that $E(\hat{H}_\infty) = \Pr (\hat{H}_\infty = 1) $ shows that
    \begin{align*}
	\eqref{eq:slowHighEntropy_stateInformed} &= \lim_{n \to \infty} \Pr (\hat{H}_n > 1-\epsilon) = \Pr (\hat{H}_\infty = 1) = \lim_{n \to \infty} \frac{1}{N} \hat{\calH}_N,
	\end{align*}
    where the second equality holds because convergence in $L^1$ implies convergence in distribution and $1-\epsilon$ is a continuity point of $\Pr(\hat{H}_\infty \leq x)$~\cite[Ch.~4]{Durrett-2019}.
    Since Lemma~\ref{lemm:HatNoHatLimitsEqual} shows that $\lim_{N \to \infty} \hat{\calH}_N/N$ equals (\ref{eq:slowHighEntropy_target}), it follows that  (\ref{eq:slowHighEntropy_stateInformed}) equals (\ref{eq:slowHighEntropy_target}).
    The last step is observing that 
    \begin{align*}
        \eqref{eq:slowLowEntropy_stateInformed} &= \lim_{n \to \infty} \Pr (\hat{H}_n < \epsilon) = \Pr (\hat{H}_\infty = 0)=1-\Pr (\hat{H}_\infty = 1)         
    \end{align*}
    holds because convergence in $L^1$ implies convergence in distribution and $\epsilon$ is a continuity point of $\Pr(\hat{H}_\infty \leq x)$.   
    Thus,  (\ref{eq:slowLowEntropy_stateInformed}) equals (\ref{eq:slowLowEntropy_target}).

    \paragraph*{Part III: (\ref{eq:slowLowEntropy_TDC})$=$(\ref{eq:slowLowEntropy_target}) and (\ref{eq:slowHighEntropy_TDC})$=$(\ref{eq:slowHighEntropy_target})}
    To prove these equalities, we will apply Lemma~\ref{lem:order_rv_convergence} to the sequences $A_n = H_n^*$ and $B_n = \hat{H}_n$.
    Theorem~\ref{thm:hatH_polarizes} shows that $\hat{H}_n$ converges in $L^1$ to $\hat{H}_\infty$ and we established in the previous part that $E(\hat{H}_\infty)$ equals (\ref{eq:slowHighEntropy_target}).  
    From the definitions in~\eqref{eq:hatHn} and~\eqref{eq:HnTDC}, it follows that $H_n^* \geq \hat{H}_n$ for all $n \in \mathbb{N}$.
    The only other element required for Lemma~\ref{lem:order_rv_convergence} is that $E(H_n^*) \to E(\hat{H}_\infty)$ and this will be shown below.
    Assuming this for now, we observe Lemma~\ref{lem:order_rv_convergence} implies that $H_n^*$ converges in $L^1$ to $\hat{H}_\infty$ and gives the desired result
    \begin{align*}
    \eqref{eq:slowLowEntropy_TDC} &= \lim_{n\to\infty} \Pr(H_n^* < \epsilon) = \Pr(\hat{H}_{\infty} < \epsilon) = \eqref{eq:slowLowEntropy_target} \\
    \eqref{eq:slowHighEntropy_TDC} &= \lim_{n\to\infty} \Pr(H_n^* > 1- \epsilon) = \Pr(\hat{H}_{\infty} > 1- \epsilon) = \eqref{eq:slowHighEntropy_target},
    \end{align*}
    where the second equality on each line holds because convergence in $L^1$ implies convergence in distribution and $\epsilon,1-\epsilon$ are continuity points of $\Pr(\hat{H}_\infty \leq x)$~\cite[Ch.~4]{Durrett-2019}.   

    To show that $E(H_n^*) \to E(\hat{H}_\infty)$, we will use the fact that
    \begin{multlinecc}
        \label{eq:HUGivenYTDC_sandwich}
        H(\vecU | \vecY,S_0,S_N) \leq H(\vecU | \vecY^\DPT) \leq \\
        H( \vecU | \vecY,S_0,S_N) + 2 \log_2 |\calS| + 2 \log_2 (N+1) \; .
    \end{multlinecc}
    Indeed, the first inequality holds because $\vecY^\DPT$ is a function of $\vecY$. The second inequality follows from first noting that
    \[
         H(\vecU|\vecY^\DPT) \leq H(\vecY,S_0,S_N,\vecU |\vecY^\DPT) \; .  \\
    \]
    And then observing that
    \begin{IEEEeqnarray*}{rCl}
        \IEEEeqnarraymulticol{3}{l}{H(\vecY,S_0,S_N,\vecU |\vecY^\DPT)} \\
        \quad &=& H(\vecY | \vecY^\DPT) + H(S_0,S_N|\vecY,\vecY^\DPT ) + H(\vecU|\vecY,\vecY^\DPT,S_0,S_N) \\
        \quad &\eqann{a}& H(\vecY | \vecY^\DPT) + H(S_0,S_N|\vecY,\vecY^\DPT ) + H(\vecU|\vecY,S_0,S_N) \\
        \quad & \eqann[\leq]{b} & H(\vecY | \vecY^\DPT) + 2 \log_2 |\calS| + H(\vecU|\vecY,S_0,S_N) \\
        \quad & \eqann[\leq]{c} & 2 \log_2(N+1) + 2 \log_2 |\calS| + H(\vecU|\vecY,S_0,S_N) \; ,
    \end{IEEEeqnarray*}
    where \eqannref{a} follows from $\vecY^\DPT$ being a function of $\vecY$, \eqannref{b} follows by $S_0$ and $S_N$ each having a support of size $|\calS|$, and \eqannref{c} follows since in order to construct $\vecY$ from $\vecY^\DPT$, it suffices to be told how many `$0$' symbols have been trimmed from each side of $\vecY$, and both numbers are always between $0$ and $N$. Combining the above two displayed equations yields the RHS of (\ref{eq:HUGivenYTDC_sandwich}).

    Finally, we divide both sides of (\ref{eq:HUGivenYTDC_sandwich}) by $N$ and take the limit as $N\to\infty$.
    Since the left-most and right-most terms converge to $E(\hat{H}_\infty)$, the sandwich property implies that the center term, $E(H_n^*)$ also converges to this quantity.

    \paragraph*{Part IV: (\ref{eq:slowLowEntropy_stateInformed})$=$(\ref{eq:slowLowEntropy_regular})$=$(\ref{eq:slowLowEntropy_TDC}) and (\ref{eq:slowHighEntropy_stateInformed})$=$(\ref{eq:slowHighEntropy_regular})$=$(\ref{eq:slowHighEntropy_TDC})}
    Note that, for $1 \leq i \leq N$, we have
    \[
        H(U_i | U_1^{i-1},\vecY,S_0,S_N) \leq
        H(U_i | U_1^{i-1},\vecY) \leq 
        H(U_i | U_1^{i-1},\vecY^\DPT) \; .
    \]
    We have already proved that (\ref{eq:slowLowEntropy_stateInformed})$=$(\ref{eq:slowLowEntropy_TDC}) and (\ref{eq:slowHighEntropy_stateInformed})$=$(\ref{eq:slowHighEntropy_TDC}).
    Thus, by the sandwich property, (\ref{eq:slowLowEntropy_stateInformed})$=$(\ref{eq:slowLowEntropy_regular})$=$(\ref{eq:slowLowEntropy_TDC}) and (\ref{eq:slowHighEntropy_stateInformed})$=$(\ref{eq:slowHighEntropy_regular})$=$(\ref{eq:slowHighEntropy_TDC}).
\end{IEEEproof}

\section{Strong polarization}
\label{sec:strong}

To rigorously claim a coding scheme for the deletion channel, one must also show strong polarization. For this, Theorem~\ref{theo:slowPolarization} is not sufficient and, so far, we have been unable to prove strong polarization for the \emph{standard} polar code construction.
Thus, we will modify the standard coding scheme to proceed.

\subsection{Overview of Coding Scheme}
\label{sec:strong_setup}
Fix a deletion probability $\delta$ and a regular hidden Markov input distribution. Recall that our goal is to achieve the information rate $\calI$ given in (\ref{eq:informationRate}). For didactic reasons, we first consider a simplified setting in which this goal is easily attained. Specifically, let $N_0$ be given parameter, and consider a block-TDC with block length $N_0$ and deletion probability $\delta$. That is, for each input block $\vecX(\phi)$ of length $N_0$, where $\phi = 1, 2, \ldots$, the channel outputs $\vecY^\DPT(\phi)$, which is the result of passing $\vecX(\phi)$ through a TDC with deletion probability $\delta$. The crucial point to note is that, contrary to a deletion channel, the output of a block-TDC \emph{contains commas between segments}. That is, we know exactly which output segment corresponds to which input block.

How would one code for such a channel and achieve a rate approaching $\calI$? For this, we will assume that
\begin{equation}
    \label{eq:N0}
    N_0=2^{n_0} \; ,
\end{equation}
and that we can choose $N_0$ to be arbitrarily large. Let
\begin{equation}
    \label{eq:Phi}
    \Phi = 2^{n_1}
\end{equation}
be the number of blocks we will transmit through the channel. Consider the following input distribution: each block $\vecX(\phi)$ will be distributed according to the input distribution that we have fixed at the start of this subsection, and the input blocks $\vecX(1), \vecX(2), \ldots, \vecX(\Phi)$ will be \emph{i.i.d.} In a nutshell, this suffices to achieve a coding rate of $\calI$ with vanishing probability of error for the following two reasons.
First, Theorem~\ref{theo:slowPolarization} shows weak polarization for each block and, in each block, we have the required fractions of high-entropy/low-entropy indices.  Second, the independence between blocks implies that strong polarization will occur.

We now back the above claim with a few more details. We denote the output of the encoder --- the concatenation of the above blocks --- by
\begin{equation}
    \label{eq:vecx}
    \vecX = \vecX(1) \concat \vecX(2) \concat \cdots \concat \vecX(\Phi) \; .
\end{equation}
This output has length
\begin{equation}
    \label{eq:NN0Phi}
    N = N_0 \cdot \Phi = 2^{n_0+n_1} = 2^n \; .
\end{equation}
We will use a sans-serif font to denote a vector whose elements are `blocks'. Thus, we will denote the partitioning of the above $\vecX$ into blocks of length $N_0$ by 
\begin{equation}
    \label{eq:bigvecx}
    \bigvecX = (\vecX(1), \vecX(2), \ldots, \vecX(\Phi)) \; .
\end{equation}
The corresponding output of the block-TDC is denoted 
\begin{equation}
    \label{eq:bigVecYDPT}
    \bigvecY^\DPT = (\vecY^\DPT(1),\vecY^\DPT(2),\ldots,\vecY^\DPT(\Phi)) \; .
\end{equation}
That is, $\bigvecY^\DPT$ is comprised of $\Phi$ distinguishable blocks --- it is \emph{not} simply the concatenation of the $\vecY^\DPT(\phi)$. The superscript `$\DPT$' in $\bigvecY^\DPT$ suggest that trimming operation is applied \emph{blockwise}.

We first consider the polar transform of $\vecX(\phi)$, denoted\footnote{We reserve the letter $U$, commonly used to denote the result of a polar transform, for a related yet distinct definition that is yet to appear.}
\begin{equation}
    \label{eq:Vphi}
    \vecV(\phi) = \calA(\vecX(\phi)) \; ,
\end{equation}
 where $1 \leq \phi \leq \Phi$. Note that $\vecV(\phi)$ is a binary vector of length $N_0$,
\[
    \vecV(\phi) = (V_1(\phi),V_2(\phi),\ldots,V_{N_0}(\phi)) \; .
\]

Recall that $\vecY^\DPT(\phi)$ is the output corresponding to $\vecX(\phi)$, and note that since we have assumed that the $\vecX(\phi)$ are i.i.d., then this must also hold for triplets $(\vecX(\phi),\vecV(\phi),\vecY^\DPT(\phi))$, when ranging over $1 \leq \phi \leq \Phi$.

For a fixed $1 \leq \phi \leq \Phi$ and a given $1 \leq i_0 \leq N_0$, consider the pair of entropies
\begin{multlinecc}
    \label{eq:twoIdealizedEntropies}
    H(V_{i_0}(\phi)|V_1^{i_0-1}(\phi) ,\vecY^\DPT(\phi)) \;\; \mbox{and} \onlysingle{\;\;\quad} \\
    H(V_{i_0}(\phi)|V_1^{i_0-1}(\phi))\; .
\end{multlinecc}
We now make two important observations. First, since we have already established that the $(\vecX(\phi),\vecV(\phi),\vecY^\DPT(\phi))$ are i.i.d.\ over $\phi$, we deduce that (\ref{eq:twoIdealizedEntropies}) is independent of $\phi$. Second, both entropies in (\ref{eq:twoIdealizedEntropies}) exhibit slow polarization, in the sense of Theorem~\ref{theo:slowPolarization}. That is, on one hand, we deduce that (\ref{eq:slowLowEntropy_TDC})$=$(\ref{eq:slowLowEntropy_target}) and (\ref{eq:slowHighEntropy_TDC})$=$(\ref{eq:slowHighEntropy_target}), if in both (\ref{eq:slowLowEntropy_TDC}) and (\ref{eq:slowHighEntropy_TDC}) we replace $U_i$, $U_1^{i-1}$, $\vecY^\DPT$, $n$ and $N$ by $V_{i_0}(\phi)$, $V_1^{i_0-1}(\phi)$, $\vecY^\DPT(\phi)$, $n_0$ and $N_0$, respectively. These statements hold for all $\delta\in [0,1]$.  For the special case of $\delta=1$, one gets a degenerate channel where $\vecY^\DPT(\phi)$ always equals the empty string. Thus, on the other hand, the same claim of (\ref{eq:slowLowEntropy_TDC})$=$(\ref{eq:slowLowEntropy_target}) and (\ref{eq:slowHighEntropy_TDC})$=$(\ref{eq:slowHighEntropy_target}), under the above substitutions continues to hold, with $\vecY$ and $\vecY^\DPT$ removed from these equations.

Since the first entropy in (\ref{eq:twoIdealizedEntropies}) is always less than or equal to the second, we deduce from the above paragraph and the first half of Theorem~\ref{theo:slowPolarization} that for $\epsilon \in (0,1)$ fixed, the fraction of indices $i_0$ for which
\begin{multlinecc*}
    H(V_{i_0}(\phi)|V_1^{i_0-1}(\phi) ,\vecY^\DPT(\phi)) < \epsilon \quad \mbox{and} \onlysingle{\quad\quad} \\
    H(V_{i_0}(\phi)|V_1^{i_0-1}(\phi)) \geq \epsilon
\end{multlinecc*}
tends to
\begin{multlinecc*}
    \left( 1 - \lim_{n_0 \to \infty} \frac{H(\vecX(\phi)|\vecY(\phi)))}{N_0} \right) \\
    - \left( 1 - \lim_{n_0 \to \infty} \frac{H(\vecX(\phi))}{N_0} \right) = \calI \; ,
\end{multlinecc*}
as $n_0 \to \infty$. For simplicity of exposition, let us further restrict $\epsilon$ to $\epsilon \in (0,1/2)$.  By both halves of Theorem~\ref{theo:slowPolarization}, we deduce that the fraction of indices $i_0$ for which
\[
    \epsilon \leq H(V_{i_0}(\phi)|V_1^{i_0-1}(\phi)) \leq 1 - \epsilon
\]
vanishes. The conclusion is stated as a lemma, for future reference.
\begin{lemm}
    \label{lemm:slowPolarizationStar}
    For $\epsilon \in (0,1/2)$ fixed, the fraction of indices $1 \leq i_0 \leq N_0$ for which
\begin{multlinecc}
    \label{eq:goodIndexI}
    H(V_{i_0}(\phi)|V_1^{i_0-1}(\phi) ,\vecY^\DPT(\phi)) < \epsilon \quad \mbox{and} \onlysingle{\quad\quad} \\
    H(V_{i_0}(\phi)|V_1^{i_0-1}(\phi)) > 1 - \epsilon
\end{multlinecc}
tends to $\calI$, as $n_0 \to \infty$, and is the same for every $1 \leq \phi \leq \Phi$.
\end{lemm}

We now note that for a given $\phi$ and $i_0$, we have an efficient method of calculating the probabilities corresponding to (\ref{eq:goodIndexI}). Namely, this is achieved by using the base trellis defined for a TDC in Subsection~\ref{subsec:TrellisTDC}, applying a series of plus and minus polarization operations on it, according to the binary representation of $i_0-1$, and then invoking (\ref{eq:trellisConditionalProbCalc}). That is, the only thing stopping us from applying the Honda-Yamamoto scheme \cite{Honda_Yamamoto_2013} at this point is the fact that the above $\epsilon$ is fixed.

Informally, we overcome the above problem as follows. Take $\epsilon$ `small' and $n_0$ as well as $n_1$ `large'. Consider a `good' index $i_0$. That is, an index $i_0$ for which (\ref{eq:goodIndexI}) holds. This will be the case for a fraction of indices `very close' to $\calI$. Next, recall the definition of $\vecX$ in (\ref{eq:vecx}), and denote its polar transform as \[
    \vecU = \calA(\vecX) \; .
\]

Consider the subvector $U_{(i_0-1) \cdot \Phi + 1}^{i_0 \cdot \Phi}$. It is not hard to prove that
\begin{equation}
    \label{eq:vecVSlice}
    U_{(i_0-1) \cdot \Phi + 1}^{i_0 \cdot \Phi} = \calA((V_{i_0}(1),V_{i_0}(2),\ldots,V_{i_0}(\Phi))) \; .
\end{equation}
That is, the LHS of (\ref{eq:vecVSlice}) is gotten by applying the \Arikan\ transform to the vector $(V_{i_0}(1),V_{i_0}(2),\ldots,V_{i_0}(\Phi))$. Since each entry of this vector satisfies (\ref{eq:goodIndexI}), `almost all' indices $i$ of $\vecU$, where $(i_0-1) \cdot \Phi + 1 \leq i \leq i_0 \cdot \Phi$ are strongly polarized. That is, satisfy
\begin{multlinecc}
    \label{eq:goodIndexJ}
    Z(U_i|U_1^{i-1} ,\bigvecY^\DPT) < 2^{-n_1 \beta} \quad \mbox{and} \onlysingle{\quad\quad} \\
    K(U_i|U_1^{i-1}) < 2^{-n_1 \beta}
\end{multlinecc}
where $Z$ and $K$ are the conditional Bhattacharyya parameter 
and the conditional total variation (see Definitions~\ref{defi:Z} and \ref{defi:K} in Appendix~\ref{sec:ZK})
, $\beta < 1/2$ is some fixed constant, and $\bigvecY^\DPT$ is the block-TDC output vector defined in (\ref{eq:bigVecYDPT}). That is, the overall fraction of useful indices $1 \leq i \leq N_0 \Phi$ with respect to the Honda-Yamamoto scheme will be `very close' to $\calI$, and the error of the scheme will approach $0$ at a rate of roughly $2^{-\sqrt{N_1}}$.

The reader may not be surprised to learn that the above informal statements can be made rigorous and proven\footnote{Such a proof is not a straightforward adaptation of the ideas in \cite{Arikan_2009} and \cite{Arikan_Telatar_2009}. Namely, it requires the use of \cite[Lemma 40]{ShuvalTal:18a}, which we indeed invoke in the proof of Theorem~\ref{theo:main}.}. Indeed, this will be done as part of the proof of Theorem~\ref{theo:main}. However, one important point remains to be addressed. That is, the channel we will in fact be coding for is the deletion channel, and \emph{not} the block-TDC. Hence, in the above description, we have implicitly assumed a genie which has manufactured the punctuated vector $\bigvecY^\DPT$ for us. The purpose of the guard-bands, defined shortly, is to approximate such a genie in practice.

Our actual coding scheme will be as follows. For the encoding step, we will first use the Honda-Yamamoto scheme with respect to the block-TDC. I.e., the information bits will be placed in indices $j$ of $\vecU$ for which (\ref{eq:goodIndexJ}) holds. The resulting codeword will be $\vecX$. Then, we will add to $\vecX$ runs of `$0$' symbols in key locations, and transmit the resulting word (which will be longer than $\vecX$) on the deletion channel. On the decoder side, a preliminary step will be to deduce the punctuated vector $\bigvecY^\DPT$ from the received vector $\vecY$. That is, we will remove the guard bands (and trim the $\vecY(\phi)$ into $\vecY^\DPT(\phi)$ in the process), thus producing $\bigvecY^\DPT$. Then, the decoder will be applied on $\bigvecY^\DPT$ to yield $\vecU$, and thus the information bits.

\subsection{Guard bands}
In this subsection, we first describe how the guard bands are added to $\vecX$ on the encoder side. We then explain how the decoder deduces the punctuated vector $\bigvecY^\DPT$ from the received vector $\vecY$. 

\begin{figure}
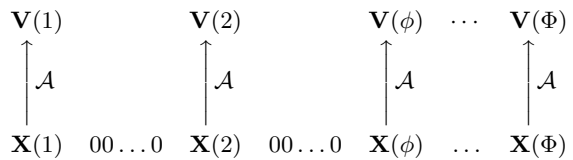

\small
    \begin{equation*}
        \begin{array}{cccccccc}
            \vecV(1) & & \vecV(2) & & \vecV(\phi) & \cdots & \vecV(\Phi) \\
            \xuparrow{0.7cm} \calA & & \xuparrow{0.7cm} \calA & & \xuparrow{0.7cm} \calA && \xuparrow{0.7cm} \calA \\[5mm]
            \vecX(1) & 00\ldots 0 & \vecX(2) & 00\ldots 0 & \vecX(\phi) & \ldots & \vecX(\Phi)
        \end{array}
    \end{equation*}
\normalsize

\caption{The $\Phi = N/N_0$ blocks, denoted $\vecX(1), \vecX(2), \ldots, \vecX(\Phi)$, have length $N_0=2^{n_0}$, are i.i.d., and each is distributed according to the regular hidden-Markov input distribution. Their polar transforms are $\vecV(1), \vecV(2), \ldots, \vecV(\Phi)$. An additional $n-n_0$ polarization steps (not shown) will be applied to $\vecV(1),\vecV(2),\ldots,\vecV(\Phi)$, resulting in $\vecU$. The transmitted codeword is gotten by separating consecutive $\vecX(\cdot)$ vectors by a `guard band'. That is, by a string of `$0$' symbols. The length of the guard bands is not constant. For example, the middle guard band is always the longest, while the first and last guard bands are always the shortest.}
\label{fig:manyGuardBands}
\end{figure}

We start by defining how guard bands are added between the blocks $\vecX(1),\vecX(2),\ldots,\vecX(\Phi)$, see Figure~\ref{fig:manyGuardBands}. That is, we define how $\vecX$ is transformed into $g(\vecX)$. This is done in a simple recursive manner. Informally, let $\vecx$ be a vector of length $2^n$. If this length is greater than the designated block-length $N_0$, we halve $\vecx$, add $\ell_n$ `$0$' symbols in the middle, and then apply $g$ recursively to each original half. Namely, for $\vecx = \xI \concat \xII \in \calX^{2^n}$ with
\[
    \xI = x_1^{2^{n-1}} \in \calX^{2^{n-1}} \; , \quad  \xII = x_{2^{n-1}+1}^{2^n} \in \calX^{2^{n-1}} 
\]
being the first and second halves of $\vecx$, respectively, we define
\begin{IEEEeqnarray}{rCl}
    \label{eq:guardBand}
    g(\vecx) &\triangleq &
    \begin{cases}
    \vecx & \text{if }n \leq n_0
    \\
        g(\xI) \concat \overbrace{00 \ldots 0}^{\ell_n} \concat g(\xII) & \text{if }n> n_0, \\
    \end{cases} \\
    \noalign{\noindent and\vspace{2\jot}}
    \ell_n & \triangleq & \floor{2^{(1-\xi) (n-1)}}, \label{eq:ln}
\end{IEEEeqnarray}
where $\xi \in (0,1/2)$ is a yet-to-be-specified `small' constant. The parameter $\xi$ controls the rate penalty of adding guard bands, on one hand, and the probability of the decoder successfully removing the guard bands, on the other hand.  We will require that $n_0 > 1$, so that the inequality
\begin{equation}
    \label{eq:lnSimpleInequality}
\ell_n > 2^{(n-1)(1-\xi)-1} 
\end{equation}
used later on will hold for all relevant $n$, i.e., for $n > n_0$. Note the above specifically implies that $\ell_n > 0$.

We now explain how the guard bands are removed, from the received word $\vecY$, in order to produce the punctuated sequence $\bigvecY^\DPT$ defined in (\ref{eq:bigVecYDPT}). Equivalently, we now show a procedure with the following outcome: for each block index $1 \leq \phi < \Phi$, we will produce the trimmed vector $\vecY^*(\phi)$ corresponding to the block $\vecX(\phi)$. Before explaining how this is done, we first mention that our method has a small yet non-zero probability of failing. That is, there is a non-zero probability that our method will fail to produce $\bigvecY^\DPT$.  This probability will be analyzed at a later stage.

    \begin{figure}
        \begin{center}
        \begin{tikzpicture}[scale=1.3]
        \node at (-3.15,4.75) {$\vecX$};
        \draw [very thick,draw=black] (-1.3,5) rectangle (0,4.5) node[pos=.5] {$\XI$};
        \draw [very thick,draw=black] (0,5) rectangle (1.3,4.5) node[pos=.5] {$\XII$};
        
        \node at (-3.15,3.75) {$\vecG$};
        \draw [very thick,draw=black] (-2.5,4) rectangle (-0.5,3.5) node[pos=.5] {$\GI$};
        \draw [very thick,draw=black] (-0.5,4) rectangle (0.5,3.5) node[pos=.5] {$\Gstar$};
        \draw [very thick,draw=black] (0.5,4) rectangle (2.5,3.5) node[pos=.5] {$\GII$};
        
        \draw[dashed] (-1.3,4.5) -- (-2.5,4);
        \draw[dashed] (1.3,4.5) -- (2.5,4);
        \draw[dashed] (0,4.5) -- (-0.5,4);
        \draw[dashed] (0,4.5) -- (0.5,4);
        
        \node at (-3.15,2.75) {$\vecY$};
        \draw [very thick,draw=black] (-2.2,3) rectangle (-0.60,2.5) node[pos=.5] {$\YI$};
        \draw [very thick,draw=black] (-0.60,3) rectangle (0.30,2.5) node[pos=.5] {$\Ystar$};
        \draw [very thick,draw=black] (0.30,3) rectangle (2.2,2.5) node[pos=.5] {$\YII$};
        
        \draw[dashed] (-2.5,3.5) -- (-2.2,3);
        \draw[dashed] (2.5,3.5) -- (2.2,3);
        \draw[dashed] (-0.5,3.5) -- (-0.60,3);
        \draw[dashed] (0.5,3.5) -- (0.30,3);
        
        \node at (-3.15,1.75) {$\vecZ =\vecY^*$};
        \draw [very thick,draw=black] (-2.0,2) rectangle (-0.60,1.5) node[pos=.5] {$\ZI$};
        \draw [very thick,draw=black] (-0.60,2) rectangle (0.30,1.5) node[pos=.5] {$\Zstar$};
        \draw [very thick,draw=black] (0.30,2) rectangle (2.0,1.5) node[pos=.5] {$\ZII$};
        
        \draw[dashed] (-2.2,2.5) -- (-2.0,2);
        \draw[dashed] (2.2,2.5) -- (2.0,2);
        \draw[dashed] (-0.60,2.5) -- (-0.6,2);
        \draw[dashed] (0.30,2.5) -- (0.3,2);
        
        \end{tikzpicture}
        \end{center}
        \caption{The random variables $\vecX$, $\vecG$, $\vecY$, and $\vecZ$.}
        \label{fig:XGYZ}
    \end{figure}

    Our procedure for producing $\bigvecY^\DPT$ will have a preliminary step, and will then involve a recursion. The preliminary step is simple: we trim the received vector $\vecY$ of leading and trailing zeros to produce $\vecY^*$. We stress that, generally, $\vecY^\DPT$ does \emph{not} equal the punctuated sequence $\bigvecY^\DPT$ defined in (\ref{eq:bigVecYDPT}).  In order to introduce notation required later on, let us now define the above operation more verbosely. Let $\XI$ and $\XII$ be the left and right halves of $\vecX$, see Figure~\ref{fig:XGYZ}. Thus, the transmitted word is $g(\vecX) = \GI \concat \Gstar \concat \GII$, where $\GI = g(\XI)$, $\GII = g(\XII)$, and $\Gstar$ is the middle guard band of length $\ell_n$, where $n$ is $\log_2$ of the length of $\vecX$. Clearly, $\GI$ and $\GII$ are of equal length. Denote the parts of $\vecY$ corresponding to $\GI$, $\Gstar$ and $\GII$ by $\YI$, $\Ystar$, and $\YII$, respectively. Note that at this stage, the decoder sees $\vecY$, but can only make an informed guess as to what parts of $\vecY$ constitute $\YI$, $\Ystar$, and $\YII$. We remove from the received word $\vecY$ all leading and trailing `$0$' symbols and denote the resulting vector $\vecZ = \vecY^*$. We denote the parts of $\vecZ$ corresponding to $\YI$, $\Ystar$, and $\YII$ by $\ZI$, $\Zstar$, and $\ZII$, respectively. In order to build up the reader's intuition, we note that in a `typical case', $\ZI$ is $\YI$ after the leading zeros have been removed, $\ZII$ is $\YII$ after the trailing zeros have been removed, and $\Zstar$ is simply $\Ystar$. As explained, the production of $\vecZ$ from $\vecY$ constitutes the preliminary step of our method.

We will now specify how the punctuated vector $\bigvecY^\DPT$ is recursively produced from $\vecZ$. For the base case, note that if $\Phi=1$, then $\bigvecY^\DPT$ is simply $\vecZ$.  Our procedure hinges on the assumption that the middle index of $\vecZ$ originated from a guard band symbol. Specifically, we will assume that the middle index of $\vecZ$ (rounding down) belongs to $\Zstar$. As explained, there is a probability of this assumption being false, and this will be analyzed at a later stage. For now, consider the case in which the assumption holds. In this case, the crucial observation is that $\YI^\DPT$ equals the first half of $\vecZ$, trimmed, while $\YII^\DPT$ equals the second half of $\vecZ$, trimmed. Namely, if $\zeta$ is the length of $\vecZ$, then
\begin{IEEEeqnarray}{rCl}
    \YI^\DPT &=& (Z_1,Z_2,\ldots, Z_{\lfloor{\zeta/2}\rfloor})^\DPT \; , \label{eq:YIFromZ} \\
    \YII^\DPT &=& (Z_{\lfloor{\zeta/2}\rfloor+1} , Z_{\lfloor{\zeta/2}\rfloor+1}, \ldots, Z_\zeta)^\DPT \label{eq:YIIFromZ}\; ,
\end{IEEEeqnarray}
since the guard band $\Zstar$ has been `trimmed out'.  Thus, we have reduced our original problem of producing $\bigvecY^\DPT$ from $\vecY^\DPT$ into two equivalent problems, each half the size of the original: find the first half of $\bigvecY^\DPT$, namely $\vecY^\DPT(1), \vecY^\DPT(2),\ldots,\vecY^\DPT(\Phi/2)$, from $\YI^*$ and the second half of $\bigvecY^\DPT$ from $\YII^*$. Thus, we continue recursively: we apply our method first to the RHS (\ref{eq:YIFromZ}) and then to the RHS of (\ref{eq:YIIFromZ}). If, during all these recursive invocations, our assumptions on the middle index being part of the middle guard band were indeed correct, then we will have succeeded in producing $\bigvecY^\DPT$. Note that the recursion depth is $n-n_0$.

There are two points that must be addressed. First, recall that adding guard bands makes the transmitted word longer. We must show that this has a vanishingly small effect on the rate of our scheme. Second, we must show that our scheme of producing $\bigvecY^\DPT$ from $\vecY$ has a vanishingly small probability of failing. Once this is done, the proof of Theorem~\ref{theo:main} will follow easily.

\subsection{Auxiliary lemmas}
\label{sec:auxiliaryLemmas}

In this section, we state and prove a number of lemmas key to the proof of Theorem~\ref{theo:main}.

In the sequel, we will choose a fixed $\nu \in (0,\frac{1}{3}]$ and set $n_0 = \lfloor \nu n \rfloor$.
The parameter $\nu$ will trade-off reliability and decoding complexity (e.g., see Theorem~\ref{theo:main}). Recall that both $\xi$, the parameter through which $\ell_n$ is defined in (\ref{eq:ln}), and $\nu$ are positive and fixed (not a function of $n$). Thus, the following lemma ensures that the rate penalty of adding guard bands is negligible as $n \to \infty$.

\begin{lemm}
    \label{lemm:blockLength}
         Let $\vecx$ be a vector of length $|\vecx| = 2^n$. Then,
         \begin{equation}
             \label{eq:vecgLengthBounds}
             |\vecx| \leq |g(\vecx)| < \left( 1 + \frac{2^{-(\xi \cdot  n_0 + 1)}}{1-2^{-\xi}} \right) \cdot |\vecx| \; .
         \end{equation}
\end{lemm}

\begin{IEEEproof}
    From the definition of $g(\vecx)$, induction shows
    \begin{equation}
        \label{eq:vecgLengthRecursion}
        |g(\vecx)| =
        \begin{cases}
            2^n & \mbox{if $n \leq  n_0$} \\
            2^{n} + \sum_{t = n_0 + 1}^{n} 2^{n - t} \cdot \ell_t & \mbox{otherwise}.  
        \end{cases}
    \end{equation}
    Thus, the lower bound in (\ref{eq:vecgLengthBounds}) is trivial, since $|\vecx| = 2^n$, and every term in the sum in (\ref{eq:vecgLengthRecursion}) is non-negative, by (\ref{eq:ln}). The upper bound in (\ref{eq:vecgLengthBounds}) is trivially true for $n \leq n_0$. For the case $n > n_0$, we have that
    \begin{IEEEeqnarray*}{rCl}
        |g(\vecx)|/|\vecx| & \eqann{a} & 1 + \sum_{t = n_0 + 1}^{n} 2^{-t} \cdot \ell_t \\
                           & \eqann[\leq]{b}  & 1 + \sum_{t = n_0 + 1}^{n} 2^{-t} \cdot 2^{(1-\xi) \cdot (t-1)} \\
                           & =  & 1 + \sum_{t = n_0 + 1}^{n}  2^{-\xi \cdot (t-1)-1} \\
                           & <  & 1 + \sum_{t = n_0 + 1}^{\infty}  2^{-\xi \cdot (t-1)-1} \\
                           & \eqann{c} &  1 + \frac{2^{-(\xi \cdot  n_0 + 1)}}{1-2^{-\xi}} \; .
    \end{IEEEeqnarray*}
    where \eqannref{a} follows from $|\vecx| = 2^n$ and (\ref{eq:vecgLengthRecursion}); \eqannref{b} follows from (\ref{eq:ln}); \eqannref{c} is simply the sum of geometric series.
\end{IEEEproof}

A key idea enabling the `genie' described earlier is the recursive processing of each half of the received sequence. This processing will be successful if the middle symbol of the received sequence is a `$0$' originating from the outermost guard band, as per the recursive definition in (\ref{eq:guardBand}). The following lemma shows that this is indeed the case, with very high probability. 
\begin{lemm}
    \label{lemm:guardBandErrorBound}
    Let the guard-band length $\ell_n$ in (\ref{eq:ln}) use a fixed $\xi \in (0,1/2)$.
    Fix the channel deletion probability $\delta$ and a regular hidden-Markov input distribution. Let $n > n_0 > 1$ and let $\vecX$ be a random vector of length $N = 2^n$ distributed according to the modified input distribution described above: i.i.d.\ blocks of length $N_0 = 2^{n_0}$, each distributed according to the specified input distribution. Denote by $\vecY$ the result of transmitting $g(\vecX)$ through the deletion channel. Then, there exists a constant $\theta > 0$, dependent only on the input distribution and the deletion probability such that, for $n_0$ large enough, the probability that the middle symbol of $\vecY^\DPT$ (rounding down) is not a `$0$' from the outer guard band of length $\ell_n$ is at most
    $
        2^{-\theta \cdot 2^{(1-2\xi) n_0}}.
    $
\end{lemm}

\begin{IEEEproof}
    Let $\vecG = g(\vecX)$ (see Fig.~\ref{fig:XGYZ}). Recall that we denote the first and second halves of $\vecX$ by $\XI$ and $\XII$, respectively. Let $\GI = g(\XI)$ and $\GII = g(\XII)$, and denote by $\Gstar$ the guard band comprised of $\ell_n$ `$0$' symbols between $\GI$ and $\GII$.  Hence, by (\ref{eq:guardBand}),
    \[
        \vecG = \GI \concat \Gstar \concat \GII \; .
    \]
    Denote by $\vecY$ the (untrimmed) result of passing $\vecG$ through the deletion channel. Let  $\YI$, $\YII$, and $\Ystar$ be the parts of $\vecY$ corresponding to $\GI$, $\GII$, and $\Gstar$, respectively. Let $\vecZ = \vecY^\DPT$ be the trimmed $\vecY$. Define $\ZI$, $\ZII$, and $\Zstar$, as the parts of $\vecZ$ corresponding to $\GI$, $\GII$, and $\Gstar$, respectively.
    
    For $\vecZ = (Z_1,Z_2,\ldots,Z_t)$ with $t\geq 1$, the middle index of $\vecZ$ (rounding down) is $s=\floor{(t+1)/2}$.
    A sufficient condition for $Z_s$ belonging to $\Zstar$ is
        \begin{equation}
            \label{eq:threePartsCondition}
             |\ZI| < |\Zstar|+|\ZII| \; , \quad |\ZII| < |\ZI| + |\Zstar| \; .
         \end{equation}
    To see that this is sufficient, we observe that $|\ZI| < |\Zstar|+|\ZII|$ implies that the middle index does not fall in $\ZI$ because then
    \begin{align*}
      \floor{(|\vecZ|+1)/2} &= \floor{(|\ZI|+ |\Zstar|+|\ZII|+1)/2} \\ &\geq \floor{(|\ZI|+ |\ZI|+2)/2} = |\ZI| + 1.
    \end{align*}
    Similarly, if $|\ZII| < |\ZI| + |\Zstar|$, then the middle index does not fall in $\ZII$ because then
        \begin{align*}
            \lfloor (|\vecZ| &+1)/2 \rfloor = \floor{(|\ZI|+ |\Zstar|+|\ZII|+1)/2} \\ &\leq \floor{(|\ZI|+ |\Zstar|+|\ZI| + |\Zstar|)/2} = |\ZI|+ |\Zstar|.
        \end{align*}

Now, we will analyze the probability of (\ref{eq:threePartsCondition}).
     Denote by $\alpha$, $\beta$, and $\gamma$ the following length differences between the three parts of $\vecG$ and the three corresponding parts of $\vecY$,
     \begin{IEEEeqnarray*}{rCl}
         \alpha  & = &|\GI| - |\YI| \; , \\
         \beta  & = & |\Gstar| - |\Ystar| \; , \\  
         \gamma & = & |\GII| - |\YII| \; . 
     \end{IEEEeqnarray*}
     Also, denote by $\alpha'$, $\beta'$, and $\gamma'$ the length differences resulting from trimming,
     \begin{IEEEeqnarray*}{rCl}
         \alpha'  & = &|\YI| - |\ZI| \; , \\
         \beta'  & = & |\Ystar| - |\Zstar| \; ,\\  
         \gamma' & = & |\YII| - |\ZII| \; . 
     \end{IEEEeqnarray*}
     Suppose that the trimming on both sides stopped short of the guard band. In this case, $\beta'=0$. Since $|\GI| = |\GII|$ and $|\Gstar| = \ell_n$, condition (\ref{eq:threePartsCondition}) would reduce to 
     \begin{IEEEeqnarray}{rCl}
         \alpha + \alpha' &<& \gamma + \gamma' + \ell_n - \beta \; , \label{eq:firstAlphaBetaGammaCondition} \\
         \gamma + \gamma' &<& \alpha + \alpha' + \ell_n - \beta \; .  \label{eq:secondAlphaBetaGammaCondition}
              \end{IEEEeqnarray}
              Our aim is to show that, with very high probability, both (\ref{eq:firstAlphaBetaGammaCondition}) and (\ref{eq:secondAlphaBetaGammaCondition}) hold, as well as the assumption leading to their formulation.

              Recall that $\delta$ is the channel deletion probability and let
              \begin{equation}
                  \label{eq:ellHat}
                  \hat{\ell} = \ell_n \cdot (1-\delta)/2 \; .
              \end{equation}
              We define the following `good' events on the random variables $\alpha$, $\alpha'$, $\beta$, $\beta'$, $\gamma$, and $\gamma$:
\begin{IEEEeqnarray}{rCl}
    A\phantom{'}: & \quad & \delta |\GI| - \hat{\ell}/4 < \alpha < \delta |\GI| + \hat{\ell}/4 \label{eq:A} \\
    A': & \quad & 0 \leq \alpha' < \hat{\ell}/4 \label{eq:APrime} \\
    B\phantom{'}: & \quad & 0 \leq \beta < \delta \cdot \ell_n + \hat{\ell} \\
    B': & \quad & \beta' = 0 \label{eq:BPrime}\\
    C\phantom{'}: & \quad & \delta |\GII| - \hat{\ell}/4 < \gamma < \delta |\GII| + \hat{\ell}/4 \label{eq:C}\\
    C': & \quad & 0 \leq \gamma' < \hat{\ell}/4 \label{eq:CPrime}
\end{IEEEeqnarray}

First, we note that the total number of symbols deleted or trimmed from $\GI$ is given by $|\GI|-|\ZI| = \alpha+\alpha'$.
If $A$ and $A'$ hold, then this is bounded by
\begin{IEEEeqnarray}{rCl} \label{eq:alpha_plus_alphap}
    \alpha + \alpha' & < & \delta |\GI| + \hat{\ell}/4 + \hat{\ell}/4 \IEEEnonumber \\
                     & = &  \delta |\GI| + \hat{\ell}/2 \label{eq:alphaPlusAlphaPrimeBound} \;.
\end{IEEEeqnarray}
By (\ref{eq:vecgLengthBounds}), $|\GI| = 2^{n-1} + t$, where $t \geq 0$. We now show that if $A$ and $A'$ hold, then $\alpha+\alpha' < |\GI|$. Indeed, by (\ref{eq:ln}) and (\ref{eq:ellHat}),
\begin{IEEEeqnarray*}{rCl}     \delta |\GI| + \hat{\ell}/2 & < & \delta |\GI| + \hat{\ell} \\
                                                           &=  & \delta (2^{n-1}+t) + 2^{-1} (1-\delta) \floor{2^{(1-\xi)(n-1)}}  \\
                                                           & < & \delta (2^{n-1}+t) + (1-\delta) 2^{n-2} \\
                                                           & = & \delta 2^{n-1} + (1-\delta) 2^{n-2} + \delta t \\
			& <  & 2^{n-1}  + \delta t \\
			& <  & 2^{n-1}  + t = |\GI|\; .
\end{IEEEeqnarray*}
The analogous claim also holds for $C$, $C'$, and $\GII$. Thus, if $A$, $A'$, $C$, and $C'$ hold, then some parts of $\GI$ and $\GII$ must remain in $\ZI$ and $\ZII$ after deletion and trimming.
Hence, the trimming has stopped short of the guard band, which implies $B'$.

If, in addition, $B$ occurs, then both (\ref{eq:firstAlphaBetaGammaCondition}) and (\ref{eq:secondAlphaBetaGammaCondition}) must also hold. To verify that (\ref{eq:firstAlphaBetaGammaCondition}) holds, note that
\begin{IEEEeqnarray*}{rCl}
    \gamma + \gamma' + \ell_n - \beta & \eqann[>]{a} & \delta |\GII| - \hat{\ell}/4 + \ell_n - \delta \cdot \ell_n - \hat{\ell} \\
                                      & = & \delta |\GII| - \hat{\ell}/4 + (1-\delta)\ell_n - \hat{\ell}\\
                                      & \eqann{b} & \delta |\GII| - \hat{\ell}/4 + 2\hat{\ell} - \hat{\ell} \\
                                      & = & \delta |\GII| + 3\hat{\ell}/4 \\
                                      & \eqann[>]{c} & \delta |\GII| + \hat{\ell}/2 \; ,
\end{IEEEeqnarray*}
where \eqannref{a} follows from (\ref{eq:BPrime}), (\ref{eq:C}), and (\ref{eq:CPrime}); \eqannref{b} follows from (\ref{eq:ellHat}); \eqannref{c} follows since $\ell_n$ is positive, by (\ref{eq:lnSimpleInequality}), and thus so is $\hat{\ell}$, by (\ref{eq:ellHat}). Next, observe that $|\GI| = |\GII|$, and apply (\ref{eq:alphaPlusAlphaPrimeBound}).
The proof of  (\ref{eq:secondAlphaBetaGammaCondition}) is the same except that the upper and lower bounds are swapped for $\alpha+\alpha'$ and $\gamma+\gamma'$. 

To recap, the occurrence of all the `good' events in (\ref{eq:A})--(\ref{eq:CPrime}) implies that the middle index falls inside $\Zstar$.  Hence, the next step is to show that each of the above events occurs with very high probability, if $n$ is large enough.
               
We now recall Hoeffding's bound \cite[Theorem 2]{Hoeffding:63p}\cite[proof of Lemma 4.13]{MitzenmacherUpfal:17b} and apply it to the deletion channel with deletion probability $\delta$.
Namely, let $D$ be a random variable equal to the number of deletions after $N$ channel uses. Hence, $E[D] = \delta N$, and for $t \geq 0$ we have by Hoeffding's bound that
\begin{IEEEeqnarray}{rCl}
    \Pr ( D \geq \delta N + t ) &\leq& e^{-2t^2/N} \; , \label{eq:HoefPlus} \\
    \Pr ( D \leq \delta N - t ) &\leq& e^{-2t^2/N} \; .  \label{eq:HoefMinus}
\end{IEEEeqnarray}

Recalling that $\xi > 0$, we now require that $n_0$ be large enough that the bracketed term in (\ref{eq:vecgLengthBounds}) is at most $2$. That is, we assume that $n_0$ is large enough such that, for $n > n_0$, we have
\begin{equation}
    \label{eq:GINotBig}
    |\GI| \leq 2 \cdot 2^{n-1} \; .
\end{equation}
Applying both (\ref{eq:HoefPlus}) and (\ref{eq:HoefMinus}), we deduce that, for $n > n_0$, we have
\begin{IEEEeqnarray}{rCl}
   1- \Pr (A) &\leq& 2e^{-2(\hat{\ell}/4)^2/|\GI|} \nonumber \\
         & = & 2e^{-2(\ell_n(1-\delta)/8)^2/|\GI|} \nonumber \\
         & \eqann[<]{a} & 2e^{-2(2^{(n-1)\cdot(1-\xi)-1}(1-\delta)/8)^2/|\GI|} \nonumber \\
         & \eqann[\leq]{b} & 2e^{-2(2^{(n-1)\cdot(1-\xi)-1}(1-\delta)/8)^2/(2 \cdot 2^{n-1})} \nonumber \\
         & = & 2e^{-\left(\frac{(1-\delta)^2}{256}\right) \cdot  2^{(n-1)(1-2\xi)}} \nonumber \\
         & \eqann[\leq]{c} & 2e^{-\left(\frac{(1-\delta)^2}{256}\right) \cdot  2^{n_0 \cdot (1-2\xi)}} \; , \label{eq:PA} 
\end{IEEEeqnarray}
where \eqannref{a} follows from (\ref{eq:lnSimpleInequality}); \eqannref{b} holds by (\ref{eq:GINotBig}); and \eqannref{c} follows from $n > n_0$. Exactly the same bound applies to $1-\Pr (C)$. For $\Pr (B)$, we again use (\ref{eq:HoefPlus}) to deduce that
\begin{IEEEeqnarray}{rCl}
    1-\Pr (B) & \leq & e^{-2 \hat{\ell}^2/\ell_n} \nonumber \\
    & \eqann{a} & e^{-2\left( \frac{\ell_n(1-\delta)}{2} \right)^2/\ell_n} \nonumber \\
    & = & e^{-2 \left( \frac{(1-\delta)}{2} \right)^2 \cdot \ell_n } \nonumber \\
    & \eqann[<]{b} & e^{-2 \left( \frac{(1-\delta)}{2} \right)^2 \cdot 2^{(n-1)(1-\xi) - 1 } } \nonumber \\
    & = & e^{- \left( \frac{(1-\delta)^2}{4} \right) \cdot 2^{(n-1)(1-\xi) }} \nonumber \\
    & \eqann[\leq]{c} & e^{- \left( \frac{(1-\delta)^2}{4} \right) \cdot 2^{n_0 \cdot (1-\xi) }} \; , \label{eq:PB}
\end{IEEEeqnarray}
where \eqannref{a} follows from (\ref{eq:ellHat}); \eqannref{b} follows from (\ref{eq:lnSimpleInequality}); and \eqannref{c} holds because $n > n_0$.


We now bound $1-\Pr (A' \cap C')$ from above. Consider $\GI$ and $\YI$ first. Next, recall that by the recursive definition of $g$ in (\ref{eq:guardBand}), the prefix of length $N_0 = 2^{n_0}$ of $\GI$ is distributed according to the underlying regular Markov input distribution (it does not contain a guard band). Denote this prefix as $X_1,X_2,\ldots,X_{N_0}$, and denote the state of the process at time $0$ as $S_0$. Since our input distribution is not degenerate, there exists an integer $\tau > 0$ and a probability $0 < p < 1$ such that for any $s \in \calS$,
\begin{equation}
    \label{eq:pnonzero}
\Pr \big((X_1,X_2,\ldots,X_\tau) = (0,0,\ldots,0) | S_0 = s\big) < p \; .
\end{equation}
Let
\begin{equation}
    \label{eq:ltilde}
    \tilde{\ell} = \ell_{n_0+1} \cdot (1-\delta)/2 \; .    
\end{equation}
Since $n > n_0$, we have by (\ref{eq:ln}) and (\ref{eq:ellHat}) that $\tilde{\ell} \leq \hat{\ell}$ and that 
\[
   \tilde{\ell}/4 <  2^{n_0} \; .
\]                                                                                                                                                                                    
Let
\[
    \rho =  \tau \cdot \left\lfloor \frac{ \tilde{\ell}/4}{\tau} \right\rfloor \; ,
\]
and partition $X_1,X_2,\ldots X_{\rho}$ into consecutive segments of length $\tau$. Then, we define event $A''$ to occur if there exists a segment that is not an all-zero vector of length $\tau$, and its first non-zero entry has not been deleted. We define $C''$ as the analogous event, with respect to $\GII$ and $\YII$, the only difference being that we are now considering the length $\rho$ suffix of $\XII$, and considering the last non-zero entry of a segment. By construction, if $A''$ and $C''$ hold, then $A'$ and $C'$ must hold. That is, if event $A''$ occurs, then the number of symbols trimmed from the left of $\GI$ is strictly less than $\tilde{\ell}/4$, since the above non-zero non-deleted symbol is not trimmed, and this assures that the ``trimming from the left'' stops before it. A similar claim holds with respect to $C''$. Thus, $1-\Pr (A' \cap C') \leq 1 - \Pr(A'' \cap C'')$.

Since (\ref{eq:pnonzero}) holds for all $s \in \calS$, we have by the Markov property that
\begin{equation}
    \label{eq:PAdoublePrime}
    1-\Pr (A'') < \big( 1 - (1-p)(1-\delta) \big)^{\rho/\tau} \; .
\end{equation}
Indeed, if $A''$ does not hold, this means that we have ``failed'' on each of the $\rho/\tau$ blocks, in the sense that each such block was either all-zero, or its first non-zero symbol was deleted. Since the probability of ``success'' conditioned on any given string of past failures is always greater than $(1-p)(1-\delta)$, the above follows.

Define
\[
    \zeta = -\log_e \big( 1 - (1-p)(1-\delta) \big) \; ,
\]
and note that $\zeta > 0$. Next, we bound $\rho$ as
\begin{IEEEeqnarray*}{rCl}
    \rho & > & \tilde{\ell}/4 - \tau \\
     & = &  \ell_{n_0+1} \cdot (1-\delta)/8 - \tau\\
     & > &  \left( 2^{(1-\xi) \cdot n_0 - 1}\right) \cdot (1-\delta)/8 - \tau ,
\end{IEEEeqnarray*}
where the second inequality follows from (\ref{eq:lnSimpleInequality}). Thus,
\[
    1-\Pr (A'') < e^{- \frac{\zeta}{\tau} \left( \left(2^{(1-\xi) \cdot n_0 - 1}\right) \cdot(1-\delta)/8 - \tau\right)} \; .
\]
Of course, exactly the same bound holds for $1-\Pr (C'')$. Hence, by the union bound, and recalling that $A'' \cap C''$ implies $A' \cap C'$, we have that
\begin{equation}
    \label{eq:PACprime}
    1-\Pr (A' \cap C') < 2 e^{- \frac{\zeta}{\tau} \left( \left(2^{(1-\xi) \cdot n_0 - 1}\right) \cdot(1-\delta)/8 - \tau\right)} \; .
\end{equation}

Putting (\ref{eq:PA}), (\ref{eq:PB}), and (\ref{eq:PACprime}) together, and applying the union bound proves the lemma.
\end{IEEEproof}

We conclude this section with the proof of our main theorem. Note that both the encoding and decoding schemes are specified in the proof.

\begin{IEEEproof}[Proof of Theorem~\ref{theo:main}]
    Our proof is divided into two parts. In the first part, we consider the `idealized' random vectors $\vecX$ and $\vecY$. That is, $\vecX$ is drawn from the probability distribution defined in Lemma~\ref{lemm:guardBandErrorBound} (there is no encoding of date) and $\vecY$ is the result of transmitting $g(\vecX)$ through our deletion channel. We will show that by previously proven lemmas, the rate penalty of expanding $\vecX$ to $g(\vecX)$ is negligible and the probability of deducing $\bigvecY^\DPT$ from $\vecY$ is very high. We conclude the first part by discussing the polarization of $\vecU = \calA(\vecX)$.
    
    In the second part of the proof, we consider the actual case at hand. That is, we show how encoding and decoding are carried out, discuss the encoding and decoding complexity, prove that the rate of our coding scheme approaches the information rate $\calI$, and prove that the probability of misdecoding tends to $0$. 

    Recall that $0 < \nu' < \nu \leq 1/3$ are fixed parameters. We let
    \begin{equation}
        \label{eq:n0}
        n_0 = \floor{\nu n} 
    \end{equation}
    and
    \begin{equation}
        \label{eq:nuDoublPrime}
                   \nu'' = \frac{\nu + \nu'}{2} \; , 
    \end{equation}
    implying that
    \begin{equation}
        \label{eq:trainOfNuInequalites}
        0 < \nu' < \nu'' < \nu \leq \frac{1}{3} \; .
        \end{equation}
    Then, set $\xi$ for the guard-band length $\ell_n$ defined in (\ref{eq:ln}) to
    \begin{equation}
        \label{eq:epsilonFromGamma}
        \xi = \frac{1-\frac{1+\nu'' / \nu}{2}}{2} = \frac{1-\nu''/\nu}{4} \; .
    \end{equation}
    Note that by (\ref{eq:NN0Phi}),
    \begin{equation}
        \label{eq:n1}
    n_1 = n - \floor{\nu n} = \ceiling{(1-\nu) n} \; .
    \end{equation}

    We start with the first part of the proof: let $\vecX$ and $\vecY$ be defined as in Lemma~\ref{lemm:guardBandErrorBound} (as yet, no coding of information).
    \begin{subclaim} \label{subclaim:ratePenalty} The rate penalty incurred by adding guard bands becomes negligible as $n \to \infty$. Namely, $ |g(\vecX)|/|\vecX|$ tends to $1$ as $n \to \infty$.
    \end{subclaim}
    This follows by Lemma~\ref{lemm:blockLength}, which shows that the rate penalty incurred by adding guard bands becomes negligible as $n_0 \to \infty$, and the connection between $n_0$ and $n$ given in (\ref{eq:n0}).

    \begin{subclaim} \label{subclaim:trimming}
        The probability of making a mistake during the partitioning of $\vecY$ into the $\blockCount = 2^{n-n_0}$ trimmed blocks $\vecY(1)^\DPT$, $\vecY(2)^\DPT$,\ldots,$\vecY(\blockCount)^\DPT$ is less than $\frac{1}{3} \cdot 2^{-2^{\nu'' n}}$, for $N = 2^n$ large enough.
\end{subclaim}
This follows from Lemma~\ref{lemm:guardBandErrorBound} and the union bound. Specifically, recalling the recursive nature of our algorithm to produce $\bigvecY^\DPT$, we note that an error is made only if the relevant portion of the received vector $\vecY$, after that portion has been trimmed, is such that the middle symbol (rounding down) does not belong to the outermost guard band. Each such probability can be bounded by using Lemma~\ref{lemm:guardBandErrorBound}. Since we produce $\Phi$ blocks, our recursion is applied $\Phi-1$ times. Hence, for $n_0$ large enough, the probability of failing to produce $\bigvecY^\DPT$ is at most
        \begin{multlinecc}
            \label{eq:genieFailureUpperBound}
            (\Phi-1) \cdot 2^{-\theta \cdot 2^{(1-2\xi) n_0}} \\
            = (2^{n-\floor{\nu n}}-1) \cdot 2^{-\theta \cdot 2^{ \floor{\nu n} \cdot((1+\nu'' / \nu)/2)}}  \; ,
        \end{multlinecc}
        where the equality follows from (\ref{eq:Phi}) and (\ref{eq:n0})--(\ref{eq:n1}). Recalling (\ref{eq:n0}), we may take $n$ large enough such that $n_0$ is indeed large enough for the above to hold. Moreover, since $0 < \nu'' < \nu$, it is straightforward to show that the RHS of (\ref{eq:genieFailureUpperBound}) is less than $\frac{1}{3} \cdot 2^{-2^{\nu'' n}}$ for large enough $n$, as required.

        \begin{subclaim} \label{subclaim:calIInformationBits} For $\vecU = \calA(\vecX)$, the fraction of indices $1 \leq i \leq N$ for which the 
            Bhattacharyya parameter satisfies
            \begin{equation}
                \label{eq:subclaim_zsmall}
                Z(U_i|U_1^{i-1},\vecY(1)^\DPT, \vecY(2)^\DPT,\ldots,\vecY(\blockCount)^\DPT) < \frac{1}{3N} \cdot 2^{-2^{\nu'' n}}
            \end{equation}
            and the total variation parameter (see Definition~\ref{defi:K} in the appendix) satisfies
            \begin{equation}
                \label{eq:subclaim_ksmall}
                K(U_i|U_1^{i-1}) < \frac{1}{3N} \cdot 2^{-2^{\nu'' n}}
            \end{equation}
            tends to $\calI$, as $n \to \infty$. 
        \end{subclaim}
        Informally, $H \approx 0$ iff $Z \approx 0$ and  $H \approx 1$ iff $K \approx 0$. For a formal statement, see e.g.\ \cite[Lemma 1]{Shuval_Tal_Memory_2017}. Thus, Lemma~\ref{lemm:slowPolarizationStar} continues to hold if we replace (\ref{eq:goodIndexI}) by the condition
\begin{multlinecc}
    \label{eq:goodIndexIZK}
    Z(V_{i_0}(\phi)|V_1^{i_0-1}(\phi) ,\vecY^\DPT(\phi)) < \epsilon \quad \mbox{and} \onlysingle{\quad\quad} \\
    K(V_{i_0}(\phi)|V_1^{i_0-1}(\phi)) < \epsilon \; .
\end{multlinecc}
That is, at the end of $n_0$ polarization stages, the fraction of indices $1 \leq i_0 \leq N_0$ satisfying the `weak polarization' in (\ref{eq:goodIndexIZK}) tends to $\calI$ for any $\epsilon > 0$. To get from the `weak polarization' implied by (\ref{eq:goodIndexIZK}) to the `strong polarization' implied by (\ref{eq:subclaim_zsmall}) and (\ref{eq:subclaim_ksmall}), we employ \cite[Lemma 40]{ShuvalTal:18a}, as follows.

For $\vecb = (b_1,b_2,\ldots,b_n)$, recall from (\ref{eq:bitReversedI}) the definition of $i(\vecb)$, and denote
        \[
            i_0(\vecb) \triangleq 1+\sum_{j=1}^{n_0} b_j 2^{n_0-j} \; .
        \]
        Thus, we may think of the random process by which $i(B_1,B_2,\ldots,B_n)$ is chosen as first selecting $i_0$, which is in fact a function of $B_1,B_2,\ldots,B_{n_0}$, and then   completing the choice of $i$ according to a new process $\tilde{B}_1,\tilde{B}_2, \ldots, \tilde{B}_{n_1}$, where
        \begin{equation}
            \label{eq:Btildeprocess}
            \tilde{B}_1 = B_{n_0+1} , \tilde{B}_2 = B_{n_0+2}, \ldots, \tilde{B}_{n_1} = B_{n_0 + n_1} \; ,
        \end{equation}
        recalling that $n_0 + n_1 = n$, by (\ref{eq:n0}) and (\ref{eq:n1}).
        
        Fix $\epsilon > 0$ to a value that will shortly be specified. Next, for now, let us fix an index $i_0$ for which (\ref{eq:goodIndexIZK}) holds. We define two processes related to (\ref{eq:Btildeprocess}), denoted $\tilde{Z}_1,\tilde{Z}_2,\ldots,\tilde{Z}_{n_1}$ and $\tilde{K}_1,\tilde{K}_2,\ldots,\tilde{K}_{n_1}$. Recall that by definition, the $(\vecX(\phi), \vecY(\phi))$ are i.i.d.\ over $1 \leq \phi \leq \Phi$. Hence, this must also be the case for $(V_{i_0}(\phi), V_1^{i_0-1}(\phi),\vecY^\DPT(\phi))$, by (\ref{eq:Vphi}). The first process is the evolution of the conditional Bhattacharyya parameter as we apply the $n_1$ polar transforms implied by (\ref{eq:Btildeprocess}), to $(\tilde{X}_\phi,\tilde{Y}_\phi)_{\phi=1}^\Phi$, where
        \[
            \tilde{X}_\phi = V_{i_0}(\phi) \; \; \mbox{and} \;\; \tilde{Y}_\phi =  (V_1^{i_0-1}(\phi),\vecY^\DPT(\phi)) \; .
        \]
        The second process is defined similarly, but now we consider the evolution of the conditional total variation parameter as we apply $n_1$ polar transforms to $(\tilde{X}_\phi,\dbtilde{Y}_\phi)_{\phi=1}^\Phi$, where
        \[
            \dbtilde{Y}_\phi =  V_1^{i_0-1}(\phi) \; .
        \]
        By our assumption of $i_0$ satisfying (\ref{eq:goodIndexIZK}),
        \[
            \tilde{Z}_1 = Z(\tilde{X}_1|\tilde{Y}_1) < \epsilon \;\; \mbox{and} \;\; \tilde{K}_1 = K(\tilde{X}_1|\dbtilde{Y}_1) < \epsilon \; .
        \]
        Since $(\tilde{X}_\phi,\tilde{Y}_\phi)$ are i.i.d.\ over $\phi$, and the same holds for $(\tilde{X}_\phi,\dbtilde{Y}_\phi)$, we have by \cite[Proposition 5]{Arikan_2009}  that
        \[
            \tilde{Z}_{t+1} \leq
            \begin{cases}
                2\tilde{Z}_t & \mbox{if $\tilde{B}_t = 0$} \\
                \tilde{Z}_t^2 & \mbox{if $\tilde{B}_t = 1$}
            \end{cases}
        \]
        and by \cite[Proposition 4]{Shuval_Tal_Memory_2017} that
        \[
            \tilde{K}_{t+1} \leq
            \begin{cases}
                \tilde{K}_t^2 & \mbox{if $\tilde{B}_t = 0$} \\
                2\tilde{K}_t & \mbox{if $\tilde{B}_t = 1$} .
            \end{cases}
        \]
        Lastly, it follows from (\ref{eq:vecVSlice}) that $\tilde{Z}_{n_1}$ equals the LHS of (\ref{eq:subclaim_zsmall}) while $\tilde{K}_{n_1}$ equals the LHS of (\ref{eq:subclaim_ksmall}), where $i$ is defined in (\ref{eq:bitReversedI}), with $B_j$ instead of $b_j$.

        To prove the sub-claim, we must show that, for every $\xi > 0$, there exists a threshold such that, if $n$ is larger than the threshold, then the fraction of indices $i$ satisfying both (\ref{eq:subclaim_zsmall}) and (\ref{eq:subclaim_ksmall}) is at least $\calI - \xi$. We will do this by choosing an $\epsilon$ and $n_0$ such that the fraction of indices satisfying (\ref{eq:goodIndexIZK}) is at least $\calI - \xi/3$.
        Of these weakly polarized indices, we will choose $n_1$ such that at least a fraction $1- 2\xi/3$ satisfy both (\ref{eq:subclaim_zsmall}) and (\ref{eq:subclaim_ksmall}).
        This is sufficient because $(\mathcal{I} - \xi/3)(1-2\xi/3) \geq \mathcal{I}-\xi$.
        To make a proper argument, however, we will work in reverse.
        
        First, we will set the parameters for strong polarization assuming sufficient weak polarization.
        In particular, we define
                \begin{equation}
                    \label{eq:betafixed}
                    \beta = 3(\nu + \nu'')/4
                \end{equation}
        and observe that (\ref{eq:trainOfNuInequalites}) implies $0 < \beta < 1/2$.
        Then, we let $\psi = \xi/3$ be the maximum fraction of weakly polarized indices that can fail to strongly polarize and apply \cite[Lemma 40]{ShuvalTal:18a} to determine a valid maximum for $\epsilon$ and minimum for $n_1$ (in \cite{ShuvalTal:18a}, $\psi$, $\epsilon$, and $n_1$ are denoted $\delta$, $\eta$, and $n$, respectively).
        This lemma implies the existence of an $\epsilon > 0$ such that if (\ref{eq:goodIndexIZK}) holds for an index $i_0$, then the fraction of $i$ values ($(i_0-1) \cdot \Phi + 1 \leq i \leq i_0 \cdot \Phi$) for which both $\tilde{Z}_i < 2^{-2^{\beta n_1}}$ and $\tilde{K}_i < 2^{- 2^{\beta n_1}}$ is at least $1-2\xi/3$, for all $n_1$ large enough\footnote{Crucially, $\epsilon$ and the $n_1$ threshold do not depend on the choice of $i_0$.}.
        Conceptually, we need to apply the lemma twice -- once for (\ref{eq:subclaim_zsmall}) and once for (\ref{eq:subclaim_ksmall}).        
        Thus, the fraction of weakly polarized indices that fail to satisfy both (\ref{eq:subclaim_zsmall}) and (\ref{eq:subclaim_ksmall}) is at most $2\psi = 2\xi/3$.
        
        Next, for the $\epsilon$ determined above, we find the minimum $n_0$ to guarantee that (\ref{eq:goodIndexIZK}) holds for at least a fraction $\mathcal{I}-\xi/3$ of the $i_0$ indices.       
        Lastly, we recall that $n_0$ and $n_1$ are monotonically increasing functions of $n$, by (\ref{eq:n0}) and (\ref{eq:n1}).
        Hence, for all large enough $n$, the parameters $n_0$ and $n_1$ will exceed the bounds computed earlier and the fraction of indices satisfying (\ref{eq:subclaim_zsmall}) and (\ref{eq:subclaim_ksmall}), where in both cases we replace the RHS by $2^{- 2^{\beta n_1}}$, is at least $\calI - \xi$. 
        
        In order to prove the sub-claim, all that remains  is to show that, for all large enough $n$, we have
        \begin{equation}
            \label{eq:betaAndNuDoublePrime}
            2^{-2^{\beta n_1}} < \frac{1}{3N} \cdot 2^{-2^{\nu'' n}} \; ,
        \end{equation}
        the latter term being RHS of (\ref{eq:subclaim_zsmall}) and (\ref{eq:subclaim_ksmall}). Indeed, by (\ref{eq:n1}) we have that $n_1 \geq (1-\nu)n$, and recalling from (\ref{eq:trainOfNuInequalites}) that $\nu \leq 1/3$, we deduce that $n_1 \geq 2n/3$. Hence, to prove (\ref{eq:betaAndNuDoublePrime}), it suffices to show that
        \begin{equation}
            \label{eq:betaAndNuDoublePrimePlusPlus}
            2^{-2^{2\beta n/3}} < \frac{1}{3N} \cdot 2^{-2^{\nu'' n}} \; .
        \end{equation}
        Indeed, by (\ref{eq:trainOfNuInequalites}) and (\ref{eq:betafixed}) we have that $2 \beta/3 > \nu''$. Thus, recalling that $N=2^n$, we deduce that (\ref{eq:betaAndNuDoublePrimePlusPlus}) holds for all $n$ large enough.

        We now move to the second part of our proof. Let us first discuss how data is encoded. We produce $\vecu = u_1^{N}$ successively, starting from $u_1$ and ending in $u_N$. If the current index $i$ satisfies (\ref{eq:subclaim_zsmall}) and (\ref{eq:subclaim_ksmall}), then $u_i$ is set to an information bit, where the information bits are assumed i.i.d.\ and Bernoulli$(1/2)$. Otherwise, $u_i$ is randomly picked according to the distribution $P(U_i = u_i | U_1^{i-1} = u_1^{i-1})$, where $u_1^{i-1}$ are the realizations occurring in previous stages. The random picks in this case are assumed to be from a random source common to both the encoder and the decoder. Typically, this is implemented using a pseudo-random number generator, common to both sides: if the pseudo-random number $0 \leq r_i \leq 1$ drawn for this stage is such that $P(U_i = 0 | U_1^{i-1} = u_1^{i-1}) \leq r_i$, we set $u_i = 0$. Otherwise, we set $u_i = 1$. These are essentially the `frozen-bits' from the seminal paper \cite{Arikan_2009}.  Transforming $\vecu$ to $\vecx = \calA^{-1}_n(\vecu)$ and adding guard bands to $\vecx$ is as described before.
        
        The following sub-claim proves a key part of our theorem and is an immediate consequence of Subclaims \ref{subclaim:ratePenalty} and \ref{subclaim:calIInformationBits}.
        \begin{subclaim}
            The rate of our coding scheme approaches $\calI$, as $n \to \infty$.
        \end{subclaim}

        Note that the probability distribution of our encoded $\vecu$ does \emph{not} generally equal that of the random variable $\vecU$ used throughout this paper. Namely, denote by $\tilde{p}$ the probability distribution corresponding to the above encoding process: the probability of the encoder producing the vector $\vecu$ is $\tilde{p}(\vecu)$. Next, denote by $p$ the probability distribution of $\vecU$. That is, the probability we would get if we were to set $u_i$ to $0$ with probability $P(U_i = 0 | U_1^{i-1} = u_1^{i-1})$, irrespective of whether $i$ satisfies (\ref{eq:subclaim_zsmall}) and (\ref{eq:subclaim_ksmall}) or not. Our plan is to show that the difference between $p$ and $\tilde{p}$ is `small'. However, we must first address a subtle point stemming from this difference in distributions. Specifically, the probability $P(U_i = 0 | U_1^{i-1} = u_1^{i-1})$ used at stage $i$ might be undefined, since we might be conditioning on an event with probability $0$. In this case, we define the above probability to be $1/2$.

        We decode as previously explained: we first recursively partition the received vector into $\vecy(1)^\DPT, \vecy(2)^\DPT,\ldots,\vecy(\blockCount)^\DPT$. Then, we employ successive cancellation decoding. That is, we produce our estimate $\hat{\vecu} = \hat{u}_1^N$ of $\vecu$ by first producing $\hat{u}_1$, then $\hat{u}_2$, etc., up to $\hat{u}_N$. If index $i$ is such that both (\ref{eq:subclaim_zsmall}) and (\ref{eq:subclaim_ksmall}) hold, then we set $\hat{u}_i$ to the value maximizing 
        \begin{multlinecc}
            \label{eq:decodeNonTrivialI}
            P(U_i = \hat{u}_i |U_1^{i-1}=\hat{u}_1^{i-1},\vecY(1)^\DPT=\vecy(1)^\DPT,\\ \vecY(2)^\DPT=\vecy(2)^\DPT,
            \ldots,\vecY(\blockCount)^\DPT = \vecy(\blockCount)^\DPT) \; .
        \end{multlinecc}
        Otherwise, if $i$ does not satisfy both (\ref{eq:subclaim_zsmall}) and (\ref{eq:subclaim_ksmall}), we set $\hat{u}_i$ is accordance with the common randomness. That is, in the pseudo-random number implementation, we set $\hat{u}_i = 0$ if 
        \begin{equation}
            \label{eq:decodeTrivialI}
            P(U_i = 0 | U_1^{i-1} = \hat{u}_1^{i-1}) \leq r_i \; .
  \end{equation}
Otherwise, we set $\hat{u}_i = 1$.

We stress that the probabilities in (\ref{eq:decodeNonTrivialI}) and (\ref{eq:decodeTrivialI}) are calculated according to the probability distribution of the random vector $\vecU$ used throughout this paper. That is, although $\vecu$ has been encoded according to the probability $\tilde{p}$, we decode it `as if' it had been encoded using $p$. This discrepancy will shortly be addressed. However, as a first step, the following sub-claim considers the case in which there is no discrepancy.

\begin{subclaim}
    If $\vecu$ were chosen according to the probability distribution $p$, then the probability of misdecoding would be less than $\frac{2}{3} \cdot 2^{-2^{\nu'' n}}$, for large enough $n$.
\end{subclaim}
To see this, note that if the above were the case, then $\vecu$ and $\vecU$ would have the same probability distribution. Thus, Subclaim~\ref{subclaim:trimming} would apply, and would imply that the probability of our partitioning algorithm failing to produce the correct $\vecy(1)^\DPT,\vecy(2)^\DPT,\ldots,\vecy(\blockCount)^\DPT$ from the received vector would be less than $\frac{1}{3} \cdot 2^{-2^{\nu'' n}}$, for large enough $n$. Also, if a `genie' were to give us the correct $\vecy(1)^\DPT,\vecy(2)^\DPT,\ldots,\vecy(\blockCount)^\DPT$, we have from (\ref{eq:subclaim_zsmall}) that the probability of misdecoding $\vecu$ would be less than $\frac{1}{3} \cdot 2^{-2^{\nu'' n}}$ for large enough $n$, using exactly\footnote{Since \cite{Arikan_2009} considers the Bhattacharyya parameter for the case of a channel with uniform input, we also need to claim that our $Z$ upper bounds the probability of maximum-aposteriori misdecoding in the more general setting where the channel input is non-uniform. This is well known, see e.g.\ \cite[Remark 1]{Shuval_Tal_Memory_2017} for a proof of a slightly stronger claim.} the same arguments as given in \cite[Proof of Theorem 2]{Arikan_2009} to bound the probability of the successive cancellation decoder failing. The result follows by applying the union bound.

For $\vecu$ such that $p(\vecu) > 0$, denote by $\Pe(\vecu)$ the probability that our decoder fails, given that $\vecu$ was encoded. Otherwise, if $p(\vecu)=0$, define\footnote{Note that we are being conservative. We could have simply defined $\Pe(\vecu)$ as the probability that our decoder fails, given that $\vecu$ was encoded. However, if our input distribution is such that some vectors $\vecu$ are given a probability of $0$, say in order to satisfy a constraint on the input, we should treat the event of the encoder producing a $\vecu$ not satisfying this constraint as an error.} $\Pe(\vecu)=1$. We have just shown that for large enough $n$,
\begin{equation}
    \sum_{\vecu \in \calX^N} p(\vecu) \Pe(\vecu) < \frac{2}{3} \cdot 2^{-2^{\nu'' n}} \; .
\end{equation}
However, recall that our ultimate goal is to upper bound the LHS, after $p(\vecu)$ is replaced by $\tilde{p}(\vecu)$. Informally, a similar bound holds for this case as well, since $p$ and $\tilde{p}$ are `close'. The two following sub-claims makes this statement precise.
\begin{subclaim}
    \[
        \sum_{\vecu \in \calX^N} | \tilde{p}(\vecu) - p(\vecu) | < \frac{1}{3} \cdot 2^{-2^{\nu'' n}}
    \]
\end{subclaim}
To see this, we use the following result from \cite[Lemma 3.5]{Korada:09z}:
    \[
    A_1^N - B_1^N = \sum_{i=1}^N B_1^{i-1} (A_i - B_i) A_{i+1}^N  
\]
where, here, $A_i^j$ denotes the product  $A_i^j= A_i \cdot A_{i+1} \cdots A_j$, and $A_1^0 = A_{N+1}^N \triangleq 1$. We now take
\[
    A_i = A_i(\vecu) =  \tilde{p}(u_i|u_1^{i-1}) \quad \mbox{and} \quad B_i = B_i(\vecu) = p(u_i|u_1^{i-1}) \; .
\]
Recall that we have defined $B_i$ to be $1/2$ if $p(u_1^{i-1}) = 0$. Similarly, we define $A_i$ to be $1/2$ if $\tilde{p}(u_1^{i-1}) = 0$.
We deduce that
\begin{IEEEeqnarray*}{rCl}
    \IEEEeqnarraymulticol{3}{l}{\sum_{\vecu \in \calX^N} |\tilde{p}(\vecu) - p(\vecu) |} \\
\quad & = & \sum_{\vecu \in \calX^N} |A_1^N - B_1^N| \\
      & = & \sum_{\vecu \in \calX^N} \left|\sum_{i=1}^N B_1^{i-1} (A_i - B_i)A_{i+1}^N \right| \\
      & \leq  & \sum_{\vecu \in \calX^N} \sum_{i=1}^N |B_1^{i-1} (A_i - B_i)A_{i+1}^N | \\
      & =  & \sum_{i=1}^N \sum_{\vecu \in \calX^N}  |B_1^{i-1} (A_i - B_i)A_{i+1}^N | \; , \IEEEyesnumber \label{eq:KoradaInternalExternal}
\end{IEEEeqnarray*}
where the first equality follows by the chain rule and the first inequality follows from the triangle inequality. Next, fix $i$, and consider the internal sum in (\ref{eq:KoradaInternalExternal}),
\begin{equation}
\label{eq:KoradaInternal} \sum_{\vecu \in \calX^N}  |B_1^{i-1} (A_i - B_i)A_{i+1}^N | \; .
\end{equation}
If $i$ is an index for which both (\ref{eq:subclaim_zsmall}) and (\ref{eq:subclaim_ksmall}) hold, then $A_i = A_i(\vecu) = 1/2$ for all $\vecu$.  
For this case, we get from (\ref{eq:subclaim_ksmall}) and Lemma~\ref{lemm:KsmallCloseToBerHalf} in Appendix~\ref{sec:ZK} that
\begin{IEEEeqnarray*}{rCl}
    \IEEEeqnarraymulticol{3}{l}{\sum_{\vecu \in \calX^N}  |B_1^{i-1} (A_i - B_i)A_{i+1}^N |} \\
\quad &=& \sum_{\vecu \in \calX^N}  B_1^{i-1} \cdot |A_i - B_i| \cdot A_{i+1}^N \\
      & = & \sum_{\vecu \in \calX^N} p(u_1^{i-1}) \cdot \left|\frac{1}{2} - p(u_i|u_1^{i-1}) \right| \cdot \tilde{p}(u_{i+1}^n | u_1^i)   \\
      & = & \sum_{u_1^i \in \calX^i} p(u_1^{i-1}) \cdot \left|\frac{1}{2} - p(u_i|u_1^{i-1}) \right| \sum_{u_{i+1}^n \in \calX^{N-i}}  \tilde{p}(u_{i+1}^n | u_1^i)   \\
     & = & \sum_{u_1^i \in \calX^i} p(u_1^{i-1}) \cdot \left|\frac{1}{2} - p(u_i|u_1^{i-1}) \right| \\
     & = &  K(U_i|U_1^{i-1}) < \frac{1}{3N} \cdot 2^{-2^{\nu'' n}} \; .
\end{IEEEeqnarray*}
Otherwise, if $i$ is an index for which either (\ref{eq:subclaim_zsmall}) or (\ref{eq:subclaim_ksmall}) do not hold, then $A_i = B_i$ for all $\vecu$, and thus (\ref{eq:KoradaInternal}) equals $0$. The sub-claim follows.

We are now ready to state our bound on the probability of misdecoding.
\begin{subclaim} \label{subclaim:proberrorN} For large enough $n$,
    \[
        \sum_{\vecu \in \calX^N} \tilde{p}(\vecu) \Pe(\vecu) < 2^{-2^{\nu'' n}} \; .
    \]
\end{subclaim}
To show this, we use the two previous sub-claims as follows,
\begin{IEEEeqnarray*}{rCl}
    \sum_{\vecu \in \calX^N}  \tilde{p}(\vecu) \Pe(\vecu) & = & \sum_{\vecu \in \calX^N} (p(\vecu) + \tilde{p}(\vecu) - p(\vecu) ) \Pe(\vecu) \\
     & \leq  & \sum_{\vecu \in \calX^N} (p(\vecu) + |\tilde{p}(\vecu) - p(\vecu)| ) \Pe(\vecu) \\
     & \leq  & \sum_{\vecu \in \calX^N} p(\vecu) \Pe(\vecu) + \sum_{\vecu \in \calX^N} |\tilde{p}(\vecu) - p(\vecu)| \\
                                                           & < &  \frac{2}{3} \cdot 2^{-2^{\nu'' n}} + \frac{1}{3} \cdot 2^{-2^{\nu'' n}} \; ,
\end{IEEEeqnarray*}
which holds for a large enough $n$.

Recall that in the statement of our theorem, we have denoted the length of our codeword (after adding the guard bands) as $\Lambda$. The following subclaim proves another key part of our theorem. 
\begin{subclaim}
    For large enough $n$, the probability of misdecoding is less than $2^{-\Lambda^{\nu' n}}$.
\end{subclaim}
The proof follows by (\ref{eq:nuDoublPrime}), Subclaim~\ref{subclaim:ratePenalty}, and Subclaim~\ref{subclaim:proberrorN}.

All that remains now is to discuss the encoding and decoding complexity of our algorithms.
\begin{subclaim}
    The encoding complexity is $O(\Lambda \log \Lambda)$.
\end{subclaim}
Like the complexity of successive cancellation decoding, the complexity of producing $\vecu$, and from it $\vecx$ is $O(N \log N)$. Adding the guard bands is a simple recursive process whose total time is $O(\Lambda)$. Since $\Lambda \geq N$, the result follows.

\begin{subclaim}
    The decoding complexity is $O(\Lambda^{1+3\nu})$.
\end{subclaim}
The complexity of partitioning the received vector $\vecy$ into the $\Phi$ trimmed blocks $\vecy(1)^\DPT, \vecy(2)^\DPT,\ldots,\vecy(\blockCount)^\DPT$ is $O(\Lambda)$. Next, consider step $i$ of the decoding algorithm, in which we decide on the value of $\hat{u}_i$. The key step is to calculate the probability
\begin{multlinecc*}
    P(U_i = 0| U_1^{i-1} = \hat{u}_1^{i-1},\vecY(1)^\DPT = \vecy(1)^\DPT,\\
    \vecY(2)^\DPT = \vecy(2)^\DPT,\ldots,\vecY(\blockCount)^\DPT = \vecy(\blockCount)^\DPT) \; .
\end{multlinecc*}
This is done in two stages. Recall (\ref{eq:vecVSlice}) and the discussion below it. First, for each $1 \leq \phi \leq \Phi$, we calculate the probabilities
\[                                                                    
    P(V_{i_0}(\phi) = 0| V_1^{i-1}(\phi) = \hat{v}_1^{i-1}(\phi),\vecY(\phi)^\DPT = \vecy(\phi)^\DPT) \; ,
\]
where $i_0$ is the unique integer for which
\[
    (i_0-1) \Phi + 1 \leq i \leq i_0 \Phi 
\]
and $\hat{v}_1^{i_0-1}(\phi)$ is related to $\hat{u}_1^{i-1}$ through (\ref{eq:vecVSlice}). That is, we have just calculated the probabilities corresponding to the first $n_0$ polarization stages. Recall that by Subsection~\ref{subsec:TrellisTDC}, this can be done using $\Phi$ trellises. Next, we apply the remaining $n-n_0$ polarization steps to these probabilities. That is, the standard SC decoder is run for the last $n-n_0$ stages, and can be thought of as effectively operating on a code of length $N_1 = 2^{n_1} = 2^{n-n_0}$.

The total running time of the second stage is well known to be $O(N_1 \log N_1)$, which is indeed $O(\Lambda^{1+3\nu})$. Recalling the discussion in Subsection~\ref{subsec:complexity}, the total running time of the first stage is
\[
    O(\Phi \cdot|\mathcal{S}|^3 N_0^4)  \; ,
\]
where $|\mathcal{S}|$ is the number of states in the Markov chain through which the input distribution is defined (and which we treat as a constant), $N_0 = 2^{n_0} = 2^{\floor{n \nu}}$ and $\Phi = 2^{n - n_0} = 2^{n - \floor{n \nu}}$. Since $N = 2^n \leq \Lambda$, the result follows.

    \end{IEEEproof}

\appendix
\subsection{Conditional Bhattacharyya and Total Variation} \label{sec:ZK}
In this section we define the conditional Bhattacharyya parameter $Z(X|Y)$ and the conditional total variation $K(X|Y)$. See~\cite[Section III]{Shuval_Tal_Memory_2017} for various connections between these and other measures, as well as for their relation to polarization transforms.
\begin{defi}[The conditional Bhattacharyya parameter] \label{defi:Z}
    Let $X \in \calX$ be a binary random variable and $Y \in \calY$ be a discrete random variable. Let their joint distribution be $P_{X,Y}$. We denote
    \begin{IEEEeqnarray*}{rCl}
        Z(X|Y) &=& 2 \sum_{y \in \calY} \sqrt{P_{X,Y}(0,y) \cdot P_{X,Y}(1,y)} \\
               &=& 2 \sum_{y \in \calY} P_Y(y) \sqrt{P_{X|Y}(0|y) \cdot P_{X|Y}(1|y)} \; .
     \end{IEEEeqnarray*}
\end{defi}

\begin{defi}[The conditional total variation] \label{defi:K}
    Let $X \in \calX$ be a binary random variable and $Y \in \calY$ be a discrete random variable. Let their joint distribution be $P_{X,Y}$. We denote
    \begin{IEEEeqnarray*}{rCl}
        K(X|Y) &=& \sum_{y \in \calY} |P_{X,Y}(0,y) - P_{X,Y}(1,y)| \\
               &=& \sum_{y \in \calY} P_Y(y) \cdot |P_{X|Y}(0|y) - P_{X|Y}(1|y)| \; .
     \end{IEEEeqnarray*}
\end{defi}

The following lemma shows that if $K(X|Y)$ is `small', then $P(X|Y)$ is `close' to the Bernoulli$(1/2)$ distribution.
\begin{lemm} \label{lemm:KsmallCloseToBerHalf}
    Let $X \in \calX$ be a binary random variable and $Y \in \calY$ be a discrete random variable. Let their joint distribution be $P_{X,Y}$. Then
    \[
        \sum_{\substack{x \in \calX\\y \in \calY}} P_Y(y) \cdot |P_{X|Y}(x|y) - 1/2| = K(X|Y) \; .
    \]
\end{lemm}
\begin{proof}
    \begin{IEEEeqnarray*}{rCl}
        \IEEEeqnarraymulticol{3}{l}{\sum_{\substack{x \in \calX\\y \in \calY}} P_Y(y) \cdot |P_{X|Y}(x|y) - 1/2|}\\
        &=& \sum_{\substack{y \in \calY}} P_Y(y) \cdot \left(|P_{X|Y}(0|y) - 1/2| + |P_{X|Y}(1|y) - 1/2| \right) \\
        &=& \sum_{\substack{y \in \calY}} P_Y(y) \cdot \left(|P_{X|Y}(0|y) - P_{X|Y}(1|y)| \right) \\
        &=& K(X|Y) \; ,
     \end{IEEEeqnarray*}
     where the penultimate equality is easily seen to hold if we denote $P_{X|Y}(0|y) = 1/2 + \delta(y)$, from which it follows that $P_{X|Y}(1|y) = 1/2 - \delta(y)$.
\end{proof}

\subsection{Capacity-Achieving Inputs for the Deletion Channel}

In~\cite{Dobrushin_1967}, Dobrushin proves a capacity result for a class of synchronization error channels that includes the binary deletion channel.
That paper also shows that the capacity can be approached by a sequence of finite-order Markov input distributions.
Unfortunately, the Markov input distribution in Dobrushin's construction is not irreducible~\cite[Lemma 4]{Dobrushin_1967}.
Thus, Dobrushin's result falls slightly short of what is required by the polar coding construction in this paper.
In~\cite{Li_2019}, Li and Tan study the capacity of the concatenation of a deletion channel and a finite-state channel.
For this setup, they prove a capacity result and show that the capacity can be approached by a sequence of finite-order Markov input distributions that are irreducible and aperiodic.
As they note in their paper, their result is sufficient to prove that the polar coding scheme in this paper can achieve capacity.

In this section, we describe a regular hidden-Markov input distribution that also achieves capacity on the deletion channel.
Though this is not required, given~\cite{Li_2019}, we include it for completeness and because the argument is somewhat different.

Denote by $P_{X^N}$ an input distribution over binary vectors of length $N$, which we will shortly optimize over. Let $\obX \triangleq(X_{1},\ldots,X_{N})$ be a random binary vector of length $N$ drawn according to $P_{X^N}$. Take $\obX$ as the input sequence to a binary deletion channel with deletion probability $\delta\in(0,1)$ and let $\obY \triangleq(Y_{1},\ldots,Y_{M})$ be the corresponding output sequence where the random variable $M$ is the output length. The maximum mutual information for a length-$N$ input is denoted by 
\begin{equation} \label{eq:delcapn}
C_{N}\triangleq\max_{P_{X^{N}}}\frac{1}{N}I(\obX;\obY).
\end{equation}
It is well-known \cite[proof of Theorem II.1]{Kanoria_2013} that $NC_{N}$ is a subadditive sequence and this implies \cite[Lemma 1.2.1, page 3]{Steele:97b} that
\[
C=\lim_{N\to\infty}C_{N}=\inf_{N\geq1}C_{N}
\]
exists and satisfies $C\leq C_{N}$ for $N\geq1$. Thus, for the optimal $P_{X^N}$ we have
\begin{equation}
    \label{eq:IobXobYBoundsC}
 \frac{1}{N}I(\obX;\obY) \geq C \; .
\end{equation}

We begin with the standard approach~\cite{Chen_2008} of using an optimal $P_{X^{N}}$ from~\eqref{eq:delcapn} to generate a length-$kN$ random input $\vecX = \vecX (1) \concat \cdots \concat \vecX (k)$ where each $\vecX(i)$ is a length-$N$ block drawn independently from $P_{X^{N}}$ and using $\concat$ to represent vector concatenation. For this input, we denote the output by $\vecY = \vecY(1) \concat \cdots \concat \vecY(k)$ where $\vecY(i)$ contains the output symbols associated with the input $\vecX(i)$. Thus, for each $i$, the pair $\vecX(i),\vecY(i)$ has the same distribution as the pair $\obX,\obY$.
The random variables $\Ybl_{i} = |\vecY(i)|$, for $i\in[k]$, are chosen to equal the number of output symbols generated by the input block $\vecX(i)$.

Using the chain rule for mutual information, we note that
\begin{align*}
I\big(\vecX;\vecY,\Ybl_1^{k}\big)&=I(\vecX;\vecY)+I\big(\vecX;\Ybl_1^{k}|\vecY\big) \\
&\leq I(\vecX;\vecY) + k\log_{2}(N+1),
\end{align*}
where inequality follows from $I\big(\vecX;\Ybl_1^{k}|\vecY\big) \leq \sum_{i=1}^k H\big(M_{i}\big)$ and $0\leq \Ybl_i \leq N$.
Thus, it follows that
\begin{align*}
I& (\vecX;\vecY) \geq -k\log_{2}(N+1) + I\big(\vecX;\vecY,\Ybl_1^{k}\big)\\
&\stackrel{\!\mathrm{(a)}\!}{=} -k\log_{2}(N+1) + I\big(\vecX;\vecY(1),\ldots,\vecY(k)\big)\\
&= -k\log_{2}(N+1) + \sum_{i=1}^{k}I\big(\vecX;\vecY(i)|\vecY(1),\ldots,\vecY(i-1)\big)\\
&\stackrel{\!\mathrm{(b)}\!}{=} -k\log_{2}(N+1) +\sum_{i=1}^{k}I\big(\vecX(i);\vecY(i)\big)\\
&= -k\log_{2}(N+1)+kI(\obX;\obY)\\
& =kN\left(\frac{1}{N}I(\obX;\obY)-\frac{\log_{2}(N+1)}{N}\right)\\
& \stackrel{\!\mathrm{(c)}\!}{\geq} kN\left(C-\frac{\log_{2}(N+1)}{N}\right),
\end{align*}
where $\mathrm{(a)}$ holds because there is an invertible mapping from $\vecY,\Ybl_1^{k}$ to $\vecY(1),\ldots,\vecY(k)$,  $\mathrm{(b)}$ follows from the pairs $(\vecX(i),\vecY(i))_{i=1}^k$ being i.i.d., and $\mathrm{(c)}$ follows from (\ref{eq:IobXobYBoundsC}).
After normalizing by the input length, this gives
\[ \frac{1}{kN} I(\vecX;\vecY)  \geq C - \frac{\log_2 (N+1)}{N}. \]
Thus, the information rate can be made arbitrarily close to $C$ by choosing $N$ large enough.

However, the infinite input distribution formed by concatenating length-$N$ blocks cannot be generated by a regular hidden-Markov process.
In order to explain how to overcome this, we will first describe this input distribution as a hidden-Markov process with state set
\[ \mathcal{S} \triangleq \bigcup_{j=0}^{N-1} \big\{ x \in \{0,1\}^j \, \big| \, P_{X^j}(x) \neq 0 \big\} , \]
where the set $\{0,1\}^i$ represents all possible states after $i$ input symbols from the length-$N$ input distribution $P_{X^N}$.
We denote the initial state by the empty string $\varepsilon \triangleq \{0,1\}^0$ and let $P_{X^0}(\varepsilon) = 1$ by convention.
To generate multiple blocks, we define the underlying Markov chain to start in the $\varepsilon$ state and return to the $\varepsilon$ state with probability 1 after generating $N$ outputs.
Thus, the underlying Markov chain is irreducible because we have only included states with positive probability and there is a path with positive probability from $\varepsilon$ to any $x \in \mathcal{S}$.

Notice that the state implicitly encodes the current input position in the length-$N$ block distribution.
For example, if $s\in \{0,1\}^j$, then next symbol is drawn according to $P_{X^{j+1}|X^j}(x|s)$.
Thus, the underlying Markov chain is periodic with period $N$.
To make it aperiodic, we will introduce one additional state, which we denote by $\tau$, that is used to dither the input block between length-$N$ and length-$(N+1)$.
State $\tau$ always outputs a dither bit whose value is $0$ and then transitions to state $\varepsilon$.
The idea is that, after a length-$N$ input block, a fair coin is used to determine if the next block will start immediately (e.g., the underlying Markov chain transitions to state $\varepsilon$) or be delayed by one symbol (e.g., the underlying Markov chain transitions to state $\tau$).
After this, the modified Markov chain will be aperiodic because the transition graph has loops of length $N$ and $N+1$.
The period of a Markov chain is the greatest common divisor of the lengths of all loops in the transition graph.
Since $N$ and $N+1$ are relatively prime, the period is 1 and the chain is aperiodic.
We also note that the new Markov chain is still irreducible because there is still a path with positive probability between any two states.

Let $S_0$ be initial state of the underlying Markov chain.
In the current formulation, we have $S_0 = \varepsilon$ with probability 1 and the Markov chain is not stationary.
One can make this Markov chain stationary by drawing the initial state $S_0$ from the stationary distribution of the underlying Markov chain.
After this change, we have constructed a regular hidden-Markov input derived from our original $P_{X^N}$ block distribution.

Now, let $\vecX$ be a length-$k(N+1)$ input drawn from the constructed hidden-Markov process.
This input can be broken into segments by adding commas before the inputs generated by the state $\varepsilon$. A complete segment is delimited by commas on both sides, and thus has length either $N$ or $N+1$. Note that $\vecX$ contains at least $k$ segments, and by discarding the first segment we get at least $k-1$ complete segments.  We call the length-$N$ prefix of a complete segment a block. Thus, we have at least $k-1$ blocks, $\vecX(2),\ldots,\vecX(k)$, where each block can be associated with an independent draw from $P_{X^N}$.
Let $T_i \in \{0,1\}$ be the side-information random variable that indicates, for the $i$-th  (possibly incomplete) segment, whether or not state $\tau$ was visited during that segment.
Given $S_0$ and $T_1^k$, it is always possible to compute the locations of the commas described above and separate $\vecX$ into the $k-1$ blocks $\vecX(2),\ldots,\vecX(k)$.
This is because $S_0$ gives the initial offset into the first segment and $T_i$ indicates whether or not each segment has the additional dither bit.

Similarly, the output $\vecY$ can be separated into subvectors associated with the above blocks by adding commas to separate outputs generated by different segments and removing any outputs caused by dither bits.
Namely, we let $M_i \in \{0,\ldots,N+1\}$ be the side-information random variable that indicates the number of outputs generated by the $i$-th \emph{segment} and $R_i \in \{0,1\}$ be the side-information random variable that indicates whether the last output in a subvector is due to a dither bit.
Given $M_1^k$ and $R_1^k$, it is always possible to separate $\vecY$ into $\vecY(2),\ldots,\vecY(k)$ where each $\vecY(i)$ is the output associated with the block $\vecX(i)$. Thus, each pair $(\vecX(i),\vecY(i))$ has the same distribution as $(\obX,\obY)$.
Using this setup, the chain rule of mutual information and cardinality upper bounds imply that%
\begin{align}
I &\big(\vecX,T_1^k;\vecY,M_1^{k},R_1^{k}|S_0\big) = I \big(\vecX,T_1^k;\vecY,M_1^{k},R_1^{k}\big) \nonumber \\
& \quad + I \big(\vecX,T_1^k;S_0|\vecY,M_1^{k},R_1^{k} \big) - I \big(\vecX,T_1^k;S_0\big) \nonumber \\
&\leq I \big(\vecX,T_1^k;\vecY,M_1^{k},R_1^{k}\big) + I \big(\vecX,T_1^k;S_0|\vecY,M_1^{k},R_1^{k} \big) \nonumber \\
&\stackrel{\!\mathrm{(a)}\!}{\leq} I \big(\vecX,T_1^k;\vecY,M_1^{k},R_1^{k}\big) + N \nonumber \\
&= I \big(\vecX;\vecY,M_1^{k},R_1^{k}\big) + I \big(T_1^{k};\vecY,M_1^{k},R_1^{k} | \vecX\big) + N \nonumber \\
&\stackrel{\!\mathrm{(b)}\!}{\leq} I \big(\vecX;\vecY,M_1^{k},R_1^{k}\big) + k + N \nonumber \\
&= I \big(\vecX;\vecY\big) + I \big(\vecX;M_1^{k},R_1^{k}|\vecY\big) + k + N \nonumber \\
&\stackrel{\!\mathrm{(c)}\!}{\leq} I \big(\vecX;\vecY\big) + k \log_2 (N+2) + k + k + N,\label{eq:capmi1}
\end{align}
where $\mathrm{(a)}$ follows from $\log_2 |\mathcal{S}| = \log_2 \left(1+\sum_{j=0}^{N-1} 2^j \right) = N$, $\mathrm{(b)}$ holds because $T_i \in \{0,1\}$, and $\mathrm{(c)}$ follows from $0\leq M_i \leq N+1$ and $R_i \in \{0,1\}$.

Based on the decompositions described above, the data processing inequality implies that
\begin{align}
I &\big(\vecX,T_1^k;\vecY,M_1^{k},R_1^k|S_0\big) \nonumber \\
&\geq I \big(\vecX(2),\ldots,\vecX(k);\vecY(2),\ldots,\vecY(k)\big | S_0) \nonumber \\
&=  I \big(\vecX(2),\ldots,\vecX(k);\vecY(2),\ldots,\vecY(k)\big) \nonumber \\
&= \sum_{i=2}^{k} I\big(\vecX(i);\vecY(i)\big) = (k-1) I(\obX;\obY).\label{eq:capmi2}
\end{align}
Combining~\eqref{eq:IobXobYBoundsC}--\eqref{eq:capmi2}, we have
\begin{align*}
I& (\vecX;\vecY)
\geq -k\log_{2}(N+2) - 2k - N + (k-1)CN.
\end{align*}
To lower bound the information rate, we can normalize by the input length to see that
\begin{multlinecc*} 
\frac{1}{k(N\!+\!1)} I(\vecX;\vecY) \\
\geq \frac{(k-1)N}{k(N\!+\!1)} \left( C - \frac{1}{k\!-\!1} \right) - \frac{2+ \log_2 (N\!+\!2)}{N\!+\!1}.
\end{multlinecc*}
By choosing $k$ and $N$ large enough, the information rate can be made arbitrarily close to $C$.
Thus, we have constructed a sequence of regular hidden-Markov input distributions that achieve capacity on the binary deletion channel.

In closing, we note that this argument works without change for channels with independent insertions, deletions, and substitutions.

\bibliographystyle{IEEEtran} 
\bibliography{mybib.bib} 

\end{document}